\documentclass[11pt,a4paper]{article}
\usepackage[latin9]{inputenc}
\usepackage{color}
\usepackage{verbatim}
\usepackage{amsmath}

\usepackage{amssymb}
\usepackage{graphicx}
\usepackage{xstring}
\usepackage{authblk}
\usepackage{esint}
\usepackage{amscd}
\usepackage[normalem]{ulem}
\usepackage{hyperref}
\makeatletter
\DeclareGraphicsExtensions{.pdf,.png,.jpg,.eps}

\usepackage{epsfig}
\@ifundefined{showcaptionsetup}{}{%
 \PassOptionsToPackage{caption=false}{subfig}}
\usepackage{subfig}
\makeatother

\usepackage{amsthm}
\usepackage{comment}

\usepackage{graphics}
\usepackage{eepic}
\usepackage{epsfig}
\usepackage{amscd}

\newtheorem{theo}{Theorem}[section]
\newtheorem{theorem}[theo]{Theorem}
\newtheorem{lemma}[theo]{Lemma}

\newtheorem{remark}[theo]{Remark}

\newtheorem{example}[theo]{Example}

\newcommand{\RR}{\mathbb{R}}

\newtheorem{assumption}{Assumption}[section]
\def\gradest{\widehat{\nabla\log{\pi}}}

\newcommand{\ee}{\end{equation}}
\newcommand{\bes}{\begin{equation*}}
\newcommand{\ees}{\end{equation*}}
\newcommand{\bea}{\begin{eqnarray}}
\newcommand{\eea}{\end{eqnarray}}

\newcommand{\R}{\mathbb{R}}

\DeclareMathOperator{\Var}{Var}

\newcommand{\norm}[1]{\left\Vert #1 \right\Vert}
\newcommand{\transpose}{T}

\newcommand{\E}{\mathbb{E}}
\def\IR{\mathbb R}
\def\IN{\mathbb N}

\def\IE{\mathbb E}

\def\IT{\mathbb T}
\global\long\def\poisol{\psi}
\global\long\def\step{k}

\global\long\def\EE{\mathbb{E}}
\global\long\def\EEnoise{\mathbb{E}_{\{\theta_i\}_i}}

\global\long\def\RR{\mathbb{R}}

\global\long\def\UU{\mathcal{U}}

\title{Exploration of the (Non-)asymptotic Bias and Variance of Stochastic Gradient Langevin Dynamics}
\author[1]{Sebastian J. Vollmer\thanks{ vollmer@stats.ox.ac.uk}}
\author[2]{Konstantinos C. Zygalakis \thanks{k.zygalakis@soton.ac.uk}}
\author[1]{Yee Whye Teh\thanks{y.w.teh@stats.ox.ac.uk}}

\affil[1]{Department of Statistics, University of Oxford}
\affil[2]{Mathematical Sciences, University of Southampton.}

\newcommand{\deriv}[2]{{#2}^{(#1)}}
\newcommand{\hess}{\nabla^2}
\newcommand{\derivf}{\nabla}

\begin{document}

\maketitle

Applying standard Markov chain Monte Carlo (MCMC) algorithms to large data sets is computationally infeasible. 
The recently proposed stochastic gradient Langevin dynamics (SGLD) method circumvents this problem in three ways: it generates proposed moves using only a subset of the data, it skips the Metropolis-Hastings accept-reject step, and it uses sequences of decreasing step sizes. In \cite{TehThierryVollmerSGLD2014}, we provided the mathematical foundations for the decreasing step size SGLD, including consistency and a central limit theorem. However, in practice the SGLD is run for a relatively small number of iterations, and its step size is not decreased to zero. The present article investigates the behaviour of the SGLD with fixed step size.  In particular we characterise the asymptotic bias explicitly, along with its dependence on the step size and the variance of the stochastic gradient. On that basis  a modified SGLD which removes the asymptotic bias due to the variance of the stochastic gradients up to first order in the step size is derived. Moreover, we are able to obtain bounds on the finite-time bias, variance and mean squared error (MSE). The theory is illustrated with a Gaussian toy model for which the bias and the MSE for the estimation of moments can be obtained explicitly. For this toy model we study the gain of the SGLD over the standard Euler method in the limit of large data sets.


\section{Introduction}
A standard approach to estimating expectations under a given target density $\pi(\theta)$ is to construct and simulate from Markov chains whose equilibrium distributions are designed to be $\pi$ \cite{MCMC}.  A well-studied approach, for example in molecular dynamics \cite{LM13,BRV10} and throughout Bayesian statistics \cite{MiT07,NealHMC}, is to use Markov chains constructed as numerical schemes which approximate the time dynamics of stochastic differential equations (SDEs). 
In this paper we will focus on the case of first order Langevin dynamics, which has the form
\begin{equation}
d\theta(t)=\frac{1}{2}\nabla\log{\pi(\theta(t))}dt+dW_{t}, \label{eq:Langevin}
\end{equation}
where $t\in \IR_+$, $\theta\in\IR^{d}$ and $W_{t}$ is a $d$-dimensional standard Brownian motion.
Under appropriate assumptions on $\pi(\theta)$,  it is possible to show that the dynamics generated by Equation \eqref{eq:Langevin} are ergodic with respect to $\pi(\theta)$. 

The simplest possible numerical scheme for approximating Equation \eqref{eq:Langevin}  is the Euler-Maruyama method. Let $h>0$ be a step size.  Abusing notation, the diffusion $\theta(k\cdot h)$ at time $k\cdot h$ is approximated by  $\theta_k$, which is obtained using the following recursion equation   
\begin{equation} \label{eq:EM}
\theta_{k+1}=\theta_{k}+\frac{h}{2}\nabla \log \pi(\theta_k)+ \sqrt{h}\xi_{k},
\end{equation}
where $\xi_{k}$ is a standard Gaussian random variable on $\IR^{d}$. One can use the numerical trajectories generated by this scheme
for the construction of an empirical measure $\pi_{h}(\theta)$ either by averaging over one single long trajectory or by averaging over many realisations in order to obtain a finite ensemble average (see for 
example \cite{MiT07}). However, as discussed in \cite{RT96}, one needs to be careful  when doing this as it could be the case that the discrete Markov chain generated by Equation \eqref{eq:EM} is not ergodic. But even if the resulting Markov Chain is ergodic, $\pi_{h}(\theta)$ will not be equal to $\pi(\theta)$ \cite{2010MattinglyPoisson,AVZ14} which thus implies that the resulting sample average is biased.  
An alternative strategy that avoids this discretization bias and the ergodicity of the numerical procedure, is to use Equation \eqref{eq:EM} as a proposal for a Metropolis-Hastings (MH) MCMC algorithm \cite{MCMC11}, with an additional accept-reject step which corrects the discretization error.

In this paper we are interested in situations where $\pi$ arises as the posterior in a Bayesian inference problem with prior density $\pi_{0}(\theta)$ and a large number $N\gg 1$ of i.i.d.\ observations $X_{i}$ with likelihoods $\pi(X_{i}|\theta)$. In this case, we can write
\begin{equation}
\pi(\theta)\propto \pi_0(\theta) \prod_{i=1}^N \pi(X_i|\theta),
\end{equation}
and we have the following gradient,
\begin{equation} \label{eq:model_density}
\nabla\log{\pi(\theta)}=\nabla\log{\pi_{0}(\theta)}+\sum_{i=1}^{N}\nabla\log{\pi(X_{i}|\theta)}.
\end{equation}
In these situations each update \eqref{eq:EM} has an impractically high computational cost of $\mathcal{O}(N)$ since it involves computations on all $N$ items in the dataset.  Likewise, each MH accept-reject step is impractically expensive.


In contrast,  the recently proposed stochastic gradient Langevin dynamics (SGLD) \cite{welling2011bayesian} circumvents this problem by generating proposals which are only based on a subset of the data, by skipping the accept-reject step and by using a decreasing step-size sequence $(h_{k})_{k \geq 0}$. In particular one has
\begin{align}
\theta_{k+1}&=\theta_{k}+\frac{h_{k}}{2}\gradest(\theta_k)+\sqrt{h_{k}}\xi_{k},\label{eq:SGLD}\\
\gradest(\theta_k)&= \nabla\log{\pi_{0}(\theta_k)}+\frac{N}{n}\sum_{i=1}^{n}\nabla\log{\pi(X_{\tau_{ki}}|\theta_k)} \label{eq:unbiased}\end{align}
where $\xi_{k}$ are independent standard Gaussian random variables on $\IR^{d}$,
and $\tau_k=(\tau_{k1},\cdots,\tau_{kn})$ is a random subset of $[N]:=\{1,\cdots,N\}$ of size $n$,
generated, for example, by sampling with or without replacement from
$[N]$. The idea behind this algorithm is that, since the stochastic gradient appearing in Equation \eqref{eq:SGLD} is an unbiased estimator
of the true gradient $\nabla\log{\pi(\theta)}$, the additional perturbation due to the gradient stochasticity is of order $h$, smaller than the $\sqrt{h}$ order of the injected noise, and so the limiting dynamics ($k\rightarrow\infty$)
of Equation \eqref{eq:SGLD} should behave similarly to the case $n=N$.  In \cite{TehThierryVollmerSGLD2014} it was  shown that in this case that the $K$-step  size weighted sample average is consistent and satisfies a CLT with rate depending on the decay of $h_k$. The optimal rate is limited to $K^{-\frac{1}{3}}$ and achieved by an asymptotic step size decay of $\asymp K^{-\frac{1}{3}}$. 

The problem with decaying step sizes is that the efficiency of the algorithm slows the longer it is run for.  A common practice for the SGLD and its extensions, the Stochastic Gradient Hamiltonian Monte Carlo \cite{Chen2014} and the Stochastic Gradient Thermostat Monte Carlo algorithm \cite{HMCT14}, is to use step sizes that are only decreasing  up to a point.  
The primary aim of this paper is to analyse the behaviour of SGLD with fixed step sizes of $h_k=h$.  We provide two complementary analyses in this setting, one asymptotic in nature and one finite time.  Let $\phi:\RR^d\rightarrow\mathbb{R}$ be a test function whose expectation  we are interested in estimating.  Using simulations of the dynamics governed by Equation \eqref{eq:SGLD}, we can estimate the expectation using
\begin{equation}
\EE_\pi[\phi(\theta)] \approx \frac{1}{K}\sum_{k=1}^K \phi(\theta_k) \label{eq:testestimator}
\end{equation}
for some large number of steps $K$.  Our analyses shed light on the behaviour of this estimator.

In the first analysis, we are interested in the asymptotic bias of the estimator \eqref{eq:testestimator} as $K\rightarrow\infty$,
\begin{equation}
\lim_{K\rightarrow\infty} \frac{1}{K}\sum_{k=1}^K \phi(\theta_k) - \EE_\pi[\phi(\theta)].
\end{equation}
Assuming for the moment that the dynamics governed by Equation \eqref{eq:SGLD} is ergodic, with invariant measure $\pi_h(\theta;n)$, the above asymptotic bias simply becomes $\EE_{\pi_h(\cdot;n)}[\phi(\theta)]-\EE_\pi[\phi(\theta)]$.
In the case of Euler-Maruyama, where $n=N$ and the gradient is computed exactly, the asymptotic behaviour of the dynamics is well understood, in particular its asymptotic bias is $\mathcal{O}(h)$ \cite{2010MattinglyPoisson}.
When $n<N$, using the recent generalisations \cite{AVZ14,SN14} of the approach by \cite{TaT90} reviewed in Section \ref{sec:3}, we are able to derive an expansion of the asymptotic bias in powers of the step size $h$.  This allows us to explicitly identify the effect, on the leading order term in the asymptotic bias, of replacing the true gradient  \eqref{eq:model_density} with the unbiased estimator  \eqref{eq:unbiased}.  In particular, we show in Section \ref{sec:Weak-Convergence-Analysis} that, relative to Euler-Maruyama, the leading term contains an additional factor related to the covariance of the subsampled gradient estimators \eqref{eq:unbiased}.

Based on this result, in Section \ref{sec:anal_mSGLD}, we propose a modification of the SGLD (referred to simply as mSGLD) which has the  same asymptotic bias as the Euler-Maruyama method up to first order in $h$.
The mSGLD is given by 
\begin{equation} \label{eq:mSGLD_intro}
\theta_{k+1}=\theta_{k}+\frac{h}{2}\gradest(\theta_k)+\sqrt{h}\left(I-\frac{h}{2}\text{Cov}\left[\gradest(\theta_k)\right]\right)\xi_{j}
\end{equation}
where $\text{Cov}\left[\gradest(\theta_k)\right]$ is the covariance of the gradient estimator.  When the covariance is unknown, it can in turn be estimated by subsampling as well.
This modification is different from the Stochastic Gradient Fisher Scoring \cite{AhnKorWel2012}, a modification of the injected noise in order to better match the Bernstein von Mises posterior. In contrast, the mSGLD is a local modification based on the estimated variance of the stochastic gradient.

The second contribution provides a complementary finite time analysis.  In the finite time case both the bias and the variance of the estimator are non-negligible, and our analysis accounts for both by focussing on bounding the mean squared error (MSE) of the estimator \eqref{eq:testestimator}.  Our results, presented in Section \ref{sec:5}, show that,
\begin{equation}\label{eq:introFiniteTimeBound}
 \E\left[ \left(\frac{1}{K}\sum_{k=0}^{K-1}\phi(\theta_{k})-\pi(\phi)\right)^2\right]\leq C(n)\left(h^2+\frac{1}{K h}\right),
\end{equation}
where the RHS only depends on $n$ through the constant $C(n)$, the $h^2$ term is a contribution of the (square of the) bias while the $1/Kh$ term is a contribution of the variance.  We see that there is a bias-variance trade-off, with bias increasing and variance decreasing monotonically with $h$.  Intuitively, with larger $h$ the Markov chain can converge faster (lower variance) with the same number of steps, but this incurs higher discretization error.   Our result is achieved by extending the work of \cite{2010MattinglyPoisson} from $\IT^{d}$ to $\IR^{d}$. The main difficulty in achieving this relates to combining the results of  \cite{Veretennikov2001Poisson}  and  \cite{TehThierryVollmerSGLD2014}, in order to establish the existence of nice, well controlled solutions to the corresponding Poisson equation \cite{2010MattinglyPoisson}.   We can minimise Equation \eqref{eq:introFiniteTimeBound} over $h$, finding that the minimizing $h$ is on the order of $K^{-\frac{1}{3}}$, and yields an MSE of order $K^{-\frac{2}{3}}$. This agrees, surprisingly, with the scaling of $K^{-\frac{1}{3}}$ for the central limit theorem established for the case of decreasing step sizes, for the Euler-Maruyama scheme in  \cite{lamberton2002recursive} and for SGLD in \cite{TehThierryVollmerSGLD2014}. This unexpected result, that the decreasing step size and fixed step size discretisations have, up to a constant, the same efficiency seems to be missing from the literature.

Our theoretical findings are confirmed by numerical simulations. More precisely, we start by studying a one dimensional Gaussian toy model both in terms of the asymptotic bias and the MSE of time averages in Section \ref{sec:GaussianModel}. The simplicity of this model allows us to obtain explicit expressions for these quantities and thus illustrate in a clear way the connection with the theory. More precisely, we confirm that the scaling of the step size and the number of steps for a prescribed MSE obtained from the upper bound in Equation \eqref{eq:introFiniteTimeBound} matches the scaling derived from the analytic expression for the MSE for this toy model.
More importantly, this simplicity allows us to make significant analytic progress in the study of the asymptotic bias and MSE of time averages in the limit of large data sets $N \rightarrow \infty$. In particular, we are able to show that the SGLD reduces the computational complexity by one order of magnitude in $N$. for the estimation of the second moment in comparison with the Euler method if the error is quantified through the MSE. 



In summary, this paper is organised as follows.  We  present our first explorations of the SGLD applied to a one-dimensional Gaussian toy model in Section \ref{sec:GaussianModel}. For this model we obtain an analytic characterisation of its bias and variance. This  serves as intuition and benchmark for the two analyses developed in Sections \ref{sec:3} to \ref{sec:5}. 
In Section \ref{sec:3} we review some known results  about the effect of the numerical discretisation of Equation \eqref{eq:Langevin} in terms of the finite time weak error as well as in terms of the invariant measure approximation. In Section \ref{sec:Weak-Convergence-Analysis} we apply these results to analyse the finite and long time properties of the SGLD, as well as to construct the modified SGLD algorithm which, asymptotically in $h$, behaves exactly as the Euler-Maruyama method ($n=N$) while still sub-sampling the data set at each step. Furthermore, in Section \ref{sec:5} we discuss the properties of  the finite time sample averages, include its MSE. In Section \ref{sec:anal_toy} we revisit the Gaussian toy model to obtain a more precise understanding of the behaviour of  SGLD in a large data and high accuracy regime. This is achieved using analytic expressions of the expectations of the sample average which are  obtained using the supplemented Mathematica\textsuperscript{\textregistered} notebook described in detail in Appendix \ref{app:Toy}. Finally, in Section \ref{sec:6} we demonstrate the observed performance of SGLD for a Bayesian logistic regression model which matches the theory, while, we conclude this paper in Section \ref{sec:7} with a discussion on some possible extensions of this work.

\section{Exploring a one-dimensional Gaussian Toy Model}\label{sec:GaussianModel}
In this section, we develop results for a simple toy model, which will serve as a benchmark for the theory developed in Sections \ref{sec:3} to \ref{sec:5}. In particular, we obtain analytic expressions for the bias and the variance of the sample average of the SGLD, allowing us to characterise its performance in detail. 

We consider a one-dimensional linear Gaussian model, 
\begin{align}
\label{eq:toymodel}
\begin{aligned}
\theta & \sim  \mathcal{N}(0,\sigma_{\theta}^{2}),\\
X_{i} \,|\, \theta & \stackrel{i.i.d.}{\sim} \mathcal{N}(\theta,\sigma_{x}^{2}) &&\text{for $i=1,\ldots,N$.}
\end{aligned}
\end{align}
The posterior is given by 
\begin{equation}
\pi=\mathcal{N}(\mu_{p},\sigma_{p}^{2})=\mathcal{N}\left(\frac{\sum_{i=1}^{N}X_{i}}{\frac{\sigma_{x}^{2}}{\sigma_{\theta}^{2}}+N},\left(\frac{1}{\sigma_{\theta}^{2}}+\frac{N}{\sigma_{x}^{2}}\right)^{-1}\right).\label{eq:SGaussianPosterior}
\end{equation}
For this choice of $\pi$, the Langevin diffusion \eqref{eq:Langevin} becomes,
\begin{equation}
d\theta(t)=-\frac{1}{2}\left(\frac{\theta(t)-\mu_{p}}{\sigma_{p}^{2}}\right)dt+dW_{t},\label{eq:Lan_OU}
\end{equation}
and its numerical discretisation by the SGLD with step size $h$ reads as follows, 
\begin{equation}
\theta_{k+1}=(1-Ah)\theta_{k}+B_{k}h+\sqrt{h}\xi_{k},\label{eq:sgld_OU}
\end{equation}
where $\xi_{k}\stackrel{i.i.d.}{\sim}\mathcal{N}(0,1)$ and 
\begin{eqnarray}
A & = & \frac{1}{2}\left(\frac{1}{\sigma_{\theta}^{2}}+\frac{N}{\sigma_{x}^{2}}\right),\nonumber \\
B_{k} & = & \frac{N}{n}\frac{\sum_{i=1}^{n}X_{\tau_{ki}}}{2\sigma_{x}^{2}},\label{eq:toyBdistribution}
\end{eqnarray}
where $\tau_k=(\tau_{k1},\cdots,\tau_{kn})$ denote a random subset of $[N]=\{1,\cdots,N\}$ 
generated, for example, by sampling with or without replacement from
$[N]$, independently for each $k$.  We note that the updates \eqref{eq:sgld_OU} will be stable only if $0\le 1-Ah<1$, that is, $0<h<1/A$\footnote{Note that the posterior variance is $\asymp 1/A$, so that steps of size $\asymp 1/A$ are $1/\sqrt{A}$ smaller than the width of the posterior.  However the injected noise has variance $\asymp 1/A$ which matches the posterior variance.}  In the following we will also consider parameterising the step size as $h=r/A$ where $0<r<1$.  

We denote $B=(B_k)_{k\ge 0}$.  At the risk of obfuscating the notation, we will denote by $\Var(B)$ the common variance of $B_k$ for all $k$.
For sampling with replacement, we have
\[
\operatorname{Var}(B)=\frac{1}{4\sigma_{x}^{4}}\frac{N}{n}\sum_{j=1}^{N}\left(X_{i}-\frac{1}{N}\sum_{i=1}^{N}X_{i}\right)^{2}
= \frac{1}{4\sigma_{x}^{4}}\frac{N(N-1)}{n}\text{Var}(X),
\]
where $\text{Var}(X)$ is the typical unbiased empirical estimate of the variance of $\{X_1,\ldots,X_N\}$.  For sampling without replacement we have,
\begin{align}
\operatorname{Var}(B)&=\frac{1}{4\sigma_{x}^{4}}\frac{N(N-n)}{n(N-1)}\sum_{i=1}^{N}\left(X_{i}-\frac{1}{N}\sum_{i=1}^{N}X_{i}\right)^{2}
=\frac{1}{4\sigma_{x}^{4}}\frac{N(N-n)}{n}\text{Var}(X). \label{eq:VarBToy}
\end{align}

\subsection{Analysis of the Asymptotic Bias\label{sub:toyBias}}

We  start by inspecting the estimate of the posterior mean. In particular,
using Equation \eqref{eq:sgld_OU} and taking expectations with respect to
$\xi_{k}$, we have 
\begin{equation}
\IE(\theta_{k+1}|B)=(1-Ah)\IE(\theta_{k}|B)+B_{k}h\label{eq:GaussianRecurrence}
\end{equation}
which can be solved in order to obtain 
\[
\IE(\theta_{M}|B)=(1-Ah)^{M}\IE(\theta_{0})+\sum_{k=0}^{M-1}h(1-Ah)^{k}B_{M-k-1}.
\]
If we now take the expectation with respect to the random subsets $B_{k}$, using the fact
that $\IE(B_{k})=\IE(B)$ and take the limit of $M\rightarrow\infty$,
we have 
\[
\IE(\theta_{\infty})=\sum_{k=0}^{\infty}(1-Ah)^{k}h\IE(B)
=h/(1-(1-Ah))\IE(B)=\frac{\IE(B)}{A}=\frac{\sum_{i=1}^{N}X_{i}}{\frac{\sigma_{x}^{2}}{\sigma_{\theta}^{2}}+N}.
\]
We  thus see that the SGLD is capturing the correct limiting mean
of the posterior independently of the choice of the step size $h$.  In other words, for the test function $\phi(\theta)=\theta$, the asymptotic bias is nil.

We  now investigate the behaviour of the limiting variance under
the SGLD. Starting with the law of total variance,
\[
\operatorname{Var}[\theta_{k+1}]=\IE(\operatorname{Var}[\theta_{k+1}\mid B])+\operatorname{Var}(\IE[\theta_{k+1}\mid B]),
\]
a simple calculation now shows that 
\[
\operatorname{Var}[\theta_{k+1}\mid B]=(1-Ah)^{2}\operatorname{Var}[\theta_{k}\mid B]+h
\]
and 
\[
\operatorname{Var}(\IE[\theta_{k+1}\mid B])=(1-Ah)^{2}\operatorname{Var}(\IE[\theta_{k}\mid B])+h^{2}\operatorname{Var}(B_k).
\]
Combining these two results, we see that 
\[
\operatorname{Var}(\theta_{k+1})=(1-Ah)^{2}\operatorname{Var}(\theta_{k})+h+h^{2}\operatorname{Var}(B_k).
\]
If we now take the limit of $k\rightarrow\infty$, we have that 
\begin{equation}\label{eq:asympVarSGLD}
\operatorname{Var}(\theta_{\infty})=\frac{1}{2A-A^{2}h}+\frac{h\operatorname{Var}(B)}{2A-A^{2}h}.
\end{equation}
where $\operatorname{Var}(B)$ is the common value of $\Var(B_k)$ for all $k\ge 0$.  
We note here that in the case of the Euler-Maruyama method (from here on we will simply refer to this as the Euler method) where $n=N$ and $\text{Var}(B)=0$, only the first term remains.  In other words, the first term is an (over-)estimate of the posterior variance $\sigma_p^2 = 1/2A$ obtained by the Euler-Maruyama discretisation at step size $h$. Our result here coincides with \cite{KZ11}. On the other hand, the second term is an additional bias term due to the variability of the stochastic gradients.  Further, using a Taylor expansion in $h$ of the second summand, we see that the SGLD has an excess bias, relative to the Euler method, with first order term equal to 
\begin{equation}
h \frac{\operatorname{Var}(B)}{2A}. \label{eq:toymodelbias}
\end{equation}
Using the fact that $\text{Var}(\theta_\infty) = \IE[\theta_\infty^2]-\IE[\theta_\infty]^2$, and that the asymptotic bias of estimating $\IE[\theta]$ is nil in this simple model, we see that the above gives the asymptotic biases of the Euler method and SGLD in the case of the test function $\phi(\theta) = \theta^2$.

We now consider the modified SGLD given in Equation \eqref{eq:mSGLD_intro} and to be discussed in Section \ref{sec:anal_mSGLD}. 
In this case the numerical discretisation of Equation \eqref{eq:Lan_OU}
becomes 
\begin{equation}
\theta_{k+1}=\theta_{k}-Ah\theta_{k}+B_{k}h+\sqrt{h}\left(1-\frac{h}{2}\operatorname{Var}(B)\right)\xi_{k}\label{eq:msgld_OU}
\end{equation}
A similar calculation as for the SGLD shows that 
\begin{align}
\IE(\theta_{\infty})&=\frac{\sum_{i=1}^{N}X_{i}}{\frac{\sigma_{x}^{2}}{\sigma_{\theta}^{2}}+N} \nonumber \\
\operatorname{Var}(\theta_{\infty})&=\frac{1}{2A-A^{2}h}+\frac{h^{2}\operatorname{Var}^{2}(B)}{4(2A-A^{2}h)}.
\label{eq:asympVarmSGLD}
\end{align}
with the last term being the excess asymptotic bias.  A Taylor expansion of the excess bias term shows that the term of order $h$ vanishes and the leading term has order $h^2$.  Hence, for small $h$, the excess bias is negligible compared to the asymptotic bias of the Euler method, and we can say that, up to first order and in this simple example, the mSGLD has the same asymptotic bias as for the Euler method.  
In Section \ref{sec:anal_mSGLD}, we will show that these results hold more generally.

\begin{figure}[htb]\label{fig:bias}
\centering 
\includegraphics[scale=0.55]{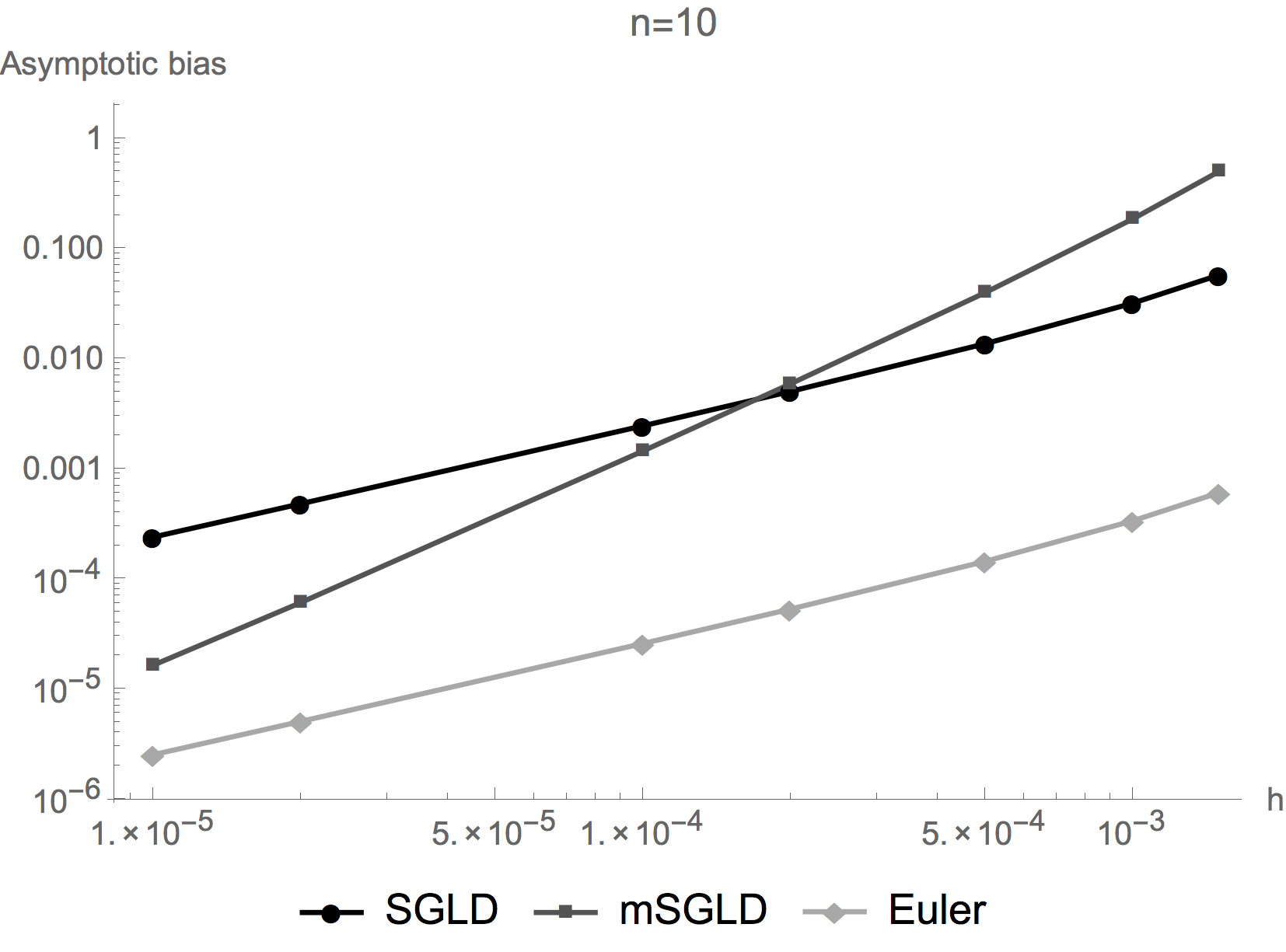} \hskip-2mm \includegraphics[scale=0.55]{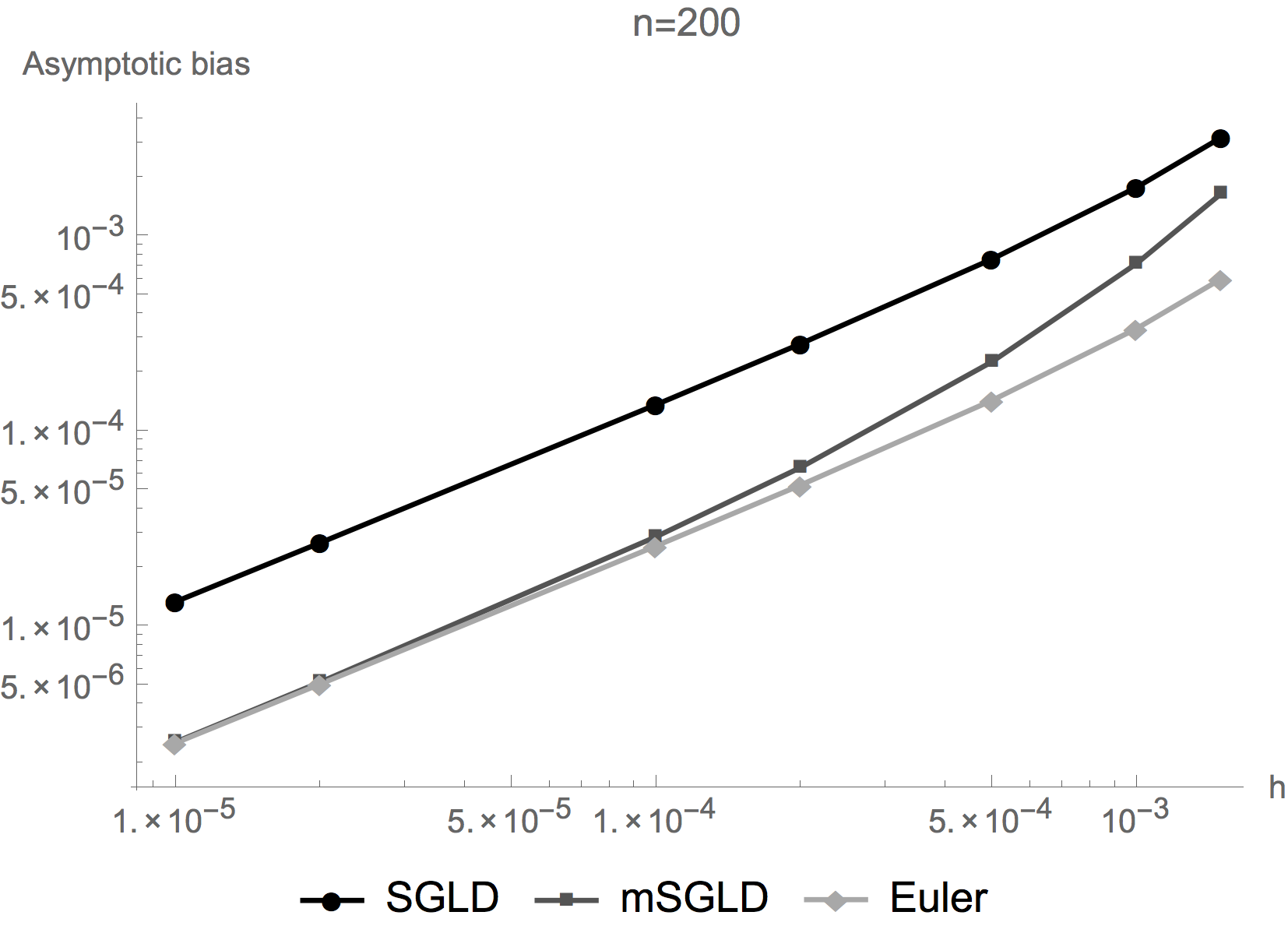}
\caption{Comparison of the asymptotic biases for the SGLD, the mSGLD and the
Euler method for the test function $\phi(\theta)=\theta^2$. For all simulations, we have used $N=10^{3}$. We used $n=10$ on the LHS and 
 $n=200$ on the RHS.}
\label{fig:comparison} 
\end{figure}

It is useful to visualise the above analytic results for the asymptotic biases of  the Euler method, SGLD and mSGLD. In Figure \ref{fig:comparison} we show this for a dataset of $1000$ points drawn according to the model. 
The first observation is that the Euler method has lowest asymptotic bias among all
three methods (although of course it is also the most computationally expensive; see Section \ref{sec:anal_toy}).
We observe that if we choose $n=10$ points for each gradient evaluation, 
for large values of the step size $h$, the SGLD is superior
to the mSGLD. However, as $h$ is reduced, this is no longer the case.
Furthermore, if we use a more accurate gradient estimation with $n=200$ data points, we see that  mSGLD outperforms
 SGLD for all the step sizes used, but more importantly its asymptotic bias
is now directly comparable with the Euler method where
all the data points are used for evaluating the gradient.

\subsection{Finite Time Analysis}\label{sec:toyMSE}

In the previous subsection we analysed the behaviours of the three algorithms in terms of their biases in the asymptotic regime.  In practice, we can only run our algorithms for a finite number of steps, say $K$, and it would be interesting to understand the behaviours of the algorithms in this scenario.  With a finite number of samples, in addition to bias we also have to account for variance due to the Monte Carlo estimation.  

A sensible analysis accounting for both bias and variance is to study the behaviour of the mean squared error (MSE), say in the Monte Carlo estimation for the second moment,
\begin{equation}
\text{MSE}_2:=\IE\left(\frac{1}{K}\sum_{k=0}^{K-1}\theta_{k}^{2}-(\mu_{p}^{2}+\sigma_{p}^{2})\right)^{2}.\label{eq:toyMSEIntro}
\end{equation}
We can expand the quadratic, and express $\text{MSE}_2$ as a linear combination of terms of the form $\IE[\theta_{j}^{p}]$ for $p=1,2,3,4$.  Each of terms can be calculated analytically, depending on the data set $X$, the total number of steps $K$, the subset size $n$, as well as the scaled step size parameter $r=hA$. We provide these calculations in Appendix \ref{sec:appToyFixedModels} and a Mathematica\textsuperscript{\textregistered} file in the supplementary materials.

In Figure~\ref{fig:MSEfixedN} we visualise the behaviour of the resulting $\text{MSE}_2$ for a fixed  dataset with $N=1000$ items, and with scaled step size $r=1/20$.  For the same number of steps $M$, the left figure shows that SGLD and mSGLD behaves similarly, decreasing initially then asymptoting at their asymptotic biases studied in the previous subsection.  At $r=1/20$ mSGLD has lower asymptotic biases than SGLD.  Further, both $\text{MSE}_2$'s decrease with increasing subset size $n$, and are higher than that for the Euler method  at
 $n=1000$.  Since SGLD and mSGLD computational costs per step are linear in $n$, the right figure instead plots the same $\text{MSE}_2$'s against the (effective) number of passes through the dataset, that is, number of steps times $n/N$.  This quantity is now proportional to the computational budget.  Now we see that smaller subset sizes produce initial gains, but asymptote at higher biases.

\begin{figure}
\includegraphics[width=0.5\textwidth]{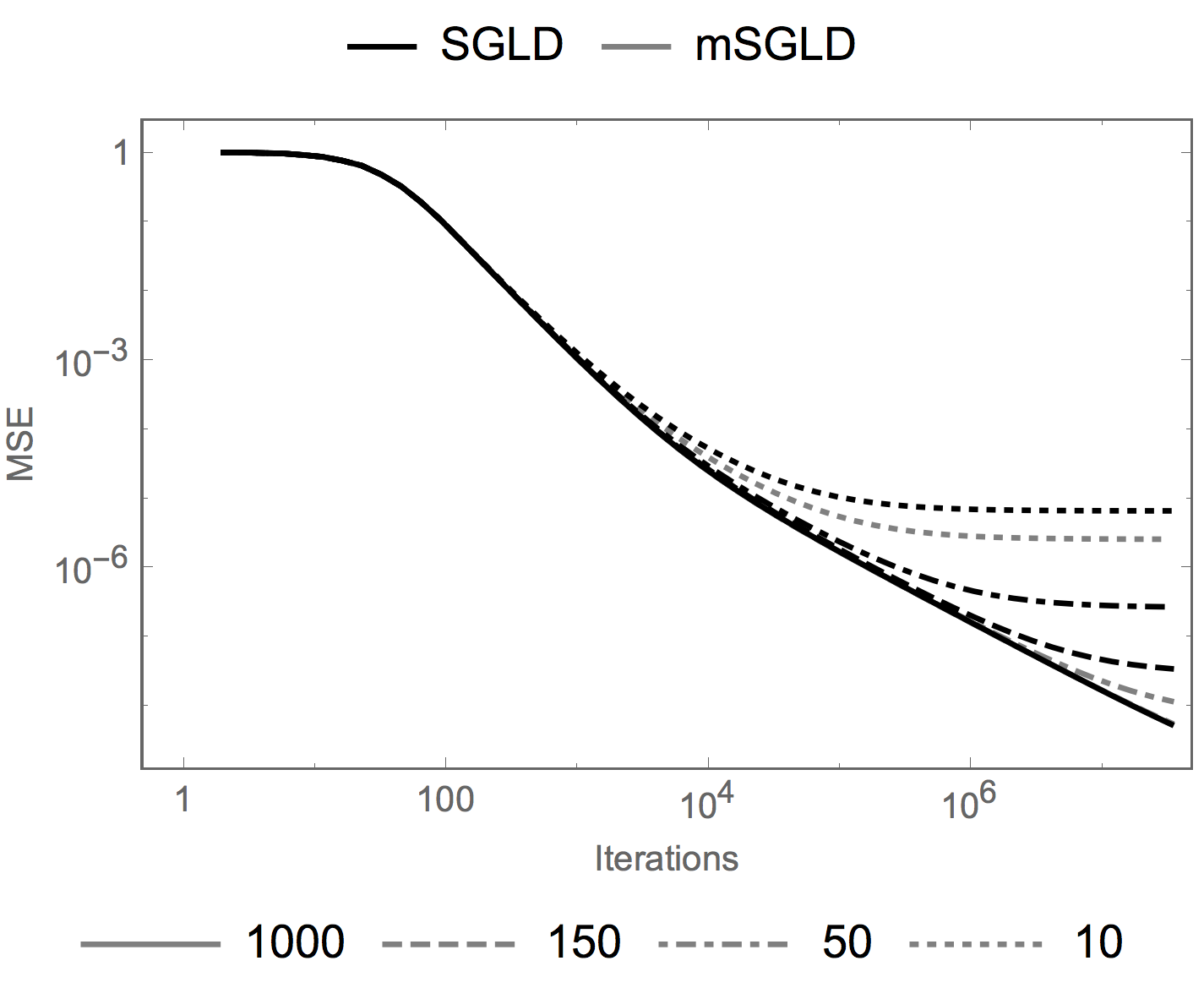}\hspace{0.02\textwidth}
\includegraphics[width=0.5\textwidth]{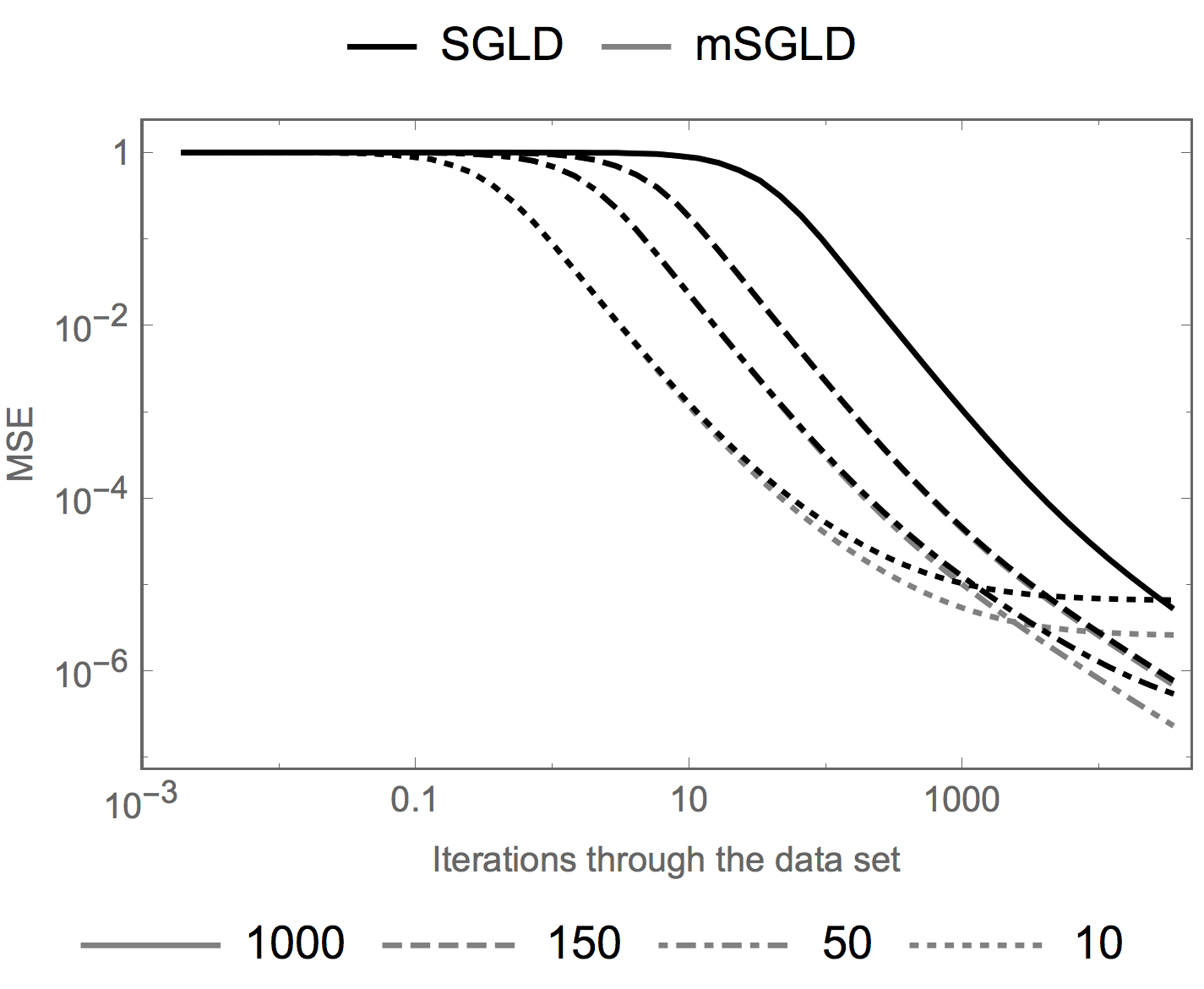}

\caption{ $\text{MSE}_2$ of the sample average for the SGLD and
the mSGLD for the second moment of the posterior.}\label{fig:MSEfixedN}
\end{figure}





\marginpar{Kostas: scrap this paragraph}
These analytical results for a simple Gaussian model demonstrate the more general theory which forms the core contributions of this paper.  Sections \ref{sec:3} and \ref{sec:4} develop a method to study the asymptotic bias as a Taylor expansion in $h$, while Section \ref{sec:5} provides a finite time analysis in terms of the mean squared error.  Both analyses are based on the behaviour of the algorithms for small step sizes, and in this regime we see that mSGLD has better performance than SGLD.  In Section \ref{sec:anal_toy} we will return to the simple Gaussian model to study the behaviour of the algorithms using different measures of performance and in different regimes.  In particular, we will see that for larger step sizes SGLD has better performance than mSGLD.

\section{Review of Weak Order Results}\label{sec:3}

In this section we review some existing results regarding the ergodicity and accuracy of numerical approximations of SDEs.
We start in Section \ref{sec:21} by introducing the framework and notation, the Fokker-Planck and backward Kolmogorov equations, 
and with some preliminary results on local weak errors of numerical one-step integrators.  Section \ref{sec:22} presents assumptions necessary for ergodicity, and extends the results to a global error expansion of the weak error as well as the error in the approximation of the invariant measure. Finally, in Section \ref{sec:33a} we apply our results to explicitly calculate the leading order error term of the numerical approximation of an Ornstein-Uhlenbeck solved by the Euler method.

\subsection{One-step Numerical Approximations of Langevin Diffusions}\label{sec:21} 

Let us denote by $\rho(y,t)$ the probability
density of $\theta(t)$ defined by the Langevin diffusion \eqref{eq:Langevin} with initial
condition $\theta(0)=\theta$ and target density $\pi(y)$. Then $\rho(y,t)$ is the solution of the Fokker-Planck equation,
\begin{align}\label{eq:Fokker} 
\frac{\partial\rho}{\partial t}  &=  \mathcal{L}^{*}\rho,
\end{align}
with initial condition $\rho(y,0)=\delta(y-\theta)$, a Dirac mass for the deterministic initial condition, and the operator $\mathcal{L}^{*}$ given by 
\begin{equation}
\mathcal{L}^{*}\rho=-\frac{1}{2}\nabla_{\theta}\cdot(\nabla\log{\pi(\theta)}\rho)+
\frac{1}{2}\nabla_{\theta}\cdot\nabla_{\theta}\cdot\rho.\label{eq:adjoint}
\end{equation}
This operator is the $L^{2}$-adjoint of the generator of the Markov process $(\theta(t))_{t\ge 0}$ given by \eqref{eq:Langevin},
\begin{equation}
\mathcal{L}=\frac{1}{2}\nabla_{\theta}\log{\pi(\theta)}\cdot\nabla_{\theta}+\frac{1}{2}\Delta_{\theta},\label{eq:generator}
\end{equation}
Given a test function $\phi$, define $u(\theta,t)$ to be the expectation,
\begin{equation}
u(\theta,t)=\mathbb{E}\left(\phi(\theta(t))|\theta(0)=\theta\right),\label{eq:exact_def}
\end{equation}
with respect to the diffusion at time $t$ when started with initial condition $\theta(0)=\theta$.
We note
that $u(\theta,t)$ is the solution of the backward Kolmogorov equation
\begin{align} 
\frac{\partial u}{\partial t} & =  \mathcal{L}u,\label{eq:Kolmogorov} \\
u(\theta,0) & =  \phi(\theta). \nonumber 
\end{align}

A formal Taylor series expansion for $u$  in terms of the
generator $\mathcal{L}$ was derived
in \cite{KZ11} and made rigorous by \cite{DeF12} for the case where the state space is 
$\theta\in\IT^{d}$. The Taylor series is of the following form,
\begin{equation}
u(\theta,h)=\phi(\theta) + \sum_{j=1}^{l}\frac{h^{j}}{j!}\mathcal{L}^{j}\phi(\theta)+h^{l+1}r_{l}(\theta),\label{e:expansion_rigorous}
\end{equation}
for all positive integers $l$,  with the remainder satisfying a bound of the form $|r_{l}(\theta)|\leq c_{l}(1+|\theta|^{\kappa_{l}})$ for some constants $c_{l},\kappa_{l}$  depending on $\pi$ and $\phi$. 

\begin{remark} \label{th:Kolmogorov_assumptions}  Another way to turn
$u(\theta,h)=\phi(\theta)+h\mathcal{L}\phi+\frac{h^{2}}{2}\mathcal{L}^{2}\phi+\cdots$
into a rigorous expansion, see Equation \eqref{e:expansion_rigorous},  is to follow the approach in \cite[Lemma 2]{TaT90} and
to assume that $\log\pi$ is $C^{\infty}$ with bounded derivatives of any
order (and this is the approach we follow here). This fact, together with the assumption that 
\begin{equation}
|\phi(\theta)|\leq C(1+|\theta|^{s})\label{eq:phi_assumption}
\end{equation}
for some positive integer $s$ is enough to prove that the solution
$u$ of Equation \eqref{eq:Kolmogorov} has derivatives of any order
that have a polynomial growth of the form of Equation \eqref{eq:phi_assumption},
with other constants $C,s$ that are independent of $t\in[0,T]$. 
In turn, these regularity bounds establish that  Equation \eqref{e:expansion_rigorous} holds. We mention here that the regularity conditions where relaxed in recent work in \cite{Kopec14a} for the elliptic case and in \cite{Kopec14b} for the hypoelliptic case.
\end{remark} 

Now assume that one solves Equation \eqref{eq:Langevin} numerically with a one step  integrator, which we shall denote by,
\begin{equation} \label{eq:num_general}
\theta_{n+1}=\Psi(\theta_{n},h,\xi_{n}),
\end{equation}
where $\theta_0=\theta(0)$, $h$ denotes the step size, $\xi_n$ are iid $\mathcal{N}(0,1)$, and $\theta_{n}$ denotes the numerical approximation of $\theta(nh)$ for each $n\in\IN$. For example in the case of the Euler method for equation \eqref{eq:Langevin} one has 
\[
\Psi(\theta,h,\xi)=\theta+\frac{h}{2}\nabla \log{\pi(\theta)}+\sqrt{h}\xi
\]
Now, using this formulation we can define
\begin{equation}
U(\theta,h)=\mathbb{E}(\phi(\theta_{1})|\theta_{0}=\theta),\label{eq:numin_def}
\end{equation}
for the expectation of the test function after one step of the numerical integrator starting with the initial condition $\theta_0=\theta$.  We will make the following (easily satisfied) regularity and consistency assumptions about the integrator:

\begin{assumption}\label{ass1} 
We assume that the following hold:
\begin{itemize}
\item
 $\nabla \log\pi$ is $C^{\infty}$ with bounded derivatives of all orders. 
\item
For all deterministic initial conditions $\theta_{0}$,  we have
\begin{equation}
|\IE(\theta_{1}-\theta_{0})|\leq C(1+|\theta_{0}|)h,\text{ and}\qquad|\theta_{1}-\theta_{0}|\leq M(1+|\theta_{0}|)\sqrt{h},\label{ass:Milstein}
\end{equation}
where $C$ is a constant independent of $h$, for $h$ small enough  and  $M$  is a random variable that has bounded moments of all orders independent of $h$ and $\theta_{0}$. 
\item
Equation \eqref{eq:numin_def}
has a weak Taylor series expansion of the form
\begin{equation}
U(\theta,h)=\phi(\theta)+hA_{0}(\pi)\phi(\theta)+h^{2}A_{1}(\pi)\phi(\theta)+\cdots,\label{eq:numin_taylor_expansion_formal}
\end{equation}
where $A_{i}(\pi),~i=0,1,2,\ldots$ are linear differential operators
with coefficients depending smoothly on the drift function $\nabla\log{\pi(\theta)}$
and its derivatives  (depending on the choice of the integrator).
\item
$A_{0}(\pi)$ coincides with the generator $\mathcal{L}$, in other words, the numerical method has weak order at least one.
\end{itemize}
\end{assumption}

Assumptions \ref{ass1} immediately imply the existence of a rigorous expansion 
\begin{equation}
U(\theta,h)=\phi(\theta)+\sum_{i=0}^{l}h^{i+1}A_{i}(\pi)\phi(\theta)+h^{l+2}R_{l}(\theta)\label{eq:numin_taylor_expansion_rigorous}
\end{equation}
for all positive integers $l$, with a remainder satisfying $|R_{l}(\theta)|\leq C_{l}(1+|\theta|^{K_{l}})$ for some constants $C_l, K_l$. We say that the numerical solution has local weak order $p$ if the first $p$ terms in the expansion \eqref{eq:numin_taylor_expansion_formal} of the numerical approximation agrees with that \eqref{e:expansion_rigorous} for the exact diffusion.  In this case, it is easy to see that the following local error formula holds,
\begin{equation}  
\IE(\phi(\theta(h))|\theta(0)=\theta)-\IE(\phi(\theta_{1})|\theta_0=\theta)=h^{p+1}\left(\frac{\mathcal{L}^{p+1}}{(p+1)!}-A_{p}\right)\phi(\theta)+{\cal O}(h^{p+2}).\label{emu:error_cst}
\end{equation}


\subsection{Global Weak Error Expansion}\label{sec:22}

In this subsection, we will extend the local weak error expansion to a global one.  Specifically, after $M$ steps of the numerical integrator with step size $h$, we are interested in the difference between $\theta_M$ and the exact diffusion $\theta(T)$ where $T=Mh$, as evaluated by the difference between the corresponding expectations of $\phi$,

\begin{equation}
E(\phi,h,T)=\IE(\phi(\theta(T))|\theta(0)=\theta)-\IE(\phi(\theta_{M})|\theta_0=\theta),\label{e:weak}
\end{equation}

In order for this study to make sense (when considering the limit $T \rightarrow \infty$), we will require that the SDE and its numerical approximation are both ergodic.  We make standard assumptions in order for the Langevin diffusion $(\theta(t))_{t\ge 0}$ as given by \eqref{eq:Langevin} to be ergodic (see \cite{Has80}):
\begin{assumption} \label{th:ergodic1}
We assume that the following hold for the Langevin diffusion $(\theta(t))_{t\ge 0}$:
\begin{itemize}
\item $\nabla\log{\pi}$ is $C^{\infty}$ with bounded derivatives of all orders.
\item there exists $\beta>0$ and a compact set $K \subset \R^{d}$ such that
$\forall\ \theta\in\IR^{d}\backslash K$,
\[
\left\langle \theta,\nabla\log{\pi(\theta)}\right\rangle \leq-\beta\|\theta\|_2^{2}.
\]
\end{itemize}
\end{assumption}
The question of the ergodicity of the numerical approximation $(\theta_n)$ is considerably more intricate in general.  There exist cases where the underlying Langevin diffusion is ergodic, but its numerical approximation is not ergodic, or does not converge exponentially fast \cite{RT96}. This relates mainly to the properties of the drift coefficient and its behaviour at infinity. For the Euler-Maruyama and the Milstein scheme this has been investigated in \cite{TaT90}.
In what follows we will simply
assume that the Markov chain $(\theta_n)$ defined by the numerical approximation  is indeed
ergodic. 
Under this assumption the following theorem, which combines results derived by \cite{TaT90} and \cite{Mil86}, can be shown (see \cite{AVZ14} for a proof):
\begin{theorem} \label{thm:talay} 
Suppose that the state space is $\R^d$, that Assumptions \ref{ass1} and \ref{th:ergodic1} hold, and that the Markov chain  $(\theta_n)_{n\ge 0}$ defined by the one step integrator \eqref{eq:num_general} is ergodic. 
If the numerical approximation has local weak order $p$, that is, Equation
\eqref{emu:error_cst}  holds, 
then we have the following expansion of the global error \eqref{e:weak},
for all $\phi\in C_{P}^{2p+4}(\R^{d},\R)$, 
\begin{equation}
E(\phi,h,T)=h^{p}\int_{0}^{T}\IE(e(\theta({s}),s))ds+\mathcal{O}(h^{p+1}), \label{eq:main_result1}
\end{equation}
where $e(\theta,t)$ is given by 
\begin{equation}
e(\theta,t)=\left(\frac{1}{(p+1)!}\mathcal{L}^{p+1}-{A}_{p}\right)v(\theta,t),\label{eq:error_coefficient}
\end{equation}
with $v(\theta,t)=\IE(\phi(\theta({T}))|\theta({t})=\theta)$ satisfying
\begin{align}
 \label{eq:Kolomogorov1} 
\frac{\partial v}{\partial t} & =  -\mathcal{L}v,\nonumber \\
v(\theta,T) & =  \phi(\theta).
\end{align}
\end{theorem}

The expression \eqref{eq:main_result1} was proved
by \cite{TaT90} for specific methods (e.g.\ the Euler-Maruyama or
the Milstein methods), while the general procedure to infer the global
weak order from the local weak order is due to \cite{Mil86}
(see also \cite[Chapter 2.2]{MiT04}). However, the  formulation
of the error function \eqref{eq:error_coefficient} here is in terms of the generator $\mathcal{L}$ and the
operators $A_{i}$ in Assumption \ref{ass1}, and does not contain any time derivatives as in \cite{TaT90,Mil86}. 
This formulation will be particularly useful for obtaining our main results.


Using Theorem \ref{thm:talay}, one can obtain a similar
expansion to that in Equation \eqref{eq:main_result1} for the difference between the
true and the numerical ergodic averages: 
\begin{theorem} \label{th:difference} 
Suppose that Assumption \ref{th:ergodic1} holds, that our numerical method with deterministic initial condition is ergodic and of weak order $p$, and
that $\phi:\R^{d}\rightarrow\R$ is a smooth function satisfying Equation \eqref{eq:phi_assumption}.
Then,
\begin{equation}
\lim_{K\rightarrow\infty}\frac{1}{K}\sum_{n=0}^{K-1}\phi(\theta_{n})-\int_{\IR^{d}}\phi(y)\pi(y)dy=-\lambda_{p}h^{p}+\mathcal{O}(h^{p+1})\label{eq:difference}
\end{equation}
where $\lambda_{p}$ is defined as 
\begin{equation}
\lambda_{p}=\int_{\IR^{d}} \int_{0}^{\infty} \left(\frac{1}{(p+1)!}\mathcal{L}^{p+1}-{A}_{p}\right)u(y,t)\pi(y)dt\,dy,\label{eq:lambda}
\end{equation}
and $u(y,t)$ satisfies Equation \eqref{eq:Kolmogorov}. 
\end{theorem}
The proof is given in \cite{AVZ14}, and is similar to that in \cite[Theorem 4]{TaT90}
with the main difference being that Equation \eqref{eq:main_result1} is used as the starting point, instead of the specific formula for the Euler-Maruyama method used in \cite{TaT90}. 

Theorem \ref{th:difference} provides an explicit expression for the leading order term of the asymptotic bias of
the numerical method.  It will thus be the key result in our analysis of the asymptotic behaviour of
SGLD later.  Intuitively Equation \eqref{eq:lambda} says that if we want to calculate the error between the numerical and the true ergodic averages, we need to take into account the long time ($t\rightarrow\infty$) discrepancy between the true and the numerical solution given by
\[
\left(\frac{1}{(p+1)!}\mathcal{L}^{p+1}-{A}_{p}\right)u(y,t),
\] 
and then average over all possible initial conditions $y$ with respect to invariant measure $\pi(y)$.  

\subsection{An Illustrative Example}\label{sec:33a}

We illustrate the weak order results above in the case of the Euler-Maruyama scheme applied to the Ornstein-Uhlenbeck process.  For ease of notation, let $f(x)=\frac{1}{2}\nabla \log{\pi(\theta)} $.  The Euler-Maruyama update steps are,
\begin{equation} \label{eq:thetameth}
\theta_{n+1} = \theta_n +hf(\theta_n) +  \sqrt h \xi_n.
\end{equation}
A straightforward calculation  \cite{KZ11} yields that the differential operator $A_1$ in \eqref{eq:numin_taylor_expansion_formal} is given by
\begin{align}\label{eq:A1}
A_1\psi &= \frac12 f^T \hess \psi f + \frac{1}2 \sum_{i=1}^d \psi'''(e_i,e_i, f)
+ \frac{1}8  \sum_{i,j=1}^d \psi^{(4)} (e_i,e_i,e_j,e_j)
\end{align}
where $e_1,\ldots,e_d$ denotes the canonical basis of $\R^d$ and  $\psi'''(\cdot,\cdot,\cdot)$ and $\psi^{(4)}(\cdot,\cdot,\cdot,\cdot)$, are the derivatives of $\psi$, which are  trilinear and quadrilinear forms, respectively. In dimension $d=1$, it reduces to
\[
A_1\psi = \frac12f^2 \psi'' + \frac{1}2  \psi'''
+ \frac{1}8   \phi^{(4)}  
\]

In the case where $\theta\in\IR$ and $\pi(\theta)=e^{-(\theta-\mu)^{2}/2\sigma^{2}}$, the Langevin diffusion
\eqref{eq:Langevin} corresponds to the one dimensional Ornstein-Uhlenbeck process,
\begin{equation}
d\theta(t)=-\frac{1}{2}\left(\frac{\theta(t)-\mu}{\sigma^{2}} \right)dt+dW_{t}\label{eq:OU1}
\end{equation}
For the test function $\phi(\theta)=\theta^{2}$,  a simple calculation
reveals that the solution of Equation \eqref{eq:Kolmogorov} is
\begin{align}
u(\theta,t)=\sigma^{2}\left(1-e^{-t/\sigma^{2}}\right)+\theta^{2}e^{-t/\sigma^{2}}+\frac{\mu \theta}{\sigma^{2}}(1-e^{-t/2\sigma^{2}})e^{-t/2\sigma^{2}}+\frac{\mu^{2}}{4\sigma^{4}}(1-e^{-t/2\sigma^{2}})^{2}.
\label{eq:ouu}
\end{align}
The Euler-Maruyama scheme has weak order $p=1$, and (see  \cite{KZ11} for more details),
\[
\frac{1}{2}\mathcal{L}^{2}-A_{1}=\frac{1}{8\sigma^{4}}(\theta-\mu)\frac{d}{d\theta}-\frac{1}{4\sigma^{2}}\frac{d^{2}}{d^{2}\theta}.
\]
Using this together with Equation \eqref{eq:ouu} for $u(\theta,t)$, we find 
\[
\left(\frac{1}{2}\mathcal{L}^{2}-A_{1}\right)u(\theta,t)=\frac{e^{-t/\sigma^2}}{2 \sigma^2}-\frac{e^{-t/2 \sigma^2} (\mu-\theta) \left(\left(1-e^{-t/2 \sigma^2}\right) \mu+e^{-t/2 \sigma^2} \theta\right)}{4 \sigma^4}\]
Formula \eqref{eq:lambda} now gives
\begin{align}
\lambda_{1}&=-\int_{0}^{\infty}\int_{-\infty}^{+\infty}\left(\frac{e^{-t/\sigma^2}}{2 \sigma^2}-\frac{e^{-t/2 \sigma^2} (\mu-\theta) \left(\left(1-e^{-t/2 \sigma^2}\right) \mu+e^{-t/2 \sigma^2} \theta\right)}{4 \sigma^4}\right)\frac{e^{-(\theta-\mu)^{2}/2\sigma^{2}}}{\sqrt{2\pi}\sigma}d\theta dt  \nonumber \\
&= \frac{1}{4}.
\end{align}
This is in agreement with known results on the stationary distribution of the Euler-Maruyama approximation to the Ornstein-Uhlenback process, see \cite{KZ11}:
\[
\pi_{h}\sim N(\mu,\sigma_{h}^{2}) \quad\text{ where }\quad\sigma_{h}^{2}=\frac{\sigma^{2}}{1-\frac{h}{4}\sigma^{-2}}=\sigma^{2}+\frac{1}{4}h+\mathcal{O}(h^{2}).
\]


\section{Weak Convergence Analysis\label{sec:Weak-Convergence-Analysis}} \label{sec:4}

We study the weak convergence properties of the SGLD
method in the light of Theorems \ref{thm:talay} and \ref{th:difference}.
The analysis in Section \ref{sec:anal_SGLD} implies that at leading
order there is a cost associated with not calculating the likelihood
over all points. Thus, we introduce in Section \ref{sec:anal_mSGLD} 
a modification of the original algorithm which 
has an error that is, asymptotically in $h$, identical to the error of the Euler method, when all data points are taken into account in the calculation of each likelihood gradient.

\subsection{Stochastic Gradient Langevin Dynamics}

\label{sec:anal_SGLD} Theorems \ref{thm:talay} and \ref{th:difference}
imply that in order to characterise the leading order error term both
for the weak convergence and the invariant measure, we need to calculate
the corresponding differential operators $A_{0},A_{1},\cdots$ in
Equation \eqref{eq:numin_taylor_expansion_formal}. To simplify the presentation
and to illustrate the main ideas, we present the calculations only in
the case where $\theta(t)$ is one dimensional. 
We start our calculations by rewriting the SGLD method in the following
form 
\begin{equation}
\theta_{j+1}=\theta_{j}+h\hat{f_j}(\theta_{j})+\sqrt{h}\xi_{j},\label{eq:SGLD_analysis}
\end{equation}
where 
\[
\hat{f_j}(\theta)=\frac{1}{2}\left(\nabla\log{\pi_{0}(\theta)}+\frac{N}{n}\sum_{i=1}^{n}\nabla\log{\pi(X_{\tau_{ji}}|\theta)}\right),
\]
$\tau_j$ is the subset (possibly with repetition) chosen at step $j$ and,  
\begin{equation}
\IE_{\tau_{j}}\hat{f_j}(\theta)=f(\theta):=\frac{1}{2}\nabla\log{\pi(\theta)},\quad\forall\quad n\leq N.\label{eq:estimator}
\end{equation}
Expanding $\phi(\theta_{j+1})$ in powers
of $h$ and then taking expectations with respect to the injected random noise
$\xi_{j}$,
\begin{align*}
\IE_{\xi_{j}}(\phi(\theta_{j+1})|\theta_j)  =  \phi(\theta_{j}) 
&+h\left(\hat{f}_j(\theta_{j})\phi'(\theta_{j})+\frac{1}{2}\phi''(\theta_{j})\right)\\
 & +  \frac{h^{2}}{2}\left(\hat{f}^{2}_j(\theta_{j})\phi''(\theta_{j})+\hat{f}_j(\theta_{j})\phi^{(3)}(\theta_{j})+\frac{1}{4}\phi^{(4)}(\theta_{j})\right)+\mathcal{O}(h^{3}).
\end{align*}
If we now take expectations with respect to $\tau_j$,
\begin{equation}
\IE(\phi(\theta_{j+1})|\theta_j)=\phi(\theta_{j})+h\mathcal{L}\phi(\theta_{j})+\frac{h^{2}}{2}\left(\IE_{\tau}(\hat{f}^{2}_j(\theta_{j}))\phi''(\theta_{j})+f(\theta_{j})\phi^{(3)}(\theta_{j})+\frac{1}{4}\phi^{(4)}(\theta_{j})\right)+\mathcal{O}(h^{3}),\label{eq:expansion_sgld}
\end{equation}
where $\mathcal{L}$ is the generator \eqref{eq:generator} of Equation \eqref{eq:Langevin}. We
thus see that the SGLD method is a first order weak method and, dropping the indexing by $j$ for notational convenience from now on,
\[
A_{1}(\pi)\phi=\frac{1}{2}\left(\IE_{\tau}(\hat{f}^{2}(\theta))\phi''+f(\theta)\phi^{(3)}+\frac{1}{4}\phi^{(4)}\right).
\]
The asymptotic bias in Equation \eqref{eq:lambda} has an expansion based on the differential operator,
\begin{eqnarray} \label{eq:exp_SGLD}
\frac{1}{2}\mathcal{L}^{2}-A_{1} & = & \frac{1}{2}\left(f(\theta)f'(\theta)+\frac{1}{2}f''(\theta)\right)\frac{d}{d\theta}+\frac{1}{2}\left(f'(\theta)+f^{2}(\theta)-\IE_{\tau}(\hat{f}^{2}(\theta))\right)\frac{d^{2}}{d\theta^{2}} \nonumber\\
 & = & \frac{1}{2}\left(f(\theta)f'(\theta)+\frac{1}{2}f''(\theta)\right)\frac{d}{d\theta}+\frac{1}{2}\left(f'(\theta)-\text{Var}(\hat{f}(\theta))\right)\frac{d^{2}}{d\theta^{2}}
\end{eqnarray}
We thus see that in the case of SGLD the leading order error term contains an extra factor of $-\frac{1}{2}\text{Var}(\hat{f}(\theta))\nabla_{\theta}^{2}$ when compared to the Euler method ($n=N$), in which case $\text{Var}(\hat{f}(\theta))=0$. This can be understood as the penalty associated with not using all the available points for calculating the likelihood at every time step. It results in an extra term in the corresponding error expressions given in Theorems  \ref{thm:talay}
and \ref{th:difference} when compared with the Euler method. More precisely, for $n \ll N$ the term $-\frac{1}{2}\text{Var}(\hat{f}(\theta))$ is of size $\mathcal{O}(N^{2})$ thus making the leading order error term $O(hN^{2})$  in  Equation \eqref{eq:main_result1}. 


\begin{example}
We illustrate the above findings on the toy model discussed in Section \ref{sec:GaussianModel}.
In particular, using the  expression for $u(\theta,t)$ from Section \ref{sec:33a} (replacing $\mu$ and $\sigma$ by $\mu_{p}$ and $\sigma_{p}$ respectively), and that $\operatorname{Var}(\hat{f}(\theta))=\operatorname{Var}(B)$ for this simple model, the extra term in Equation \eqref{eq:exp_SGLD} when compared with the Euler method is now given by 
\[
-\frac{1}{2}\operatorname{Var}(B) \partial^{2}_{\theta} u(\theta,t)=-\operatorname{Var}(B)e^{-t/\sigma_{p}^{2}}.
\]
A simple integration  of this term according to the formula \eqref{eq:lambda} gives that the overall contribution of the extra term, which is,
\[
\sigma_{p} \operatorname{Var}(B)=\frac{\operatorname{Var}(B)}{2A},
\]
and thus agreeing with Equation \eqref{eq:toymodelbias} derived in t Section \ref{sub:toyBias}.
\end{example}

\subsection{Modified SGLD}\label{sec:anal_mSGLD} 

As we have seen in the previous section, the
SGLD method introduces an extra term $-\frac{1}{2}\text{Var}(\hat{f}(\theta))\nabla_{\theta}^{2}$ in the leading
order error term related to the weak error (Theorem \ref{thm:talay}) and to the ergodic averages (Theorem \ref{th:difference}).  When $n\ll N$, this term is of order $\mathcal{O}(hN^{2})$. 
In this section we will explore a modification of SGLD (mSGLD) for which this term is removed, so that the leading order term is exactly the same as for the Euler-Maruyama scheme.  Specifically, the mSGLD updates are,
\begin{equation}
\theta_{j+1}=\theta_{j}+h\hat{f}(\theta_{j})+\sqrt{h}\left(1-\frac{h}{2}\text{Var}\hat{f}(\theta_{j})\right)\xi_{j}\label{eq:SGLD_mod}.
\end{equation}
We can again derive the weak order expansion as in the previous subsection.
Our first step is to expand $\phi(\theta_{j+1})$ in powers of $h$ and then take expectations with respect to the random variable $\xi_{j}$. In particular, we obtain
\begin{align*}
\IE_{\xi_{j}}(\phi(\theta_{j+1}))  = & \phi(\theta_{j})+h\left(\hat{f}_j(\theta_{j})\phi'(\theta_{j})+\frac{1}{2}\phi''(\theta_{j})\right)\\
 & +  \frac{h^{2}}{2}\left(\left[\hat{f}^{2}_j(\theta_{j})-\text{Var}\hat{f}(\theta_{j})\right]\phi''(\theta_{j})+\hat{f}_j(\theta_{j})\phi^{(3)}(\theta_{j})+\frac{1}{4}\phi^{(4)}(\theta_{j})\right)+\mathcal{O}(h^{3}).
\end{align*}
Taking expectations with respect to the random sampling and using  Equation \eqref{eq:estimator}, we obtain
\begin{align} \label{eq:expansion_msgld}
\IE(\phi(\theta_{j+1}))=&\phi(\theta_{j})+h\mathcal{L}\phi(\theta_{j}) \\ 
&+\frac{h^{2}}{2}\left(\left[\IE_{\tau_j}(\hat{f}_j^{2}(\theta_{j}))-\text{Var}\hat{f}(\theta_{j})\right]\phi''(\theta_{j})+f(\theta_{j})\phi^{(3)}(\theta_{j})+\frac{1}{4}\phi^{(4)}(\theta_{j})\right) +\mathcal{O}(h^{3}) \nonumber,
\end{align} 
where $\mathcal{L}$ is the generator of Equation \eqref{eq:Langevin}. We thus see that the mSGLD is a first order weak method and
\[
A_{1}(\pi)\phi=\frac{1}{2}\left(\left[\IE_{\tau}(\hat{f}^{2}(\theta))-\text{Var}\hat{f}(\theta)\right]\phi''+f(\theta)\phi^{(3)}+\frac{1}{4}\phi^{(4)}\right)
\]
Using the expression for $\mathcal{L}^{2}$ as in the case of SGLD, we have that,
\begin{align}
\frac{1}{2}\mathcal{L}^{2}-A_{1}
&=\frac{1}{2}\left(f(\theta)f'(\theta)+\frac{1}{2}f''(\theta)\right)\frac{d}{d\theta}+\frac{1}{2}\left(f'(\theta)+f^{2}(\theta)-\IE_{\tau}(\hat{f}^{2}(\theta))+\text{Var}\hat{f}(\theta) \right)\frac{d^{2}}{d\theta^{2}}
\nonumber \\
&=\frac{1}{2}\left(f(\theta)f'(\theta)+\frac{1}{2}f''(\theta)\right)\frac{d}{d\theta}+\frac{1}{2}f'(\theta)\frac{d^{2}}{d\theta^{2}}\label{eq.mSGLDbias}
\end{align}
We see that the leading order term in the weak error and the
error for the ergodic averages is  the same as for the Euler method, which uses all data at every step. In higher dimensions, a similar calculation gives the mSGLD updates,
\begin{equation}
\theta_{j+1}=\theta_{j}+h\hat{f}_j(\theta_{j})+\sqrt{h}\left(I-\frac{h}{2}\text{Cov}\hat{f}(\theta_{j})\right)\xi_{j}\label{eq:SGLD_mod_high}
\end{equation}
where 
\[
\text{Cov}\hat{f}(\theta)=\IE\left[\left(\hat{f}(\theta)-\IE(\hat{f}(\theta))\right)\left(\hat{f}(\theta)-\IE(\hat{f}(\theta))\right)^{\top}\right]
\]
and $\xi_{j}$ is a $d$-dimensional standard normal random variable.

\begin{remark}\label{rem:varMSGLD}

Except for special cases, $\Var f(\theta_{j})$ does not have a closed form. The simplest possible way to proceed without it is to replace it by an unbiased estimator, for example in case of sampling without replacement,
\[
\widehat{\Var}\widehat{f}_j(\theta):=\frac{N(N-n)}{n(n-1)}\sum_{i=1}^{n}\left(\nabla\log\pi\left(x_{\tau_{ji}}\mid\theta\right)-\frac{\widehat{f}_j(\theta)}{N}\right)^{2}.
\]
This replacement does not change Equation \eqref{eq.mSGLDbias} because
the smallest order contribution to Equation \eqref{eq:expansion_sgld}
is of the form 
\[
-h^{2}\IE\left[\widehat{\Var} f(\theta_j)\xi_{j}^{2}\right]=-h^{2}\Var f(\theta_j).
\]
However, estimating the variance of the stochastic gradient will affect  higher order terms in $h$. For fixed $h$  these terms may have larger contribution to the overall error depending on the choice of $n$ and $N$. In fact, this is true even if  we use the exact variance for the toy model in Section \ref{sub:toyBias}. More precisely, we compare the bias of the mSGLD and the SGLD in Equation \eqref{eq:BiasMSGLDvsSGLD}  notice that $h^2$ term might be larger depending on the choice of $n$ and $N$.
\end{remark}
\section{Finite Time Sample Averages}
\label{sec:5}
Having focused on the SGLD in the asymptotic regime, we will now provide non-asymptotic analysis of the mean squared error (MSE) of the finite time sample averages of the SGLD. In particular, we will decompose the MSE into  bias and variance. The main result of this section will be of the form 
\begin{align}
\begin{aligned}\label{eq:MSEboundIntro}
\text{Bias:}&& \left|\EE \frac{1}{K}\sum^{K-1}_{i=0}\phi(\theta_i) -\int \phi(x)\pi(x)dx\right| &=  O\left(h+\frac{1}{K h}\right)\\
\text{MSE:}&& \EE\left(\frac{1}{K}\sum^{K-1}_{i=0}\phi(\theta_i) -\int \phi(x)\pi(x)dx\right)^{2} & =  O\left(h^{2}+\frac{1}{K h}\right) 
\end{aligned}
\end{align}

\begin{remark}\label{rem:cmpDecStep}
In \cite{TehThierryVollmerSGLD2014}, a central limit theorem was provided for the decreasing step size SGLD which shows a convergence rate of $O(K^{-\frac{1}{3}})$. At first sight, the bound in Equation \eqref{eq:MSEboundIntro} seems better because of the $\frac{1}{Kh}$ term in the upper bound. However, due to the bias, an additional term of order $O(h^2)$ appears. In order to compare  \eqref{eq:MSEboundIntro} with the previous result of  \cite{TehThierryVollmerSGLD2014}, we optimise the sum of both terms over the step size $h$. This results in a bound on the MSE of the SGLD of order  $O(K^{-\frac{2}{3}})$ and agrees with the rate achieved by the decreasing step size SGLD. This agreement between decreasing step size discretisation and fixed step size discretisation is, to our knowledge, not a widely-known observation in the literature.  In contrast, for standard MCMC algorithms the MSE is bounded by order $O(K^{-1})$ due to the Metropolis-Hastings correction that removes the bias.
%
Nevertheless, experimental results in the literature demonstrate that the SGLD might be advantageous in the initial transient phase of learning, see e.g.\cite{PatTeh2013a,Chen2014} \end{remark}

In Section \ref{sec:PoissonMain} we will focus on establishing the bound in Equation \eqref{eq:MSEboundIntro} which is an extension of the work by \cite{2010MattinglyPoisson}. The authors obtained similar results for finite time sample averages of discretisations of diffusions of the form
\begin{equation}
d\theta_{t}=f(\theta_{t})+g(\theta_{t})dW_{t}\label{eq:poissonSDE}
\end{equation}
on the torus which we review subsequently in Section \ref{sec:PoissonPrelim}.

\subsection{Preliminaries on the Poisson Equation and Time Averages \label{sec:PoissonPrelim}}
In the following a connection between  time averages of the diffusion and the corresponding Poisson equation will be presented. For a more elaborate description of this technique
we point the reader to Section 4.2 of \cite{2010MattinglyPoisson}
and references therein.

The Poisson equation is an elliptic PDE on the basis of the  generator associated with Equation \eqref{eq:poissonSDE}. The generator of Equation \eqref{eq:poissonSDE} is
\[
\mathcal{L}\psi=\nabla\psi\cdot\nabla f+\frac{1}{2}g(\theta)^\top\nabla^{2}\psi g(\theta),
\]
while the Poisson equation is given by
\begin{equation}
\mathcal{L}\psi=\phi-\bar{\phi}\text{ on }\RR^{d}\label{eq:poisson}
\end{equation}
where $\phi$ is a test function and $\bar{\phi}:=\int\phi(x)\pi(dx)$ with $\pi$ being the invariant distribution of \eqref{eq:poissonSDE}.  
For applications in Bayesian statistics $\pi$ represents the posterior and the quantity $\bar{\phi}$ the posterior expectation of interest. The posterior expectation $\bar{\phi}$ is estimated by the time average $\frac{1}{t}\int_{0}^{t}\phi(\theta({s}))ds$ of the Langevin dynamics. The difference between the two can be expressed explicitly by using  Itô's formula on the  solution $\psi$ of the Poisson equation
\begin{eqnarray*}
\psi\left(\theta(t)\right)-\psi\left(\theta(0)\right) & = & \int_{0}^{t}\phi(\theta{(s)})-\bar{\phi}ds+\int_{0}^{t}\nabla\psi\left(\theta({s})\right)\cdot g(\theta({s}))dW_{s},\\
\frac{1}{t}\int_{0}^{t}\phi(\theta({s}))ds-\bar{\phi} & = & \frac{1}{t}\left(\psi\left(\theta(t)\right)-\psi\left(\theta(0)\right)\right)-\frac{1}{t}\int_{0}^{t}\nabla\psi\left(\theta({s})\right)\cdot g(\theta({s}))dW_{s}.
\end{eqnarray*}
If the first term and the variance of the second term (the martingale term) on the right hand side can be bounded, an error bound for the time average is obtained. 

In this article, we are interested in the time average of the Euler discretisation
and the SGLD. We can build on the ideas of Section 5 in \cite{2010MattinglyPoisson}
which considers time discretisations of Equation (\ref{eq:poissonSDE}) 
of the following form
\[
\theta_{k+1}=\theta_{k}+hf(\theta_{k},h)+\sqrt{h}g(\theta_{k},h)\eta_{k},\quad\eta_k\sim\mathcal{N}\left(0,I\right).
\]
In \cite{2010MattinglyPoisson} a Taylor expansion is used to express
\[
\Delta\psi(\theta_{k+1}):=\psi(\theta_{k+1})-\psi(\theta_{k})=h\left(A_{0}\psi\right)(\theta_{k})+R_{k}
\]
where $R_{k}$ is the remainder term. The term $A_0$ was introduced in Equation \eqref{eq:numin_taylor_expansion_formal} in Section \ref{sec:21}.

Using that $\mathcal{L}\psi=\phi-\bar{\phi}$, summing over $k$ and dividing
by $hK$ yields 
\[
\hat{\phi}_{K}:=\frac{1}{K}\sum_{k=0}^{K-1}\phi(\theta_{k})=\bar{\phi}+\frac{1}{Kh}\left(\psi\left(\theta_{K}\right)-\psi\left(\theta_{0}\right)\right)-\frac{1}{hK}\sum_{k=0}^{K-1}h\left(A_{0}-\mathcal{L}\right)\psi(\theta_{k})-\frac{1}{Kh}\sum_{k=0}^{K-1}R_{k}.
\]
Controlling $A_{0}-\mathcal{L}$ and the remainder gives rise to Theorem
5.1 and 5.2 in \cite{2010MattinglyPoisson} stating that 
\begin{align}\label{eq:MattinglyFiniteTimeBounds}
\begin{aligned}
\left|\EE\hat{\phi}_{K}-\bar{\phi}\right| &\leq C\left(h+\frac{1}{h\cdot K}\right)\\
 \EE\left(\hat{\phi}_{K}-\bar{\phi}\right)^{2}&\leq C\left(h^{2}+\frac{1}{h\cdot K}\right). 
\end{aligned}
\end{align}

In particular, these results were derived for discretisations
of SDEs on the torus. This simplifies the presentation because the
derivates of $\psi$ are bounded on a compact set. 
However, the same
arguments hold if the following assumption is imposed instead
\begin{equation}\label{eq:supExpec}
\sup_{k}\EE\norm{\psi^{(i)}\left(\theta_{k}\right)}<\infty \text{ for }i=1,\dots,4.
\end{equation}
verifying this condition will allow us to work on $\mathbb{R}^d$.

\subsection{The Bias and the MSE of Finite Time SGLD Averages\label{sec:PoissonMain}}
We consider the SDE 
\begin{equation}
d\theta_{t}=f(\theta_{t})dt+g\left(\theta_{t}\right)dW_{t}.\label{eq:SDE}
\end{equation}
with  $g=I$ being the identity matrix but we keep $g$ in order to make the presentation
clearer.
Based on this setup the recursion of the corresponding SGLD reads as follows
\[
\Delta_{k+1}=\theta_{k+1}-\theta_{k}=\hat{f}_{k}h+h^{\frac{1}{2}}g_{k}\xi_{k+1}
\]
where  $\hat{f}_k$ is an unbiased estimate of $f$.
The focus of this section is to establish results similar to Equation \eqref{eq:MattinglyFiniteTimeBounds} for the SGLD. They will be formulated in Theorem \ref{thm:FiniteTimeBiasVariance}.

For the readability of the subsequent calculation we use the following notations $$\Delta_{k+1}=\theta_{k+1}-\theta_{k}, \phi_{k}=\phi\left(\theta_{k}\right),$$
$\hat{f}_{k}=\hat{f}(\theta_{k},\tau_{k},h)$ for the estimate of
the drift, $g_{k}=g(\theta_{k},h)=I$, $\poisol_{k}=\poisol(\theta_{k})$,
$V_{k}=V(\theta_{k})$ and $D^{k}\psi_{k}=D^{k}\psi(\theta_{k})$.
The term $A_0$, as introduced in in Equation \eqref{eq:numin_taylor_expansion_formal} in Section \ref{sec:21}, satisfies $A_0=\mathcal{L}$ but we keep $A_0$ for clarity. Thus, we have 
\[
A_{0}\psi_k=\derivf\psi\cdot\mathbb{E}_{\tau}\hat{f}(\theta_k,\tau,h)+\frac{1}{2}S\left(\cdot,h\right):\hess \psi(\theta_k)
\]
where $S(x,h)=g(x,h)g(x,h)^{T}=I$. 

We use the following third order Taylor expansion on $\poisol(\theta_{k+1})-\poisol(\theta_{k})$ in order to obtain a bound on $frac{1}{K}\sum_{k=0}^{K-1}\left(\phi_{k}-\bar{\phi}\right)$
\begin{eqnarray*}
\psi_{k+1} & = & \poisol_{k}+\derivf\poisol_{k} \cdot \Delta_{k+1}+\frac{1}{2}\Delta_{k+1}^T \hess\poisol_{\step}\Delta_{k+1}+\frac{1}{6}\deriv{3}{\poisol_{k+1}}\left(\Delta_{k+1},\Delta_{k+1},\Delta_{k+1}\right)+R_{k+1}\\
R_{k+1} & =\frac{1}{6} & \int_{0}^{1}s^{3}\deriv{4}{\poisol}\left(s\theta_{k}+(1-s)\theta_{\step+1}\right)\left(\Delta_{k+1},\Delta_{k+1},\Delta_{k+1},\Delta_{k+1}\right)ds.
\end{eqnarray*}
Here $\deriv{3}{\poisol}$ and $\deriv{4}{\poisol}$ are  the third and fourth order derivative in the form of a trilinear and a quadrilinear form, respectively.
 In this setting, a third order expansion is required in order to obtain the $h^2$ term in the $h^2+\frac{1}{T}$ bound in the MSE (see Equation \eqref{eq:MattinglyFiniteTimeBounds} or Theorem \ref{thm:FiniteTimeBiasVariance}). More precisely, the remainder of this expansion is forth order which together with the term $\sqrt{h}\xi_m$ in Equation \eqref{eq:SGLD} contributes to the  $h^2$ error term.
In order to make the connection to the Poisson equation, we write the expansion above in terms of $A_{0}$. This yields
\begin{eqnarray*}
\poisol_{k+1} & = & \poisol_{k}+hA_{0}\poisol_k+h^{\frac{1}{2}}\derivf\poisol_{k}\cdot \left(g_{k}\xi_{k+1}\right)+h\derivf\poisol_{k} \cdot \left(\underbrace{
\hat{f}_{k}-\mathbb{E}_{\tau}\hat{f}(\theta_{k},\tau,h)}_{H_{k}}\right)+h^{\frac{3}{2}}\left(g_{k}\xi_{k+1}\right)^T\hess\psi_{k}\hat{f}_{k}\\
 &  & +\frac{1}{2}\hat{f}_{k}^\transpose h^{2}\hess\poisol_{k}\hat{f}_{k}+\frac{1}{6}\deriv{3}{\poisol_{k}}\left(\Delta_{k+1},\Delta_{k+1},\Delta_{k+1}\right)+r_{k+1}+R_{k+1}
\end{eqnarray*}
where $r_{k+1}=h\frac{1}{2}\left(\left(g_{\step}\xi_{\step+1}\right)^\transpose\hess\poisol_{k}\left(g_{\step}\xi_{\step+1}\right)-S(x,h)\right)$.

Notice that $\frac{1}{hK} \sum_{k=0}^{K-1}hA_{0}\poisol_k =\frac{1}{K}\sum_{k=0}^{K-1}\left(\phi_{k}-\bar{\phi}\right)$ is the error of interest. In order to control this error, we sum the expression for $\psi_{k+1}$ for $k=0,\dots,K-1$ and divide by $T=hK$. Grouping the terms for subsequent inspection gives
 \begin{eqnarray}
\frac{\psi_{K}-\psi_{0}}{Kh} & = & \frac{1}{K}\sum_{k=0}^{K-1}\left(\phi_{k}-\bar{\phi}\right)+\underbrace{\sum_{k=0}^{K-1}\left(\mathcal{L}-A_{0}\right)\psi_{k}}_0\nonumber \\
 & = & \frac{1}{T}\underbrace{\sum_{k=0}^{K-1}r_{\step+1}}_{M_{1,K}}+\frac{1}{T}\underbrace{h^{\frac{1}{2}}\sum_{k=0}^{K-1}\derivf \poisol_{\step}\left(g_{\step}\xi_{\step+1}\right)}_{M_{2,K}}+\frac{1}{T}\underbrace{h^{\frac{3}{2}}\sum_{k=0}^{K-1}\hat{f}_{\step}^\transpose\hess\psi_{\step}\left(g_{\step}\xi_{\step+1}\right)}_{M_{3,K}}\nonumber \\
 &  & +\frac{1}{T}\underbrace{h\sum_{k=0}^{K-1}\derivf\poisol_{\step}\cdot\left(\hat{f}_{k}-\mathbb{E}_{\tau}f(\theta_{k},\tau,h)\right)}_{M_{4,K}}+\frac{1}{T}\underbrace{\frac{1}{2}\sum_{k=0}^{K-1}h^{2}\hat{f}_{\step}^\transpose\hess\poisol_{\step}\hat{f}_{\step}}_{S_{1,K}}\nonumber \\
 &  & +\frac{1}{T}\underbrace{\sum_{k=0}^{K-1}R_{k+1}}_{S_{2,K}}+\frac{1}{T}\underbrace{\frac{1}{6}\sum_{k=0}^{K-1}\deriv{3}{\poisol_{\step}}\left(\Delta_{\step+1},\Delta_{\step+1},\Delta_{\step+1}\right)}_{S_{3,K}}. \label{eq:sampleAvg}
\end{eqnarray}
where the $M_{i,k}$ indicate the martingale terms and the $S_{i,k}$
other remainder terms. We split 
\[
S_{\text{3,K}}=M_{0,K}+\tilde{M}_{0,K}+S_{0,K}+\tilde{S}_{0,K}
\]
in terms of
\begin{eqnarray*}
M_{0,K} & = & \frac{1}{6}h^{\frac{3}{2}}\sum_{k=0}^{K-1}\left(\deriv{3}{\poisol_{k}}\left(\left(g_{k}\eta_{\step+1}\right),\left(g_{k}\eta_{\step+1}\right),\left(g_{k}\eta_{\step+1}\right)\right)\right)\\
\tilde{M}_{0,K} & = & \frac{1}{2}\sum_{k=0}^{K-1}h^{\frac{5}{2}}\deriv{3}{\poisol_{k}}\left(\hat{f}_{\step},\hat{f}_{\step},g_{k}\eta_{\step+1}\right)\\
S_{0,K} & = & \frac{1}{6}\sum_{k=0}^{K-1}3h^{2}3\deriv{3}{\poisol_{k}}\left(g_{k}\eta_{\step+1},g_{k}\eta_{\step+1},\hat{f}_{\step}\right)\\
\tilde{S}_{0,K} & = & \frac{1}{6}\sum_{k=0}^{K-1}h^{3}\deriv{3}{\poisol_{k}}\left(\hat{f}_{\step},\hat{f}_{\step},\hat{f}_{\step}\right).
\end{eqnarray*}

Rearranging Equation \eqref{eq:sampleAvg} for $\frac{1}{K}\sum_{k=0}^{K-1}\left(\phi_{k}-\bar{\phi}\right)$ and controlling the resulting right hand side of Equation \eqref{eq:sampleAvg} gives rise to the following theorem.

\begin{theorem} \label{thm:FiniteTimeBiasVariance} Suppose that there
exists a function $V$ such that the following three assumptions hold: 
\begin{enumerate}
\item {There are  $p_{\poisol,1},\dots p_{\poisol,4}\in (0,\infty)$ such that the derivatives of the solution $\psi$ to the Poisson equation
satisfy the following bound 
\begin{equation}
\left\Vert \deriv{k}{\psi}\right\Vert \lesssim V^{p_{\poisol,k}},\quad\text{ for }k=0,\dots,4.\label{eq:PoissonDerivBounds}
\end{equation}
} 
\item {The drift $f$ and the error from the estimate $H:=\hat{f}(\theta,\tau)-f(\theta)$ satisfy 
\begin{align}
\begin{aligned}\label{eq:subsamplingCond}
\EE_{\tau}H\left(\theta,\tau\right)^{2p} & \lesssim  V(\theta)^{p}\quad\forall p\leq p^{\star}\\
\left\Vert f\right\Vert ^{2} & \lesssim  V.
\end{aligned}
\end{align}
for $p^{\star}=\max\left\{ 2p_{\poisol,2}+2,2p_{\poisol,4}+4,2p_{\poisol,3}+1,2p_{\poisol,3}+3\right\} $.
Moreover, we suppose that the $\EE V^{p}(\theta_{k})$ is bounded
from above and that this bound is independent of $k$, that is 
\begin{equation}
\sup_{k}\EE V^{p}\left(\theta_{k}\right)<\infty,\quad\forall p\leq p^{\star}.\label{eq:LyapunovBddApriori}
\end{equation}
} 
\item {$V$ satisfies 
\begin{equation}
\sup_{s}V\left(s\theta_1+(1-s)\theta_2\right)^{p}\lesssim V(\theta_1)^{p}+V(\theta_2)^{p},\quad \text{for all } \theta_1,\theta_2,p\leq p^{\star}.\label{eq:lineV}
\end{equation}
} 
\end{enumerate}
Under these assumptions there exists $h_0>0$ and constant $C$ such that for all $h<h_0$
\begin{eqnarray}
\textup{Bias}\left(\hat{\phi}_{K}\right) & = & \left|\EE\hat{\phi}_{K}-\bar{\phi}\right| \leq C\left(h+\frac{1}{K h}\right)\label{eq:bias}\\
\EE\left(\hat{\phi}_{K}-\bar{\phi}\right)^{2} & \leq & C\left(h^{2}+\frac{1}{K h}\right)\label{eq:vairance}
\end{eqnarray}
where 
\[
\hat{\phi}_{K}=\frac{1}{K}\sum_{k=0}^{K-1}\phi(\theta_{k})\quad\text{ and }\quad\bar{\phi}=\EE_{\pi}\phi.
\]
\end{theorem}
\begin{proof}
For each term we bound the term inside the sum  by a power of $V^{p}$  and then obtain an overall bound using $\sup_{i}\EE V_{i}^{p}<\infty$. For example, $\frac{1}{T}\mathbb{E}S_{1,K}$ can be bounded as follows
\begin{eqnarray*}
\frac{1}{T}\mathbb{E}S_{1,K} & \lesssim & \frac{1}{T}\EE\sum_{k=0}^{K-1}h^{2}V_{k}^{p_{\poisol,2}}\EE_{\tau_{k}}\left\Vert \hat{f}_{\step}\right\Vert ^{2}\\
 & \lesssim & \frac{1}{T}\sum_{k=0}^{K-1}h^{2}\sup_{i}\EE V_{i}^{p_{\poisol,2}+1} \lesssim  \frac{1}{T}h^{2}K\lesssim h.
\end{eqnarray*}
The details of this computation are contained in Appendix \ref{sec:finitetimeproof}.
\end{proof}

Theorem \ref{thm:FiniteTimeBiasVariance} and the results for the decreasing step size SGLD  \cite{TehThierryVollmerSGLD2014} hold under assumptions formulated in terms of the solution $\psi$ of the Poisson equation. More precisely, the crucial step is to establish a bound of the form 
$$\sup_{k}\EE\norm{\deriv{i}{\poisol_{k+1}}\left(\theta_{k}\right)}<\infty \text{ for }k=1,\dots,4.$$
This bound is established using Equations \eqref{eq:PoissonDerivBounds} and \eqref{eq:LyapunovBddApriori}
\begin{equation*}
\sup_{k}\EE\norm{D^{(i)}\psi\left(\theta_{k}\right)}\lesssim\sup_{k}\EE\norm V^{p_{\poisol,i}} < \infty \text{ for }i=1,\dots,4.
\end{equation*}
Thus, we are left with finding an appropriate Lyapunov function $V$ such that Equations \eqref{eq:PoissonDerivBounds} and \eqref{eq:LyapunovBddApriori} hold. In Appendix \ref{sec:regPoisson}, we formulate strong sufficient conditions on $\pi$ that ensure that these assumptions are satisfied and that Theorem \ref{thm:FiniteTimeBiasVariance} is applicable.

\section{An Analytic Investigation of the Toy Model \label{sec:anal_toy}}
We now extend our analysis of the one-dimensional Gaussian toy model introduced in Section \ref{sec:GaussianModel} beyond the general results of the previous two sections. 
More precisely, in Section  \ref{sub:comp_effort}, we compare the Euler method, the SGLD and the mSGLD by comparing the computational cost for fixed level accuracy specified in terms of the mean square error in estimating the second moment ($\text{MSE}_{2}$), optimising over the step size $h$, the subsample size $n$ and the number of steps $M$. A numerical solution to the resulting optimisation problem demonstrates that the SGLD is advantageous  in the lower accuracy regime while it degenerates to $n=N$ in the high accuracy regime. On the other hand the mSGLD does not degenerate and seems to maintain a constant speed up  compared to the Euler method. In Section  \ref{sub:toyFixedAndIncreasingN} we then consider the $\text{MSE}_2$ and use an analytic expression to study the behaviour of these algorithms for growing $N$.   This allows us to extend the analysis of Sections \ref{sec:GaussianModel} and  \ref{sec:Weak-Convergence-Analysis}  (in which we only consider the case  limit $h\rightarrow 0$) and study the asymptotic bias of the SGLD and  the mSGLD by scaling both $n$ and $h$  in $N$.


In Section \ref{sub:toyLimitNinfty} we finally adopt a different viewpoint by considering a fixed value of our parameter $\theta$, denoted by $\theta^{\dagger}$, while we take expectations with respect to the realisation of the data $\{X_i\}$. This enables us to study  how $\EE_X (\text{MSE}_2)$  behaves in the limit  of $N\rightarrow\infty$. In particular, we find that for the case of the SGLD, the computational cost in order for  $\EE_X (\text{MSE}_2) \rightarrow 0$ is reduced by a factor of $N$ when compared to the Euler method. A similar analysis for the expected relative error in estimating the posterior variance (ERE) 
\begin{equation}
\text{ERE}=\EEnoise \frac{\frac{1}{K}\sum_{i=0}^{K-1}\theta_{i}^{2}-\left(\frac{1}{K}\sum_{i=0}^{K-1}\theta_{i}\right)^{2}}{\sigma_{p}^{2}}-1.\label{eq:toyEREintro}
\end{equation}
reveals  that under the constraint $\EE_X (\text{ERE})\rightarrow 0$ the Euler method and the SGLD have the same computational cost on the algebraic scale in $N$.

\subsection{Minimising Computational Effort for Constrained $\text{MSE}_2$} \label{sub:comp_effort}

In Section \ref{sec:toyMSE} we compared the Euler
method, the SGLD and the mSGLD for the same choice of $r=\frac{h}{A}.$
In the following we numerically minimise the computational effort with
respect to the condition $\text{MSE}_2\leq \epsilon^2$. We assume that the computational
cost is proportional to $M\cdot n$ which leads to the problem of solving
\begin{eqnarray}
\min_{F} &  & M\cdot n\label{eq:optimisation}\\
\text{subject to} &  & \text{MSE}_2(r,M,n)\leq \epsilon^2\nonumber \\
\text{w.r.t.} &  & r<1,M,n.\nonumber 
\end{eqnarray}
Even though we have analytic expressions for the $\text{MSE}_2$, the solution
to the optimisation problem does not have a closed form. To conclude our analysis, 
we illustrate the numerical solution to this problem for $N=1000$ for
the Euler method, the SGLD and the mSGLD.  The results,   depicted
in Figure \ref{fig:MSEoptimised}, can be summarised as follows: 
\begin{enumerate}
\item as $\epsilon$ becomes smaller, the gain of the SGLD over the Euler method
in terms of computational effort decreases (due to the fact that
$n$ increases);
\item as $\epsilon$ becomes smaller, the mSGLD gains efficiency over the SGLD (the reason being that $n$ seems to asymptote as $\epsilon$ decreases).
\end{enumerate}
\begin{figure}
\subfloat[Minimised computational cost for the Euler, the SGLD and the mSGLD algorithm]{\includegraphics[width=0.42\textwidth]{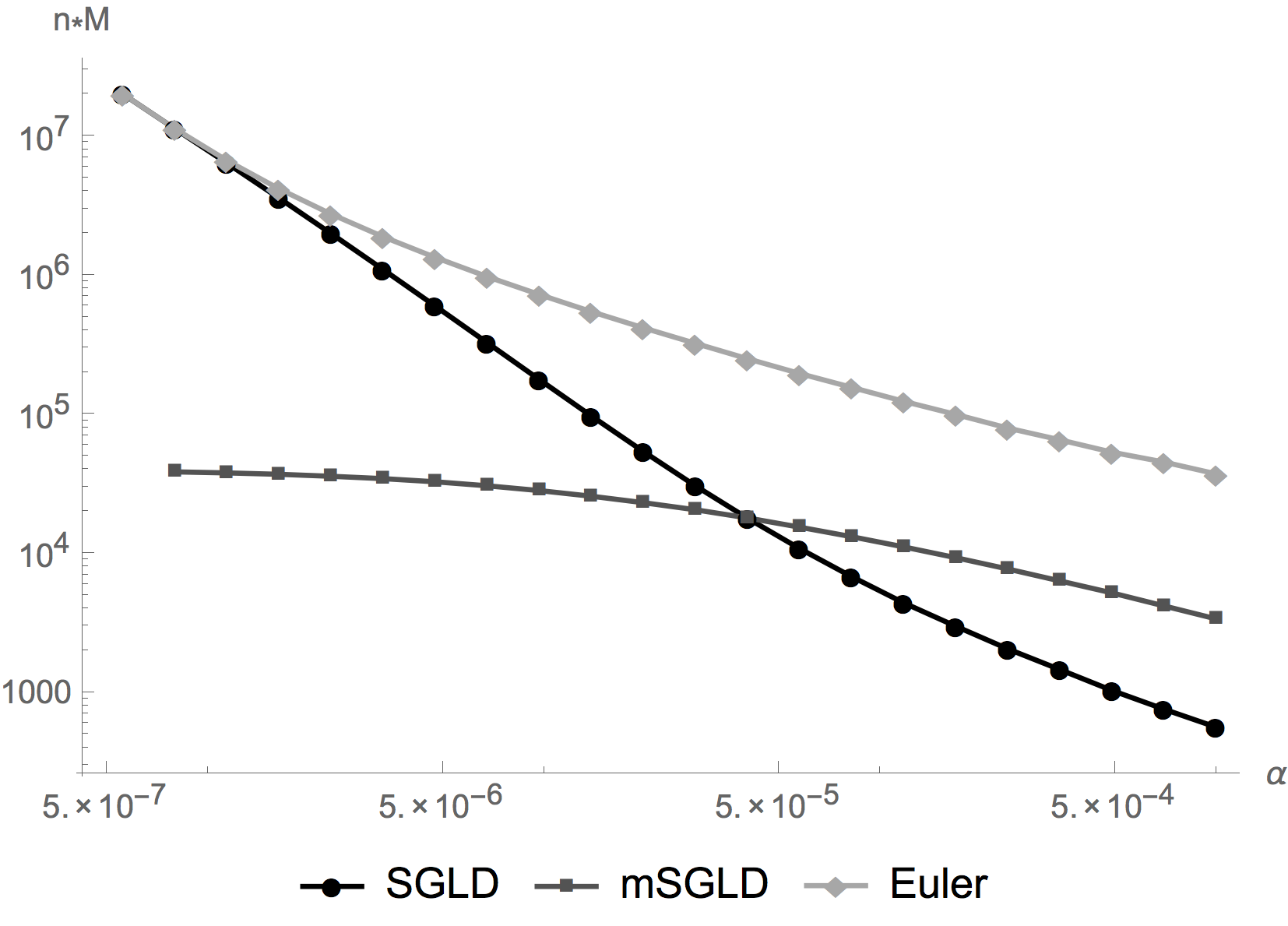}
}\subfloat[Subset size $n$ for minimised computational cost]{\includegraphics[width=0.42\textwidth]{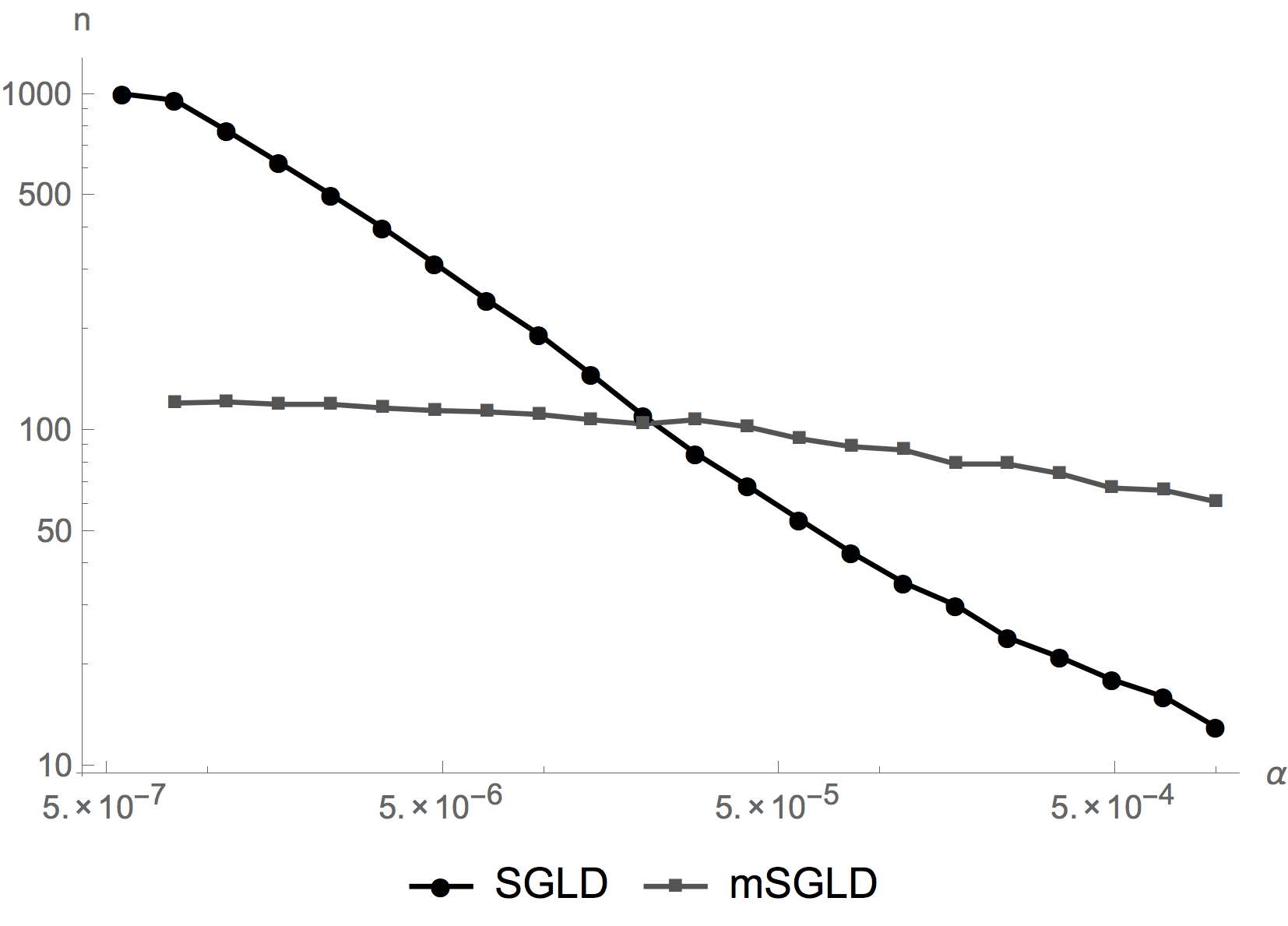}
}
\newline
\subfloat[Step size $r=\frac{h}{A}$ for minimised computational cost]{\includegraphics[width=0.42\textwidth]{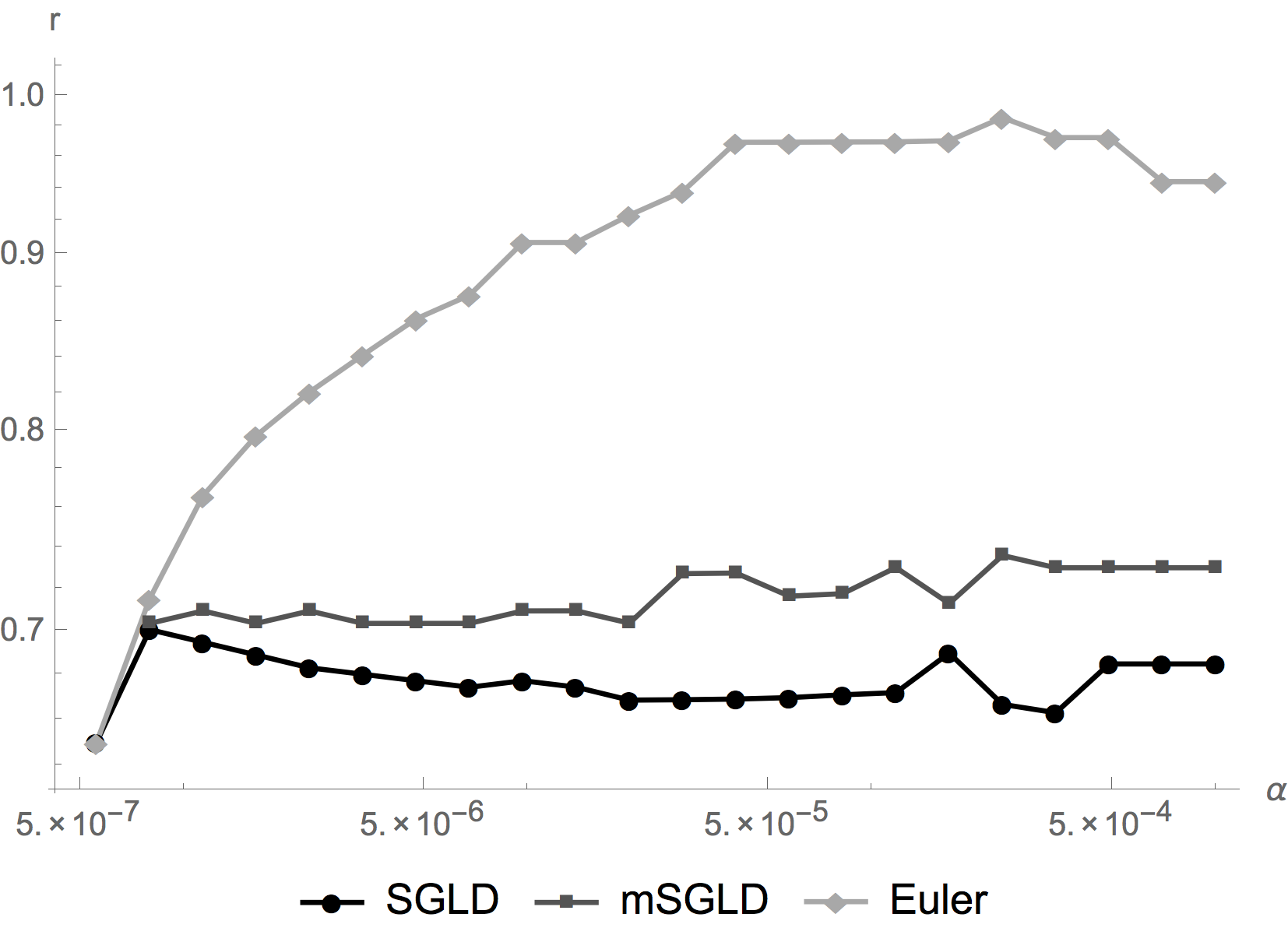}
}\caption{\label{fig:MSEoptimised}Minimisition of computational cost $\propto M\cdot n$
subject to $MSE\le \epsilon^2$}
\end{figure}

The upper bound obtained in Equation (\ref{eq:vairance}) suggests a scaling of $M\sim \epsilon^{-3}$ and $r \sim \epsilon$ to obtain an MSE of order $\epsilon^2$ with minimal computational effort. The numerical minimisation of $M$ with respect to $r$ and $M$ subject to the condition $\text{MSE}(r,M)\leq \epsilon^2$ confirms this scaling  empirically, see Figure \ref{fig:euleroptim}.

\begin{figure}
\subfloat[Line fit yields $M(\epsilon) \sim \epsilon^{-3.03}$ as $\epsilon \rightarrow 0$]{\includegraphics[width=0.43\textwidth]{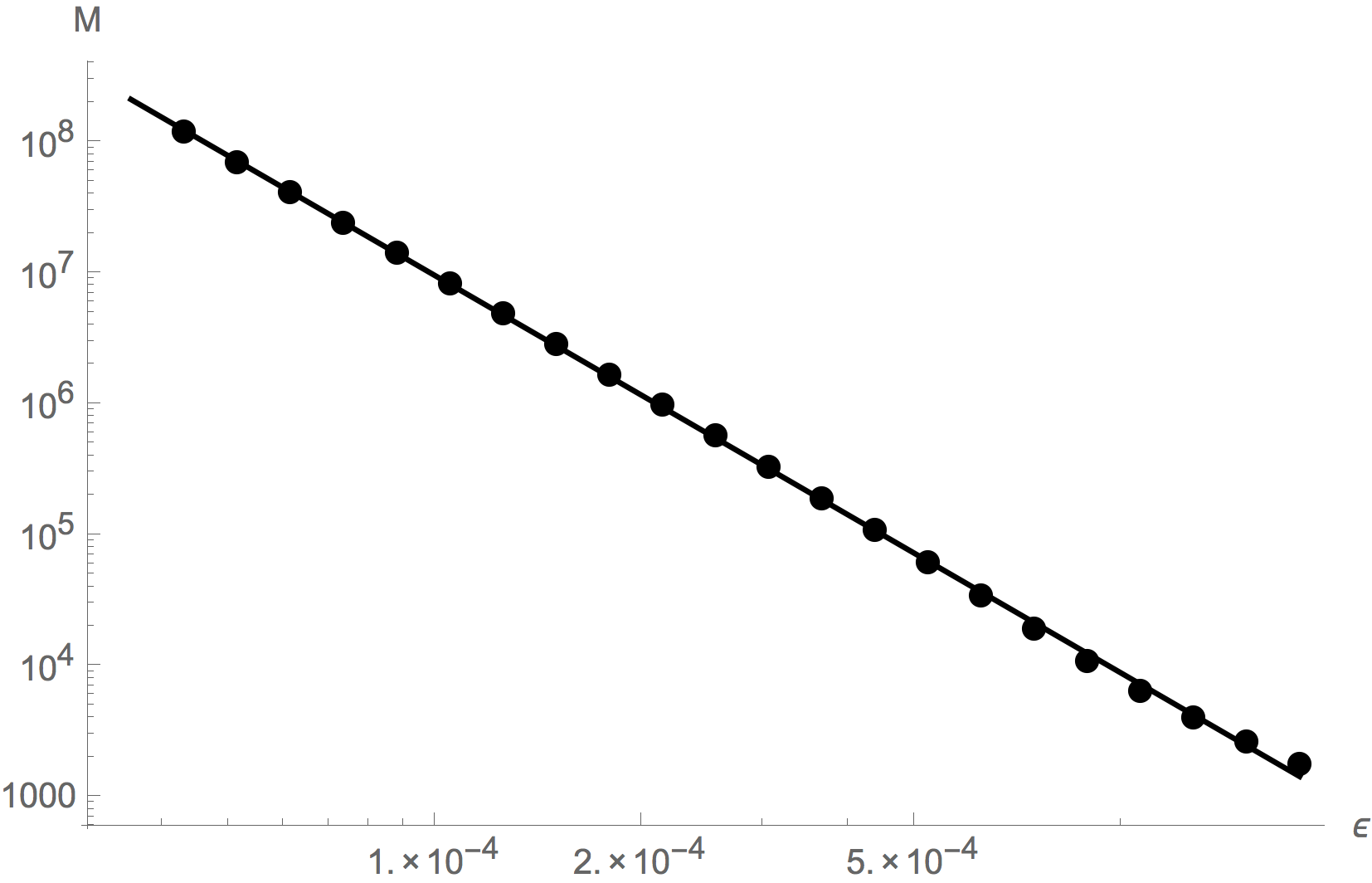}

} \subfloat[Line fit yields $r(\epsilon) \sim \epsilon^{0.95}$ as $\epsilon \rightarrow 0$]{\includegraphics[width=0.43\textwidth]{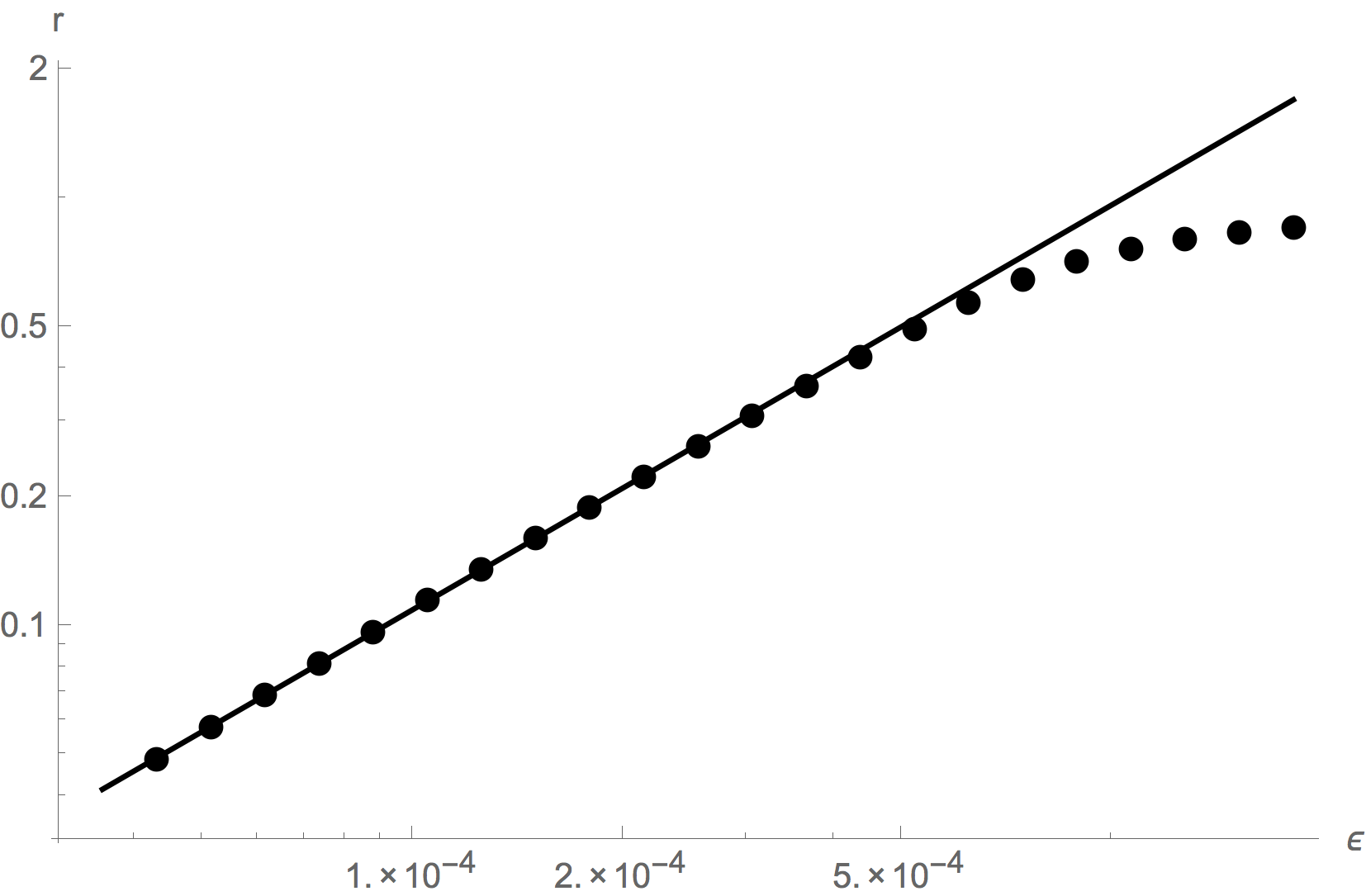}
}

\caption{\label{fig:euleroptim} Scaling of $r(\epsilon)$ and $M(\epsilon)$ for minimal computational cost subject to the $\text{MSE}_2(r(\epsilon),M(\epsilon))\leq \epsilon^2$}
\end{figure}

\subsection{The $\text{MSE}_2$  for fixed and increasing $N$\label{sub:toyFixedAndIncreasingN}}

We now consider the behaviour for growing data size $N$ where a new data set $X$ is generated in each instance. Figure \ref{fig:MSEGrowN} depicts the $\text{MSE}_2$ for $N=10^{i}$
with $i=1,\dots,4$ for the subset choices $n=N^{0.1},N^{0.5}$ and
$N^{0.9}$; each compared to the Euler method corresponding to $n=N$. In this plot we notice that the SGLD outperforms the mSGLD for $n=N^{0.1}$
and $n=N^{0.5}$.  The behaviour in Figure \ref{fig:MSEGrowN} suggests that the mSGLD has a larger bias than the SGLD which seems to contradict the findings of Section \ref{sec:GaussianModel} and \ref{sec:Weak-Convergence-Analysis}. Previously, we have just considered the asymptotic of  $h\rightarrow0$. In contrast, we scale both  $h=r \frac{1}{A(N)}$ and $n=N^p$ in terms of $N$ in Figure \ref{fig:MSEGrowN}. In the following we investigate this relationship further by using the explicit formula for the asymptotic bias which has been made available in Section \ref{sec:GaussianModel} .

  For simplicity we consider sampling without replacement, using the expression in Equation \eqref{eq:VarBToy}.  Similar conclusions hold for sampling with replacement.
First we consider the mSGLD.  From Equation \eqref{eq:asympVarmSGLD} and using the parameterisation $h=r/A$ of the step size, the excess asymptotic bias becomes
\[
h^{2}\frac{\text{Var}(B)^{2}}{4(2A-A^{2}h)}
=\frac{(r/A)^{2}\left(\frac{N-n}{n}\right)^{2}N^{2}\Var(X)^{2}}{4(2-r)A}.
\]
Because $A\sim N$, we see that the excess bias stays bounded for large $N$ if and only if $n\gtrsim N^{\frac{1}{2}}$.  In contrast, the same consideration for the SGLD shows that the excess bias due to subsampling vanishes so
 long as $n\rightarrow\infty$ when $N\rightarrow \infty$.

\begin{figure}
\subfloat[The SGLD for $n=N^{\frac{1}{10}}$]{\includegraphics[width=0.43\textwidth]{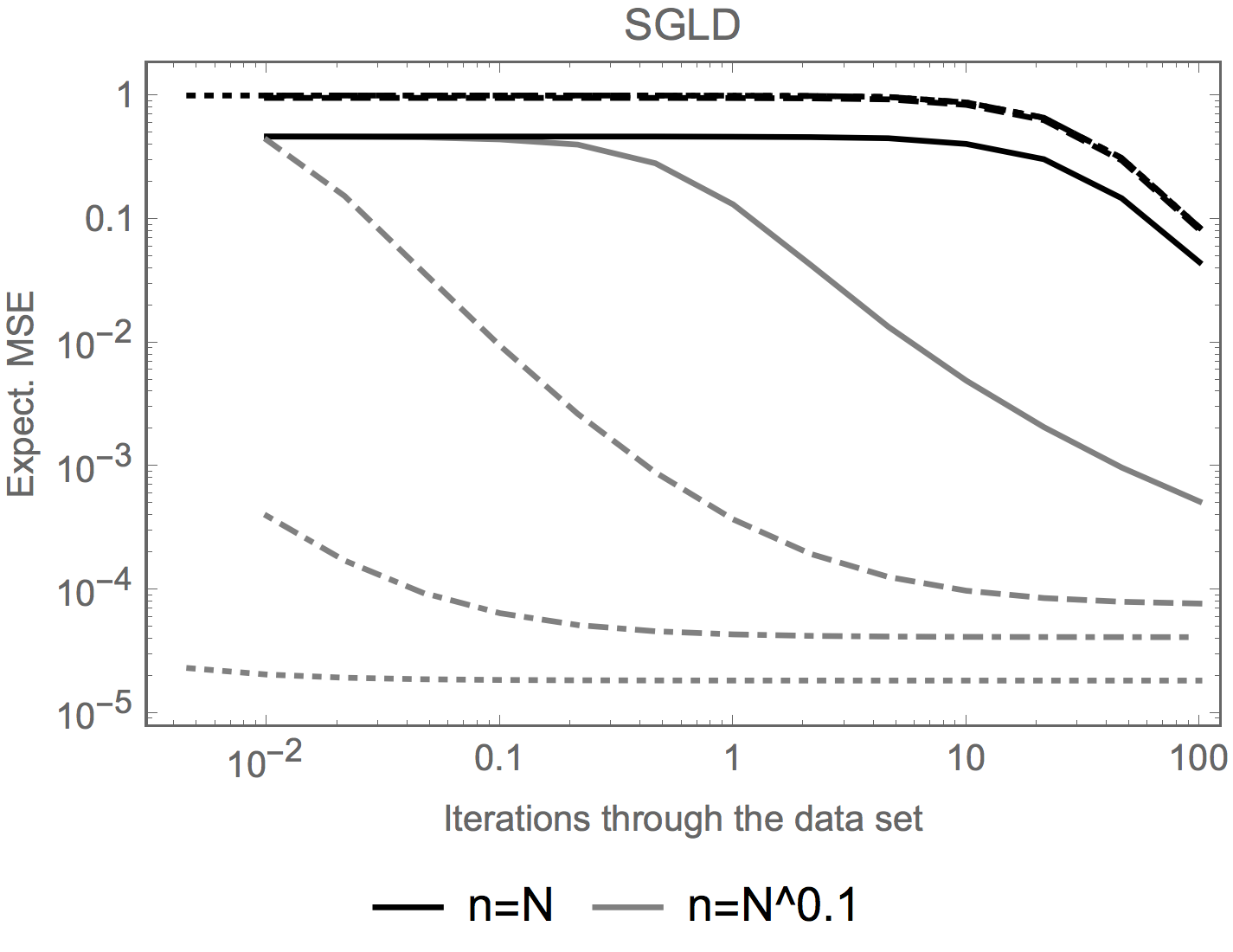}

}\hspace{0.02\textwidth}\subfloat[The mSGLD for $n=N^{\frac{1}{10}}$]{\includegraphics[width=0.43\textwidth]{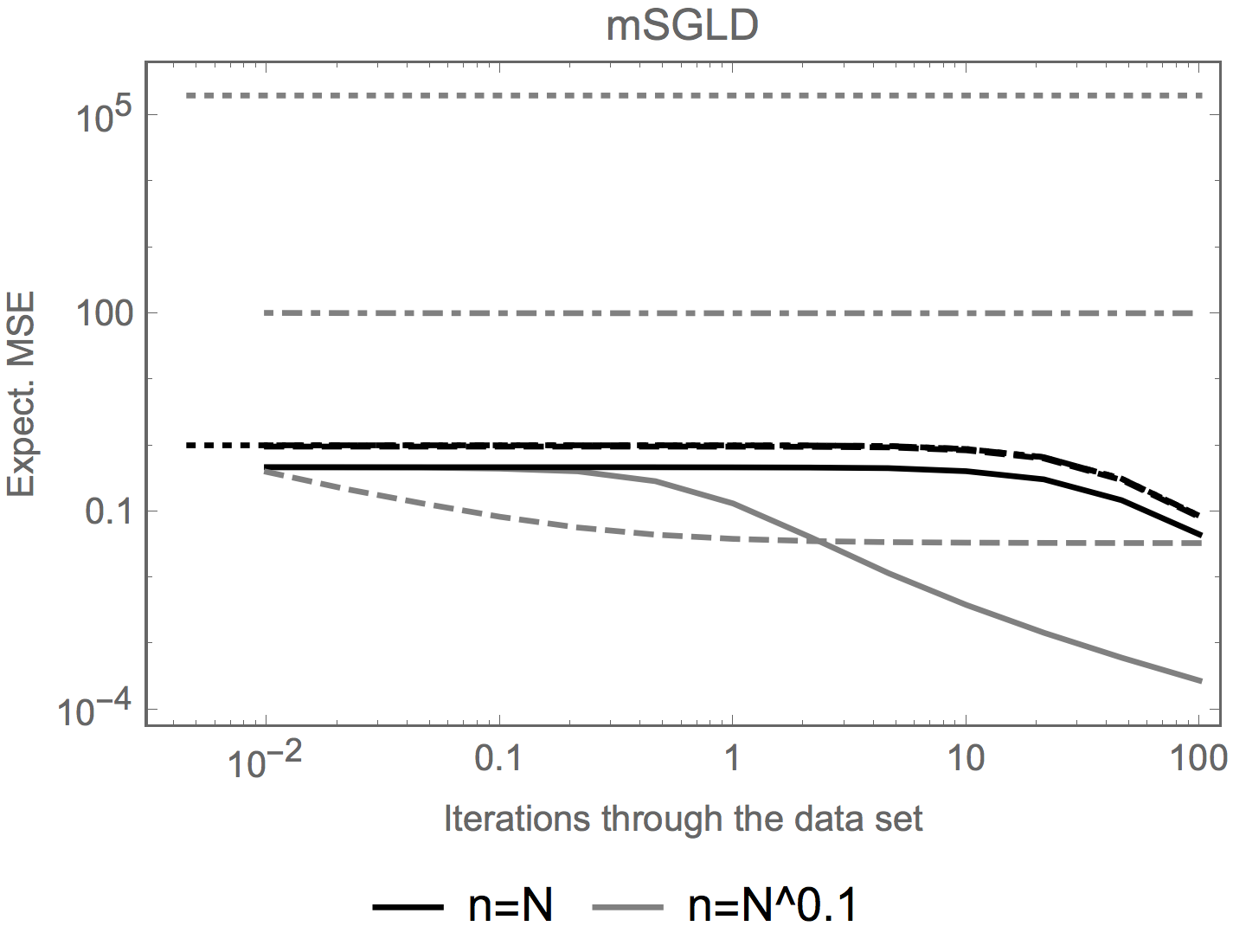}
}
\subfloat[The SGLD for $n=N^{\frac{1}{2}}$]{\includegraphics[width=0.43\textwidth]{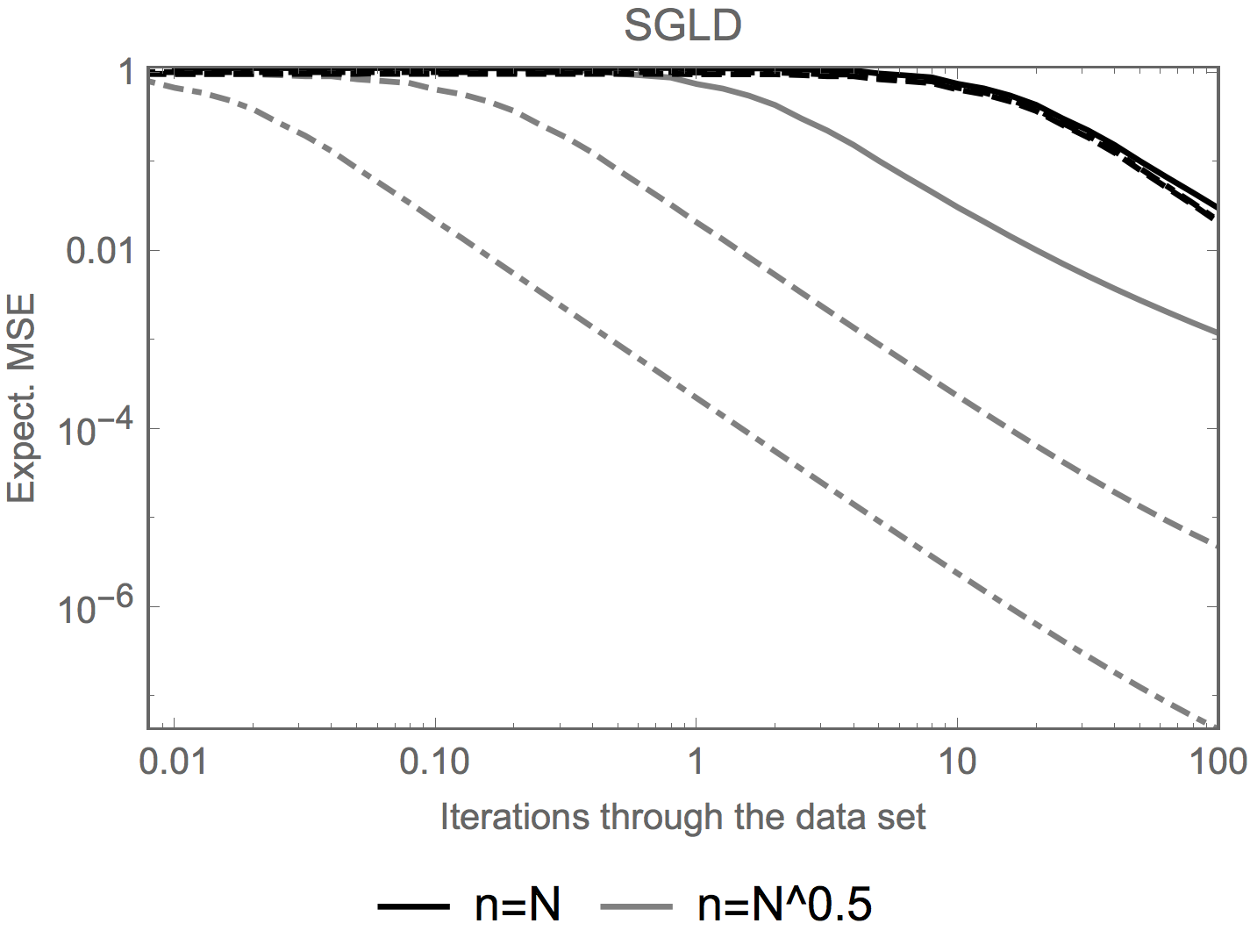}
}\hspace{0.02\textwidth}\subfloat[The mSGLD for $n=N^{\frac{1}{2}}$]{\includegraphics[width=0.43\textwidth]{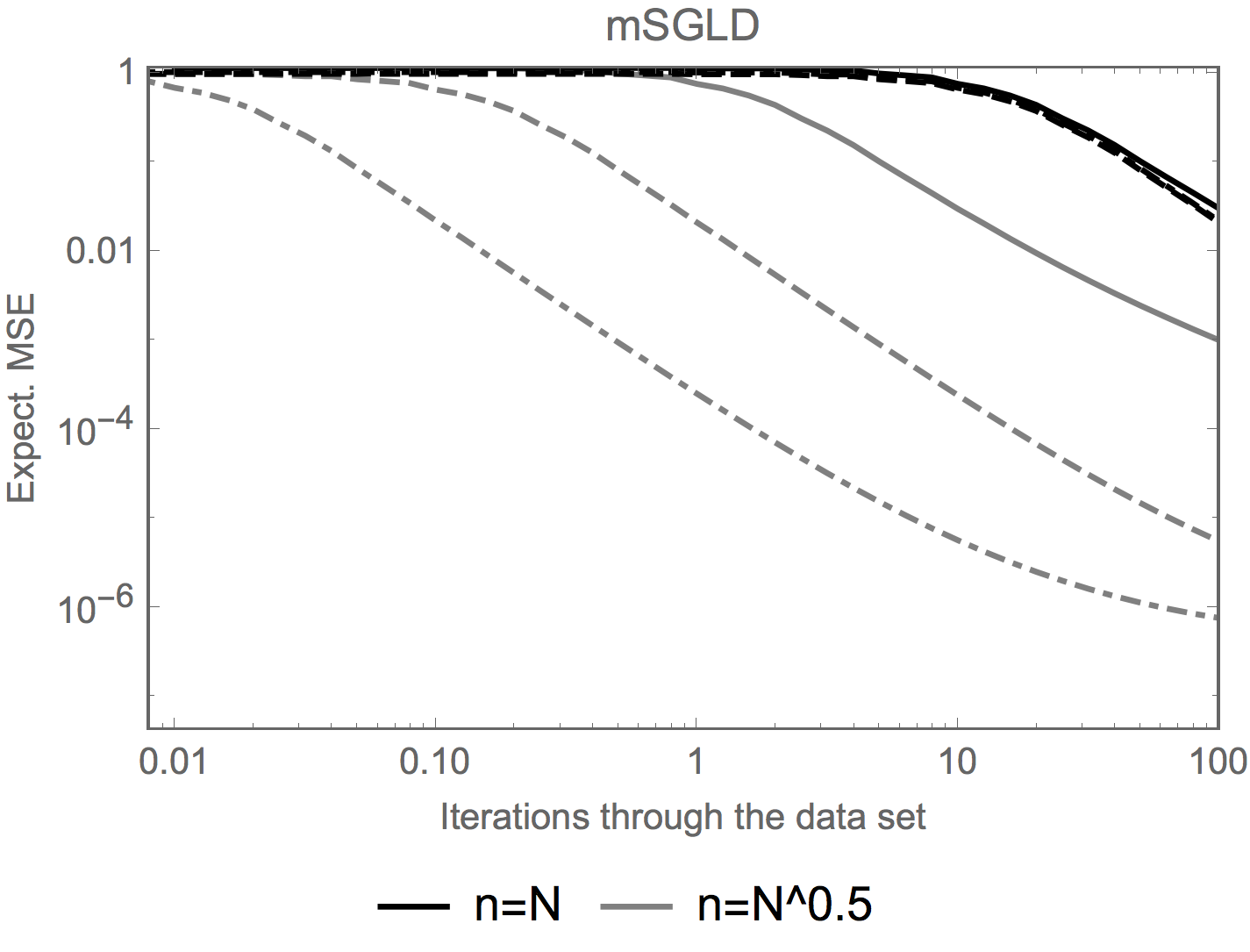}
}
\subfloat[The SGLD for $n=N^{\frac{9}{10}}$]{\includegraphics[width=0.43\textwidth]{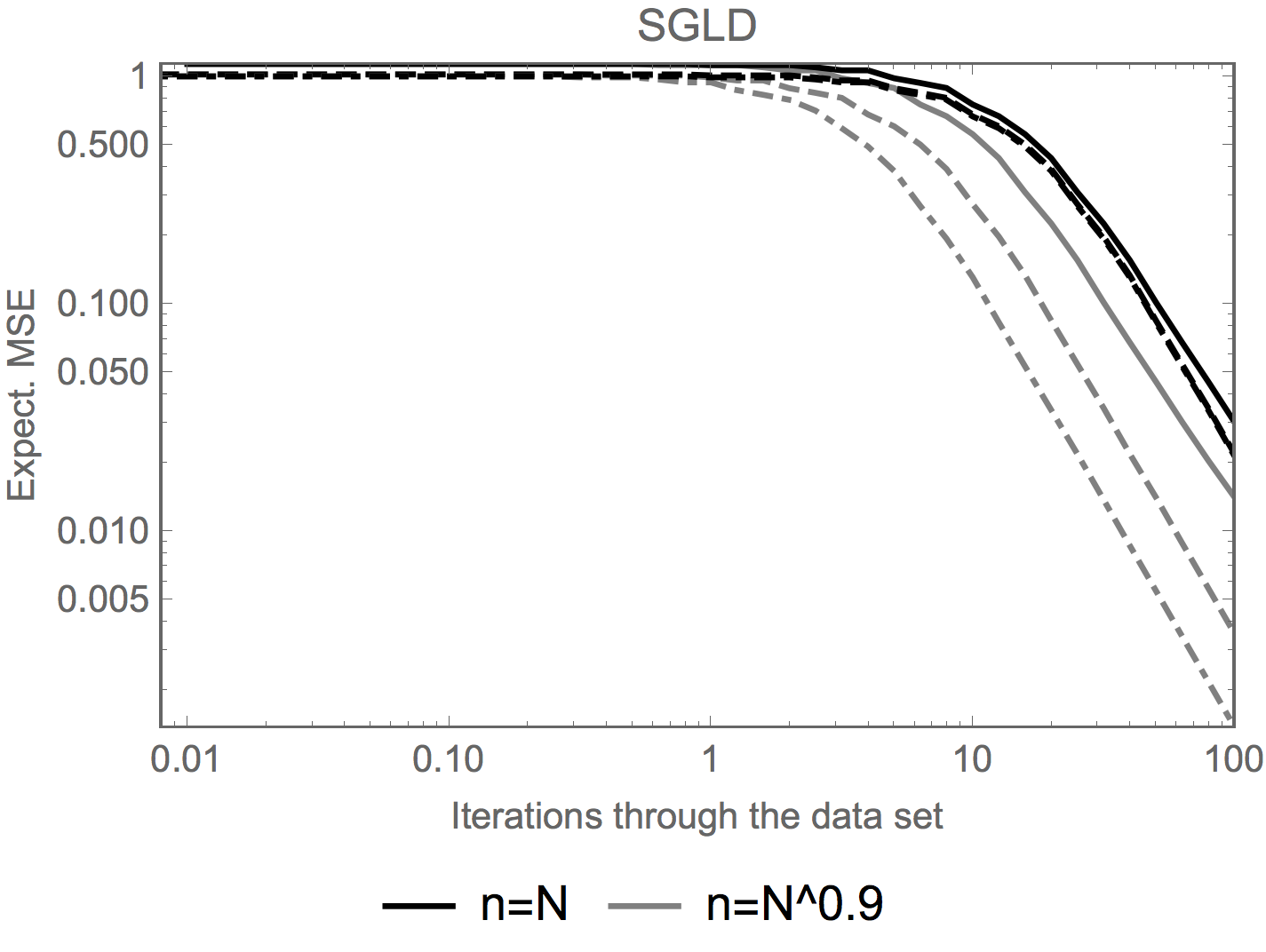}
}\hspace{0.02\textwidth}\subfloat[The mSGLD for $n=N^{\frac{9}{10}}$]{\includegraphics[width=0.43\textwidth]{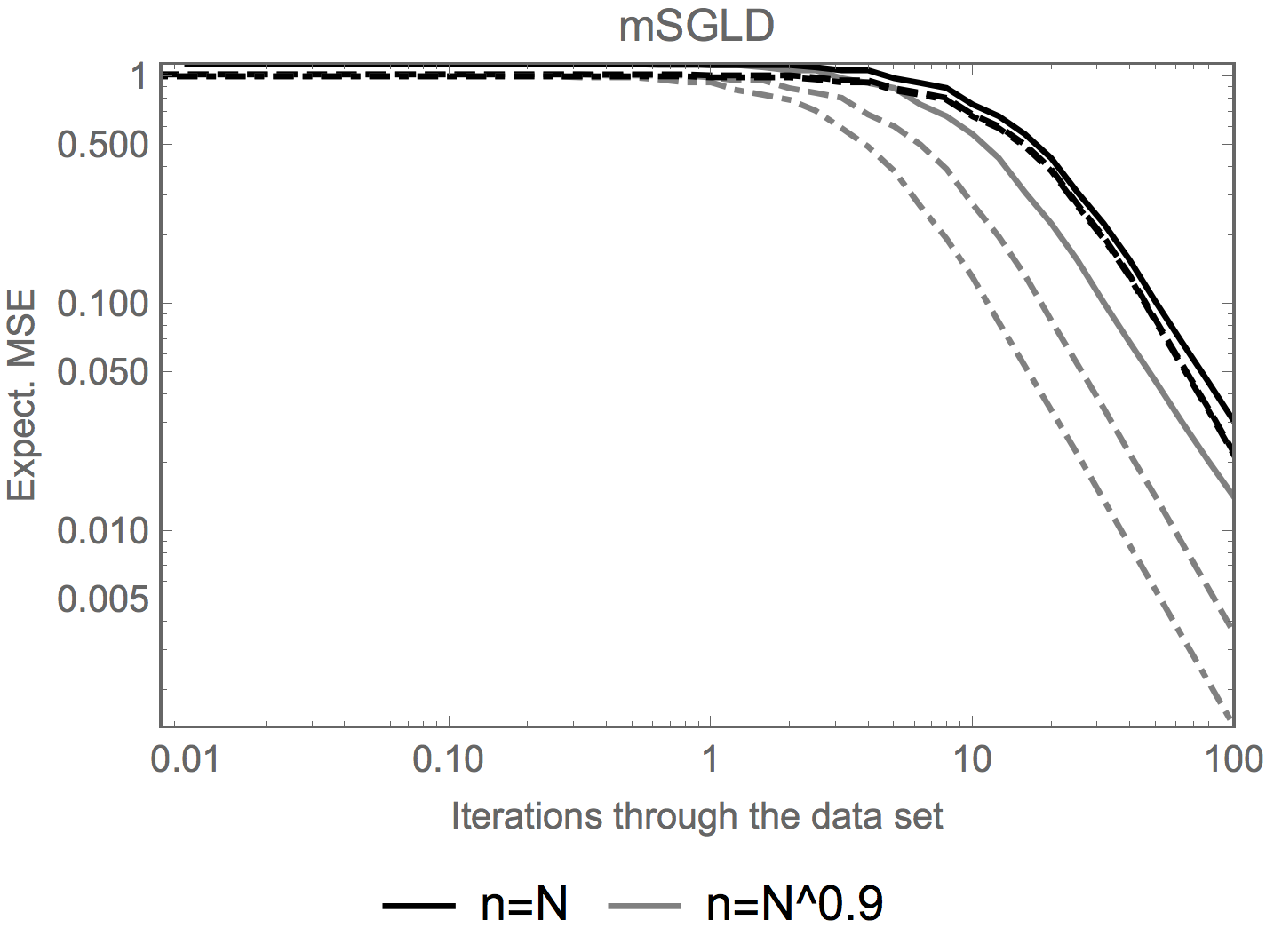}
}

\begin{centering}
\includegraphics[width=0.4\textwidth]{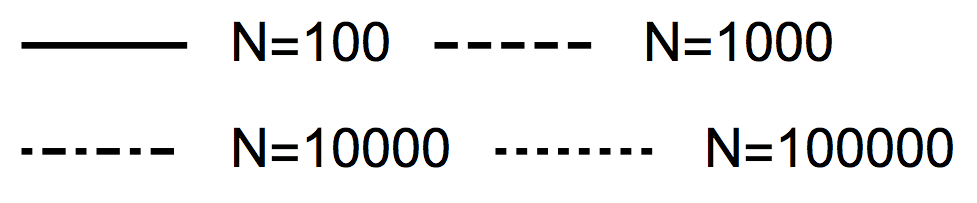} 
\par\end{centering}

\protect\protect\caption{\label{fig:MSEGrowN} $\text{MSE}_2$ of the time average for the SGLD and
the mSGLD for the second moment of the posterior as $N\rightarrow\infty$. Notice that Figures (c) and (d) and (e) and (f) have the same scaling respectively. Moreover, figures (a) and (b) have separate scaling because of the instability of the mSGLD.}
\end{figure}


We can also identify the regime in which  the mSGLD has a smaller asymptotic bias than  the SGLD.  From Equations \eqref{eq:asympVarSGLD} and \eqref{eq:asympVarmSGLD} we see that this is the case when
\begin{equation}\label{eq:BiasMSGLDvsSGLD}
  \frac{h^{2}\operatorname{Var}^{2}(B)}{4(2A-A^{2}h)}\leq \frac{h\operatorname{Var}(B)}{2A-A^{2}h}.
\end{equation}
Using Equation \eqref{eq:VarBToy},  the above can be rearranged to
\[
1 \ge \frac{h}{4}  \frac{1}{16 \sigma^2_x} \frac{N(N-n)}{N}\text{Var}(X). 
\]
Let $c=\frac{n}{N}$ be the relative size of the subsampling.  Using $h=r/A$ where $A$ is given in Equation \eqref{eq:toyBdistribution}, we get
\[
c \ge \frac{ 2 r N \text{Var}\,X}{16\sigma^2_x\left(\frac{1}{\sigma^2_\theta}+\frac{N}{\sigma^2_x}\right)+2 r N \text{Var}\,X}
\]
which in the limit of large data sets $N \rightarrow \infty$ yields
\[
c \geq \frac{ 2 r  \text{Var}\,X}{16+2 r  \text{Var}\,X} > 0.
\]
In conclusion, for a fixed step size given by $r/A$, the mSGLD has a smaller bias than the SGLD if the above holds.  In other words, the subsampling size has to be linear in the size of the data set for a fixed choice of $r$.  

\subsection{Limit of the MSE and ERE for well-specified Data as $N\rightarrow\infty$\label{sub:toyLimitNinfty}}
In order to investigate the limit of $N\rightarrow\infty$, we need to specify the behaviour of the data as well. We study this in the well-specified case, in other words, we assume that the data is generated by the model  for $\theta^{\dagger}=1$. Previously, we have obtained an analytic expression for the expectation of the $\text{MSE}_2$ with respect to the realisation of the noise driving the algorithm. In contrast, we take expectations  with respect to the realisation of the data $X$ and the noise driving the algorithm in the following results. These results are formulated in analytic expressions\footnote{see Appendix \ref{sub:appToyAnalyticNtoInfty} for a sketch of the derivation for the $\text{MSE}_2$ (the derivation for ERE is similar)}  for the $\text{MSE}_2$ and the ERE depending only on $M,n,r$ and $N$.   We then choose $M,n$ and $r$ as functions of $N$ and study the limit $N \rightarrow \infty$ and how this affects the computational cost and the behaviour of the ERE and the $\text{MSE}_2$ as $N\rightarrow \infty$ for the different algorithms.

For the Euler method ($n=N$) we need to take $h<\frac{1}{A}\asymp\frac{1}{N}$
in order to make Equation \eqref{eq:GaussianRecurrence} stable. Moreover,
we need the number of steps $M$ to be of order $N$ to approximate
the diffusion to a time of order $\mathcal{O}(1)$. Because we evaluate
$N$ data points per step, this heuristic argument suggests that the complexity is of order $\mathcal{O}\left(N^{2}\right)$. Furthermore, we verify (using Mathematica\textsuperscript{\textregistered}), that for the Euler method  ($n=N$)
\begin{equation} \label{eq:Euler_beh}
\lim_{N\rightarrow\infty}\EE_X \text{MSE}  =0, \quad 
\lim_{N\rightarrow\infty}\EE_X \text{ERE}  =0
\end{equation}
for the choices $M=N^{1+2\epsilon}$ and $r=N^{-\epsilon}$ for any
$\epsilon>0$. The computational cost for fixed $N$ is $M\cdot n=N^{2+2\epsilon}$.  Thus, this confirms the heuristics we used for the Euler method in Section
\ref{sub:toyFixedAndIncreasingN}.

A natural next question  to ask in terms of the SGLD is if one can have Equation \eqref{eq:Euler_beh} to hold  but for smaller computational complexity than the Euler method.  Using Mathematica\textsuperscript{\textregistered}, we obtain  the following theorem for the MSE
\begin{theorem} For any $\epsilon>0$  and the choices $h=N^{-1-\epsilon}$,
$M=N^{1+2\epsilon}$ and $n=1,$  the SGLD satisfies 
\[
\lim_{N\rightarrow\infty}\mathbb{E}_{{\theta_i}_i,X}\left(\frac{1}{M}\sum_{k=0}^{M-1}\theta_{j}^{2}-(\mu_{p}^{2}+\sigma_{p}^{2})\right)^{2}=0.
\]
\end{theorem}
This constitutes a substantial gain compared to the Euler method because
it reduces the computational complexity in the data size $N$ from being almost quadratic to almost
linear.

We now draw our attention to the expected relative error in estimating the posterior variance, abbreviated by
\begin{equation}
\text{ERE}:=\EEnoise \frac{\frac{1}{K}\sum_{i=0}^{K-1}\theta_{i}^{2}-\left(\frac{1}{K}\sum_{i=0}^{K-1}\theta_{i}\right)^{2}}{\sigma_{p}^{2}}-1.\label{eq:toyEREintro}
\end{equation}
Because the posterior variance goes to zero as $N\Rightarrow\infty$, it is conceivable that it requires more computational effort to ensure that $\lim_{N\rightarrow\infty} \EE_X \text{ERE}=0$. In order to illustrate the behaviour of the ERE, we first consider the behaviour for a fixed data set and repeat the experiment of Figure \ref{fig:MSEGrowN} in Figure \ref{fig:EREGrowN}. The latter demonstrates that the asymptotic bias for the choice $n=N^{\frac{1}{2}}$ (the asymptotes for the grey lines) have an increasing value in $N$.  We used $h=\frac{1}{20A}\sim\frac{1}{N}$ which decreases with $N$. However, we show below that is requirement cancels exactly  the gain from $n\ll N$ at least on the algebraic scale in $N$.

\begin{figure}
\begin{minipage}[t][1\totalheight][c]{0.5\columnwidth}%
\includegraphics[width=1\textwidth]{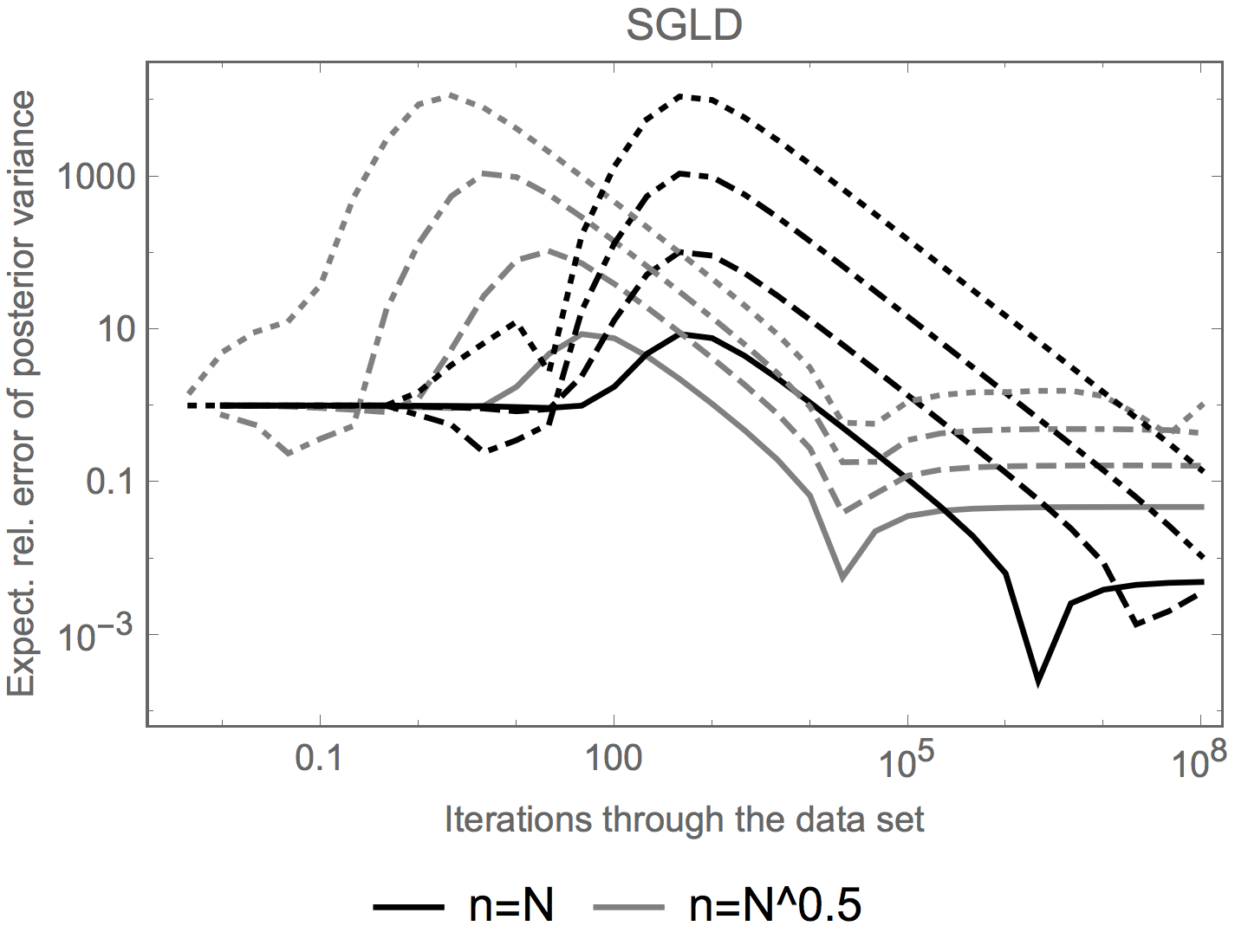}
\end{minipage} \hspace{1cm}\begin{minipage}[t][0.5\totalheight][c]{0.5\columnwidth}
\vspace{3cm}
\begin{centering}
\includegraphics[width=0.8\textwidth]{legendNtoInf} 
\end{centering}
\end{minipage}

\protect\protect\caption{\label{fig:EREGrowN}Expected relative error of the estimate of the variance of the posterior based on the SGLD.}
\end{figure}

In particular, we now choose $r=N^{-\alpha}$, $n=N^{\beta}$ and $M=N^{\gamma}$. Hence
the computational cost is $N^{\beta+\gamma}$. The step size $h=\frac{r}{A}$
satisfies $h\sim N^{-1-\alpha}$ because $A\sim N$. Since the algorithm
performs $M=N^{\gamma}$ steps, we expect it to approximate the diffusion
on the time interval $h\cdot M=N^{-1-\alpha+\gamma}$. 
  Therefore it is
reasonable to require that $\gamma>1+\alpha$. Under this assumption and with the help of 
Mathematica\textsuperscript{\textregistered}, we reduced the limit above to 

\begin{align*}
\lim_{N\rightarrow\infty} \EE_X \text{ERE} & =\lim_{N\rightarrow\infty}\left(\frac{2\cdot N^{-\alpha-\beta+3}}{(N+1)^{2}\left(N^{-\alpha}-2\right)^{2}}-\frac{N^{-2\alpha-\beta+3}}{(N+1)^{2}\left(N^{-\alpha}-2\right)^{2}}\right)\\
 & =\begin{cases}
0 & \text{if }\alpha+\beta>1\\
\infty & \text{if }\alpha+\beta<1.
\end{cases}
\end{align*}
Thus for $\lim_{N\rightarrow\infty} \EE_X \text{ERE}=0$  it is necessary that $\alpha+\beta\ge1$.
 This condition  in turn implies that the computational complexity satisfies 
\[
\underbrace{N^{\gamma}}_{\text{steps}}\times\underbrace{N^{\beta}}_{\text{cost per step}}=N^{1+\alpha}N^{\beta}=N^{1+\alpha+\beta}\gtrsim N^{2}.
\]
Thus, there is no computational gain for the ERE in the limit $N\rightarrow\infty$
on the algebraic scale in $N.$ We note that picking $\theta^\star=\frac{1}{N}$ instead of $\theta^\star=1$ does not change the results. Thus, a closer initialisation does not change the result for the ERE.

\section{Logistic Regression}
\label{sec:6}

\global\long\def\logit{\sigma}

In the following we present numerical simulations for a Bayesian logistic
regression model. The data items are given by covariates $x_i\in \mathbb{R}^d$  that are labeled by $y_i\in\{-1,1\}$.We assume the data $y_{i}\in\{-1,1\}$ is modelled
by 
\begin{equation}
p(y_{i}\vert x_{i},\beta)=\logit(y_{i}\beta^{t}x_{i})\label{eq.logistic}
\end{equation}
where $\logit(z)=\frac{1}{1+\exp(-z)}\in[0,1]$. The model posses the assumption that $y_i$ depends on $x_i$ through the linear relationship $\beta^{t}x_{i}$. Nevertheless, logistic regression is commonly used after a preprocessing has taken place and is therefore used here for numerical illustration. 

We put a Gaussian prior $\mathcal{N}(0,C_{0})$
on $\beta$, for simplicity we use $C_{0}=I$ subsequently. By Bayes'
rule the posterior $\pi$ satisfies 
\[
\pi(\beta)\propto\exp\left(-\frac{1}{2}\norm\beta_{C_{0}}^{2}\right)\prod_{i=1}^{N}\logit(y_{i}\beta^{T}x_{i}).
\]
 We consider $d=3$ and $N=1000$ data points and choose the covariate
to be 
\[
x=\left(\begin{array}{ccc}
x_{1,1} & x_{1,2} & 1\\
x_{2,1} & x_{2,2} & 1\\
\vdots & \vdots & \vdots\\
x_{1000,1} & x_{1000,2} & 1
\end{array}\right)
\]

for a fixed sample of $x_{i,j}\overset{\text{i.i.d.}}{\sim}\mathcal{N}\left(0,1\right)$
for $i=1,\dots1000$ and $j=1,2$. We use a long run of the Random-Walk-Metropolis algorithm to estimate
the posterior mean.

On that basis we estimate the MSE of the SGLD
based mean estimate using $100$ runs of the algorithm with step size
$h=0.002$ for various subset sizes. Figure \ref{fig:logisticRegression}
depicts the MSE as function of the iterations and effective iterations
through the data set. 
Notice that for this example the variance of the stochastic gradient ${\hat{f}(\theta)}$ depends on both $\theta$ and all the data items. For this reason we replace ${\text{Var }\hat{f}(\theta)}$ by an estimate ${\widehat{\text{Var }}\hat{f}(\theta)}$, see also Remark \ref{rem:varMSGLD}. 
We note that the mSGLD
is superior for $n=150$, inferior for $n=50$ and for $n=10$ the
MSE of the mSGLD does not drop below $1$.

\begin{figure}
\includegraphics[width=0.47\textwidth]{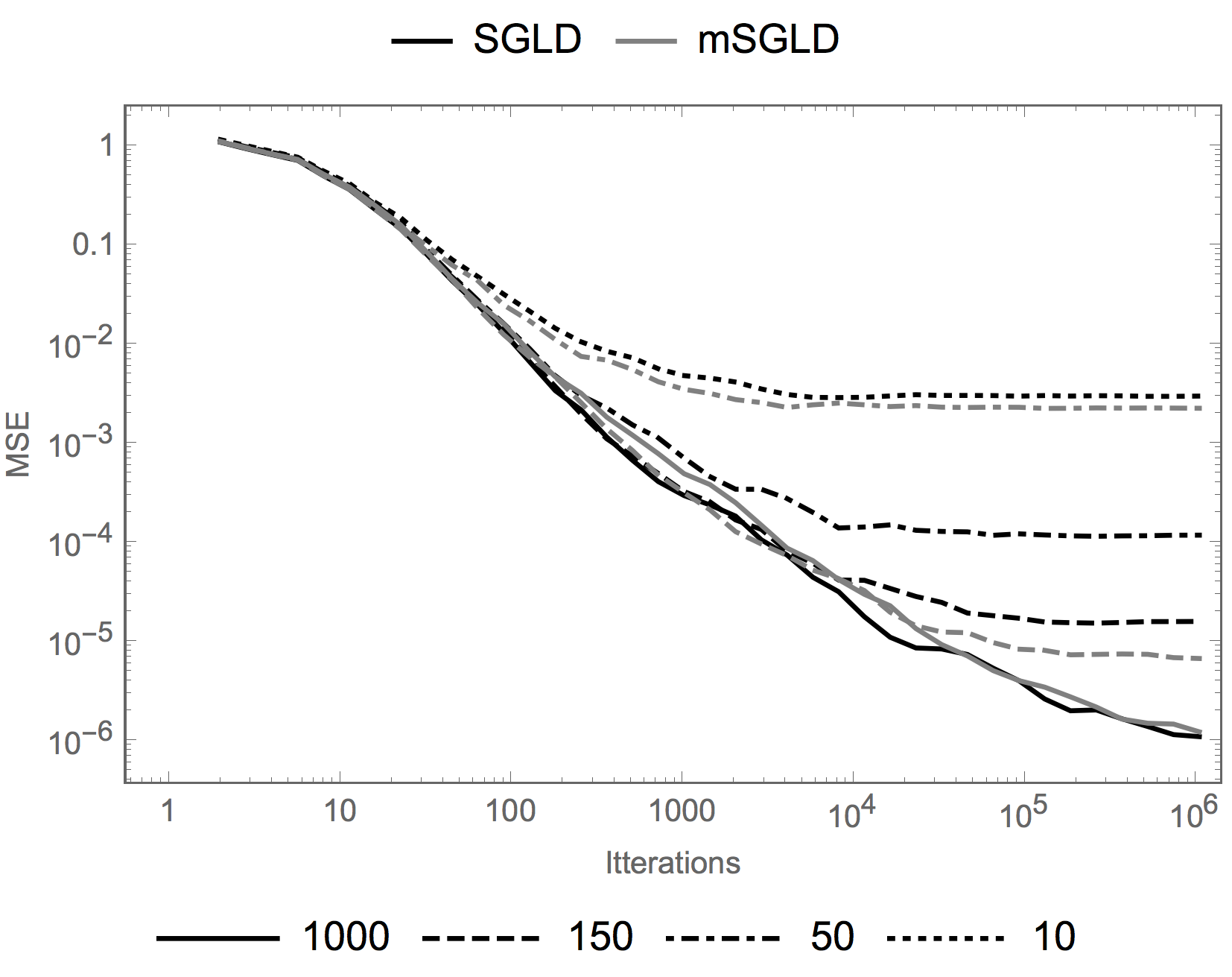}\hspace{0.02\textwidth}\includegraphics[width=0.47\textwidth]{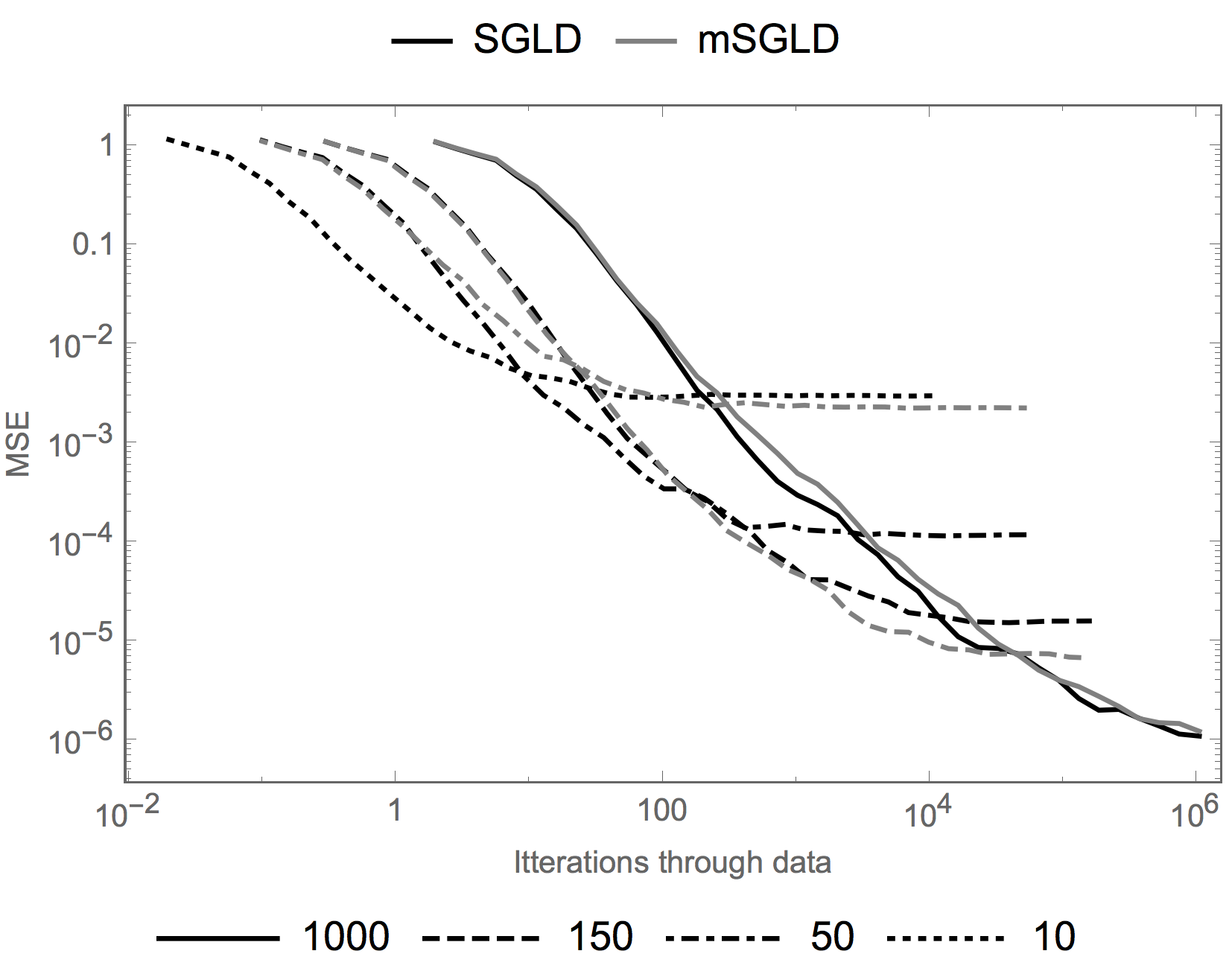}

\protect\protect\caption{\label{fig:logisticRegression}Expected MSE of time average for the
SGLD and the mSGLD for the mean of the posterior}
\end{figure}

\section{Conclusion}
\label{sec:7}
 This article presents the mathematical foundations that are necessary for posterior sampling for stochastic gradient methods with fixed step sizes. We derived an error expansion of the asymptotic bias in terms of powers of the step size and identified how the constant in the leading order term depends on the unbiased estimator of the gradient. We construct a modified SGLD to match  the Euler method in this asymptotic expansion. These asymptotic results are complemented by upper bounds on the bias and the MSE over a finite time horizon. Minimising the MSE with respect to the step size yields a decay of the error at the same rate as the decreasing step size SGLD, see Remark \ref{rem:cmpDecStep}. 
These theoretical findings are completed with extensive analytic investigations of a one dimensional toy model that allows the derivation of analytic expressions for the sample average and its moments. Finally, this yields an exact quantification of expected errors. The results of this investigation can be summarised as follows: \begin{itemize}
\item  In the high accuracy regime the SGLD deteriorates to the Euler method while the mSGLD prevails.
\item For small data batches the bias of the mSGLD is larger than for the SGLD.
\item  In the limit as the number of data items goes to infinity the SGLD reduces  computational complexity of estimating the second moment with vanishing MSE by one power of the number of data items. 
\end{itemize}
This recommends the construction of new and a study of existing modifications of the SGLD such as the Stochastic Gradient Hamiltonian Monte Carlo \cite{Chen2014} and the Stochastic Gradient Thermostat Monte Carlo \cite{HMCT14} algorithms.

\subsection*{Acknowledgement}
SJV and YWT  acknowledge EPSRC for research funding through grant  EP/K009850/1 and EP/K009362/1.
The authors thank Rémi Bardenet for fruitful discussions.

\bibliographystyle{apalike}
\bibliography{abd_biblio,complete1,stochasticGradients}

\global\long\def\theequation{A.\arabic{equation}}
 \global\long\def\theremark{A.\arabic{remark}}
 \global\long\def\thelemma{A.\arabic{lemma}}
  \setcounter{equation}{0}
\section{Appendix A: Finite time Ergodic Average based Poisson Equation}\label{sec:finitetimeproof}

\begin{proof}[Proof of Theorem \ref{thm:FiniteTimeBiasVariance}]
Rearranging Equation \eqref{eq:sampleAvg} for $\frac{1}{K}\sum_{k=0}^{K-1}(\phi_{k}-\bar{\phi})$, 
the bias and the MSE can be controlled as follows:

\[
\EE\frac{1}{Kh}\left(\psi_{K}-\psi_{0}\right)\lesssim\frac{1}{T}\sup_{i}\EE V_{i}^{p_{\poisol,0}}\lesssim\frac{1}{T}.
\]
Because $A_{0}=\mathcal{L}$ for the SGLD it is left to bound to bound
$\EE\frac{1}{T}S_{i,K}$ for $i=0\dots4$. First we consider $i=1$
and use Equation (\ref{eq:ExpecDeltaAndDrift})

\begin{eqnarray*}
\frac{1}{T}\mathbb{E}S_{1,K} & \lesssim & \frac{1}{T}\EE\sum_{k=0}^{K-1}h^{2}V_{k}^{p_{\poisol,2}}\EE_{\tau_{k}}\left\Vert \hat{f}_{\step}\right\Vert ^{2}\\
 & \lesssim & \frac{1}{T}\sum_{k=0}^{K-1}h^{2}\sup_{i}\EE V_{i}^{p_{\poisol,2}+1} \lesssim  \frac{1}{T}h^{2}K\lesssim h.
\end{eqnarray*}
This procedure will be used over and over again. It can be summarised
as follows: 
\begin{enumerate}
\item bounding the terms in the sum by a power of $V^{p}$, using Equation
(\ref{eq:ExpecDeltaAndDrift}) and the assumption on derivates of
$\psi$; 
\item then derive the bound using $\sup_{i}\EE V_{i}^{p}<\infty$. 
\end{enumerate}
For $i=2$ we additionally use that Equation (\ref{eq:lineV}) implies
\[
\int_{0}^{1}s^{3}\deriv{4}{\poisol}\left(s\theta_{k}+(1-s)\theta_{k+1}\right)\left(\Delta_{\step+1},\Delta_{\step+1},\Delta_{\step+1},\Delta_{\step+1}\right)ds\lesssim\left(V_{k}^{p_{\poisol,4}}+V_{k+1}^{p_{\poisol,4}}\right)\left\Vert \Delta_{\step}\right\Vert ^{4}.
\]
which allows us to follow the general procedure 
\begin{eqnarray*}
\frac{1}{T}\mathbb{E}S_{2,K} & \lesssim & \frac{1}{T}\EE\sum_{k=0}^{K-1}R_{k+1}\\
 & \lesssim & \frac{1}{T}\EE\sum_{k=0}^{K-1}h^{2}\left(V_{k}^{p_{\poisol,4}}+V_{k+1}^{p_{\poisol,4}}\right)\EE_{\tau}\left\Vert \Delta_{k+1}\right\Vert ^{4}\\
 & \lesssim & \frac{1}{T}\EE\sum_{k=0}^{K-1}h^{2}\sup_{i}\EE V^{p_{\poisol,4}+2}(\theta_{i})  \lesssim  \frac{1}{T}h^{2}\lesssim h.
\end{eqnarray*}
We apply the general procedure to $\mathbb{E}S_{0,K}$

\begin{eqnarray*}
\frac{1}{T}\mathbb{E}S_{0,K} & \lesssim & \frac{1}{T}\EE\sum_{k=0}^{K-1}h^{2}V_{k}^{p_{\poisol,3}}\EE_{\tau}\left\Vert g_{k}\eta_{\step+1}\right\Vert ^{2}\left\Vert \hat{f}_{k}\right\Vert \\
 & \lesssim & \frac{1}{T}\EE\sum_{k=0}^{K-1}h^{2}\sup_{i}\EE V^{p_{\poisol,3}+\frac{1}{2}}(\theta_{i})\\
 & \lesssim & \frac{1}{T}h^{2}\lesssim h.\\
\frac{1}{T}\mathbb{E}\tilde{S}_{0,K} & \lesssim & \frac{1}{T}\EE\sum_{k=0}^{K-1}h^{3}V_{k}^{p_{\poisol,3}}\EE_{\tau}\left\Vert \hat{f}_{k}\right\Vert ^{3}\\
 & \lesssim & \frac{1}{T}\sum_{k=0}^{K-1}h^{3}\sup_{i}\EE V^{p_{\poisol,3}+\frac{3}{2}}(\theta_{l}) \lesssim  \frac{1}{T}Kh^{3}\lesssim h^{2}.
\end{eqnarray*}
Thus, we have established the bound on the bias given by Equation
(\ref{eq:bias}).

In order to establish the bound one the MSE in Equation (\ref{eq:vairance})
we note that Equation \eqref{eq:sampleAvg} yields

\begin{eqnarray*}
\EE\left(\frac{1}{K}\sum_{k=0}^{K-1}\left(\phi-\bar{\phi}\right)\right)^{2} & \lesssim & \EE\frac{\left(\psi_{K}-\psi\right)^{2}}{T^{2}}\\
 &  & +\frac{1}{T^{2}}\sum_{i=0}^{2}\EE S_{i,K}^{2}+\frac{1}{T^{2}}\sum_{i=0}^{2}\EE M_{i,K}^{2}
\end{eqnarray*}
First we note that 
\[
\EE\frac{\left(\psi_{K}-\psi\right)^{2}}{T^{2}}\lesssim\frac{1}{T^{2}}\sup_{i}\EE V_{i}^{2p_{\poisol,0}}\lesssim\frac{1}{T^{2}}.
\]
The $S_{i,K}^{2}$ terms can be bound in similar way as above with
the additional use of Cauchy-Schwartz inequality to break the correlation
between $V_{i}^{p}$ and $V_{j}^{p}$.

\begin{eqnarray*}
\frac{1}{T^{2}}\mathbb{E}S_{1,K}^{2} & \lesssim & \frac{1}{T^{2}}\EE\sum_{i,j=0}^{K-1}h^{4}\hess \poisol_{i}\left[\hat{f}_{i},\hat{f}_{i}\right]\hess \poisol_{j}\left[\hat{f}_{j},\hat{f}_{j}\right]\\
 & \lesssim & \frac{1}{T^{2}}\EE\sum_{i,j=0}^{K-1}h^{4}\left\Vert \hess \poisol_{i}\right\Vert \EE_{\tau}\left\Vert \hat{f}_{i}\right\Vert ^{2}\left\Vert \hess \poisol_{j}\right\Vert \EE_{\tau}\left\Vert \hat{f}_{j}\right\Vert ^{2}\\
 & \lesssim & \frac{1}{T^{2}}\sum_{i,j=0}^{K-1}h^{4}\EE V_{i}^{p_{\poisol,2}+1}V_{j}^{p_{\poisol,2}+1}\\
 & \lesssim & \frac{1}{T^{2}}\sum_{i,j=0}^{K-1}h^{4}\left(\EE V_{i}^{2p_{\poisol,2}+2}\right)^{\frac{1}{2}}\left(\EE V_{j}^{2p_{\poisol,2}+2}\right)^{\frac{1}{2}}\\
 & \lesssim & \frac{1}{T^{2}}\sum_{i,j=0}^{K-1}h^{4}\left(\sup_{i}\EE V_{i}^{2p_{\poisol,2}+2}\right) \lesssim  \frac{K^{2}h^{4}}{T^{2}}\lesssim h^{2}
\end{eqnarray*}
Similarly, we bound

\begin{eqnarray*}
\frac{1}{T^{2}}\mathbb{E}S_{2,K}^{2} & \lesssim & \frac{1}{T^{2}}\EE\sum_{i,j=0}^{K-1}R_{i+1}R_{j+1}\\
 & \lesssim & \frac{1}{T^{2}}\EE\sum_{i,j=0}^{K-1}\left(V_{i}^{p_{\poisol,4}}+V_{i+1}^{p_{\poisol,4}}\right)\left(V_{j}^{p_{\poisol,4}}+V_{j+1}^{p_{\poisol,4}}\right)\EE_{\tau}\left\Vert \Delta_{i+1}\right\Vert ^{4}\EE_{\tau}\left\Vert \Delta_{j+1}\right\Vert ^{4}\\
 & \lesssim & \frac{h^{4}}{T^{2}}\EE\sum_{i,j=0}^{K-1}\left(V_{i}^{p_{\poisol,4}}+V_{i+1}^{p_{\poisol,4}}\right)\left(V_{j}^{p_{\poisol,4}}+V_{j+1}^{p_{\poisol,4}}\right)V_{i}^{2}V_{j}^{2}\\
 & \lesssim & \frac{h^{4}}{T^{2}}\EE\sum_{i,j=0}^{K-1}\left(V_{i}^{p_{\poisol,4}}+V_{i+1}^{p_{\poisol,4}}\right)\left(V_{j}^{p_{\poisol,4}}+V_{j+1}^{p_{\poisol,4}}\right)V_{i}^{2}V_{j}^{2}\\
 & \lesssim & \frac{h^{4}}{T^{2}}\sum_{i,j=0}^{K-1}\left(\sup_{i}\EE V_{i}^{2p_{\poisol,4}+4}\right)  \lesssim  \frac{K^{2}h^{4}}{T^{2}}\lesssim h^{2}.
\end{eqnarray*}

Similar bounds can also be obtained for $\frac{1}{T^{2}}\mathbb{E}S_{0,K}^{2}$
and $\frac{1}{T^{2}}\mathbb{E}\tilde{S}_{0,K}^{2}$ 
\begin{eqnarray*}
\frac{1}{T^{2}}\EE S_{0,K}^{2} & \lesssim & \frac{h^{4}}{T^{2}}\sum_{i,j=0}^{K-1}\EE V_{i}^{p_{\poisol,3}}\left\Vert \hat{f}_{i}\right\Vert V_{j}^{p_{\poisol,3}}\left\Vert \hat{f}_{j}\right\Vert \\
 & \lesssim & \frac{h^{4}K^{2}}{T^{2}}\sup_{i}\EE V_{i}^{2p_{\poisol,3}+1}\lesssim h^{2}\\
\frac{1}{T^{2}}\mathbb{E}\tilde{S}_{0,K}^{2} & \lesssim & \frac{h{}^{6}}{T^{2}}\sum_{i,j=0}^{K-1}\EE V_{i}^{p_{\poisol,3}}\left\Vert \hat{f}_{i}\right\Vert ^{3}V_{j}^{p_{\poisol,3}}\left\Vert \hat{f}_{j}\right\Vert ^{3}\\
 & \lesssim & \frac{h^{6}K^{2}}{T^{2}}\sup_{i}\EE V_{i}^{2p_{\poisol,3}+2}\lesssim h^{4}
\end{eqnarray*}
For Martingale terms the cross terms vanish which allows us to obtain
the following bounds 
\begin{eqnarray*}
\frac{1}{T^{2}}\EE M_{1,K}^{2} & = & \frac{1}{T^{2}}\sum_{i=0}^{K-1}\left(\EE D_{lm}^{2}\poisol_{i}\left(g_{i}^{l,a}\eta_{i+1}^{a}g_{i}^{m,b}\eta_{i+1}^{b}-g_{i}^{k,l}g_{i}^{k,m}\right)\right)^{2}\\
 & \lesssim & \frac{1}{T^{2}}\sum_{i=0}^{K-1}\sup_{i}\EE V_{i}^{2p_{\poisol,2}}.
\end{eqnarray*}
The following term is the crucial Martingale term as it yields the
$\mathcal{O}\left(\frac{1}{T}\right)$ contribution

\begin{eqnarray*}
\frac{1}{T^{2}}\EE M_{2,K}^{2} & \lesssim & \frac{1}{T^{2}}h\sum_{k=0}^{K-1}\EE\left\Vert \derivf \poisol_{\step}\right\Vert ^{2}\left\Vert g_{\step}\xi_{\step+1}\right\Vert ^{2}\\
 & \lesssim & \frac{1}{T^{2}}h\sum_{k=0}^{K-1}\EE V_{k}^{2p_{\poisol,1}}\lesssim\frac{1}{T}.
\end{eqnarray*}
Similarly, we estimate 
\begin{eqnarray*}
\frac{1}{T^{2}}\EE M_{3,K}^{2} & \lesssim & \frac{1}{T^{2}}h^{3}\sum_{k=0}^{K-1}\EE\left\Vert \hess\psi_{\step}\right\Vert ^{2}\left\Vert \hat{f}_{\step}\right\Vert ^{2}\left\Vert g_{\step}\xi_{\step+1}\right\Vert ^{2}\\
 & \lesssim & \frac{1}{T^{2}}h^{3}\sum_{k=0}^{K-1}\EE V_{k}^{2p_{\poisol,2}+1}\lesssim\frac{h^{2}}{T}.
\end{eqnarray*}
The terms $\frac{1}{T^{2}}\EE M_{0,K}^{2}$ and $\frac{1}{T^{2}}\EE\tilde{M}_{0,K}^{2}$
can be bounded in the same way 
\begin{eqnarray*}
\frac{1}{T^{2}}\EE M_{0,K}^{2} & \lesssim & \frac{h^{3}}{T^{2}}\sum_{k=0}^{K-1}\EE V_{k}^{2p_{\poisol,3}}\lesssim\frac{h^{3}K}{T^{2}}\leq\frac{h^{2}}{T}\\
\frac{1}{T^{2}}\EE\tilde{M}_{0,K}^{2} & \lesssim & \frac{h^{5}}{T^{2}}\sum_{k=0}^{K-1}\EE V_{k}^{2p_{\poisol,3}+2}\lesssim\frac{h^{5}K}{T^{2}}\leq\frac{h^{4}}{T}
\end{eqnarray*}
The additional part for the SGLD is the term corresponding to the
Martingale $M_{4,K}$ 
\begin{eqnarray*}
\frac{1}{T^{2}}M_{4,K}^{2} & \lesssim & \frac{1}{T^{2}}\EE h^{2}\sum_{k=0}^{K-1}\left(\derivf \poisol_{\step}(H_{k})\right)^{2}\\
 & \lesssim & \frac{1}{T^{2}}\EE h^{2}\sum_{k=0}^{K-1}V^{2p_{\poisol,1}}\EE_{\tau}\left\Vert H_{k}\right\Vert ^{2}\\
 & \lesssim & \frac{1}{T^{2}}h^{2}\sum_{k=0}^{K-1}\EE V^{2p_{\poisol,1}+1} \lesssim  \frac{h}{T}
\end{eqnarray*}
For all these calculations to go through need $\sup_{i}EV_{i}^{p^{\star}}$
to be bounded. Collecting the orders present, we see that 
\[
p^{\star}=\max\left\{ 2p_{\poisol,2}+2,2p_{\poisol,4}+4,2p_{\poisol,3}+1,2p_{\poisol,3}+3\right\} 
\]
is sufficient. \end{proof}

\subsection{Sufficient Conditions on $\pi$  Ensuring Finite Time Bounds on Bias and MSE  \label{sec:regPoisson}}
We formulate a sufficient condition on $\pi$ that ensure that Theorem \ref{thm:FiniteTimeBiasVariance} is applicable. This hinges on deriving a sufficient condition for Equations  \eqref{eq:PoissonDerivBounds},  \eqref{eq:LyapunovBddApriori}  and (\ref{eq:lineV}). The aim of this section is to establish and motivate the sufficient condition formulated in the following theorem.

\begin{theorem} Suppose the following condition holds
\begin{align}
\left\langle  \theta,\,\,\nabla\log\pi_0(\theta)\right\rangle\,\;& \leq\,-\alpha\, \left\Vert \theta \right \Vert^2+\beta\\
\left\langle  \theta,\,\,\nabla\log\pi(X_i \mid \theta)\right\rangle\,\;&\leq\,-\alpha\, \left\Vert \theta \right \Vert^2+\beta \text{ for }i=1,\dots,N.
\end{align}
then Theorem \ref{thm:FiniteTimeBiasVariance} is applicable for polynomially bounded and continuous $\phi$, that is
\begin{eqnarray*}
\text{Bias}\left(\hat{\phi}_{K}\right) & = & \left|\EE\hat{\phi}_{K}-\bar{\phi}\right| \leq C\left(h+\frac{1}{K h}\right)\label{eq:bias}\\
\EE\left(\hat{\phi}_{K}-\bar{\phi}\right)^{2} & \leq & C\left(h^{2}+\frac{1}{K h}\right)
\end{eqnarray*}
\end{theorem}

First we appeal to a sufficient condition for Equation\eqref{eq:LyapunovBddApriori} before summarising a regularity results of \cite{Veretennikov2001Poisson} which allows us to establish Equation \eqref{eq:PoissonDerivBounds}.
We believe that the sufficient conditions above  can be weakened, but this requires improving the results of \cite{Veretennikov2001Poisson} which is out of the scope of this article.


The condition $\sup_{k}\EE V^{p}\left(\theta_{k}\right)<\infty$ (that is Equation (\ref{eq:LyapunovBddApriori})) 
is established for all $p\leq p^\star$ by Lemma 5 in \cite{TehThierryVollmerSGLD2014} if $p^\star$ satisfies the following assumption. \begin{assumption} \label{assu:fromPaper1}The
drift term $\theta\mapsto\frac{1}{2}\,\nabla\log\pi(\theta)$ is continuous.
The function $V:\RR^{d}\to[1,\infty)$  tends
to infinity as $\|\theta\|\to\infty$, is twice differentiable
with bounded second derivatives and  satisfies the following
conditions: 
\begin{enumerate}
\item {V is a Lyapunov function for the Langevin dynamics, i.e. there are constants $\alpha,\beta>0$
such that for every $\theta\in\RR^{d}$ we have 
\begin{equation}
\,\left\langle\nabla V(\theta),\,\frac{1}{2}\,\nabla\log\pi(\theta)\right\rangle\,\;\leq\,-\alpha\, V(\theta)+\beta.\label{eq.lyapunov.drift}
\end{equation}
} 
\item {The following bounds hold

\begin{itemize}
\item {There exists an exponent $p_{H}\geq2$ such that 
\begin{equation}
\EE[\,\|H(\theta,\UU)\|^{2p_{H}}\,]\lesssim V^{p_{H}}(\theta).\label{eq.bound.H}
\end{equation}
Moreover, this implies that $\EE[\,\|H(\theta,\UU)\|^{2p}\,]\lesssim V^{p}(\theta)$
for any exponent $0\leq p\leq p_{H}$. } 
\item {For every $\theta\in\RR^{d}$ we have 
\begin{equation}
\|\nabla V(\theta)\|^{2}+\|\nabla\log\pi(\theta)\|^{2}\;\lesssim\; V(\theta).\label{eq.lyapunov.size}
\end{equation}
} 
\end{itemize}

}

\end{enumerate}
\end{assumption} Notice that for $\hat{f}$ based on subsampling
we obtain that $p_{H}=\infty$ if 
\[
\norm{\nabla\log p\left(y\mid\theta\right)}^{2}\leq C(y)V.
\]
Notice that Assumption \ref{assu:fromPaper1} also implies 
\begin{equation}
\begin{aligned}\EE\norm{\theta_{k+1}-\theta_{k}}^{2p} & \leq V_{k}^{p}\\
\EE\norm{\hat{f}_{k}}^{2p} & \leq V_{k}^{p}
\end{aligned}
\label{eq:ExpecDeltaAndDrift}
\end{equation}
if for $p\leq p_{H}$. Equation (\ref{eq:lineV}) could now simply be formulated as an additional
assumption, however currently we need even stronger assumptions to
verify Equation (\ref{eq:PoissonDerivBounds}). 
Subsequently, we show how the results of \cite{Veretennikov2001Poisson} can be used to establish  Equation (\ref{eq:PoissonDerivBounds}) if  Equations (\ref{eq.lyapunov.drift}) and
(\ref{eq.lyapunov.size}) hold for  $V=\norm{\theta}^{2}+1$.

Theorem 1 and 2 of \cite{Veretennikov2001Poisson} characterise the
smoothness and growth of the solution to the Poisson equation associated
with Equation (\ref{eq:SDE}). This is important for our results because the key ingredient for the proof of Theorem \ref{thm:FiniteTimeBiasVariance} is \[\sup_{k}\EE\norm{\deriv{(i)}{\psi}\left(\theta_{k}\right)}<\infty \text{ for } k=1,\dots,4.\]
Because we have already established in Section \ref{sec:4} that 

$\sup_{i}\EE V^{p}\left(\theta_{i}\right)  <\infty$
it is left to verify 
\begin{align}
\left\Vert \deriv{k}{\psi}\right\Vert  & \lesssim V^{p_{\poisol,k}},\quad\text{ for }k=0,\dots,4\tag{\ref{eq:PoissonDerivBounds}}
\end{align}
\marginpar{here}
The assumptions needed to apply the Theorem \ref{thm:FiniteTimeBiasVariance}
results are 
\begin{eqnarray}
\left\langle f(\theta),\frac{\theta}{\norm{\theta}}\right\rangle  & \leq & -r\norm{\theta},\quad\norm{\theta}\ge M_{0}\label{eq:assPardouxDrift}\\
0<\lambda_{-} & \leq & \left\langle g\left(\theta\right)g^{\star}\left(\theta\right)\frac{\theta}{\norm{\theta}},\frac{\theta}{\norm{\theta}}\right\rangle \leq\lambda_{+}<\infty.\label{eq:assPardouxVolat}
\end{eqnarray}
This holds if Assumption \ref{assu:fromPaper1} is satisfied with
$V\left(\theta\right)=\norm{\theta}^{2}+1$. \begin{theorem} \label{thm:pardouxPoisson}\cite{Veretennikov2001Poisson} Let
$\bar{f}=0$
and Equations (\ref{eq:assPardouxDrift}) and (\ref{eq:assPardouxVolat})
are satisfied. Then there exists a solution $\psi\in W_{\text{loc}}^{2}$
to the Poisson equation 
\[
\mathcal{L}\psi=\phi-\bar{\phi}
\]

\begin{enumerate}
\item If there is a C such that 
\[
\left|\phi(\theta)\right|\leq C\left(1+\norm{\theta}\right)^{\beta}
\]
for some $\beta<0$, then $u$ is bounded. Moreover, 
\[
\sup_{\theta}\left|\psi(\theta)\right|\leq C\sup_{\theta}\left|f\right|\left(1+\norm{\theta}\right)^{-\beta}
\]
and 
\[
\norm{\nabla\psi}\leq C
\]

\item if there exists a constant $C$ and some $\beta>0$ such that 
\[
\left|\phi(\theta)\right|\leq C(1+\norm{\theta})^{\beta}
\]
then there exist a constant such that 
\[
\left|\psi\left(\theta\right)\right|\leq C^{\prime}\left(1+\norm{\theta}\right)^{\beta}.
\]
Finally there exists $C$ such that 
\[
\norm{\nabla\psi}\leq C\left(1+\left|\theta\right|^{\beta}\right).
\]

\end{enumerate}
\end{theorem} \begin{remark} We believe that assumption Theorem
\ref{thm:pardouxPoisson} can be weakened to be of the form of Equation
\eqref{eq.lyapunov.drift} but it is out of the scope of this article
to explore this direction. \end{remark}

In order to iterate Theorem \ref{thm:pardouxPoisson} we note that
the derivatives $\psi$ can be expressed as solution to Poisson equations
with different RHSs. 
\begin{eqnarray}
\mathcal{L}\psi & = & \phi-\bar{\phi}\label{eq:PoiDeriv0}\\
A\partial_{i}\psi & = & \partial_{i}\phi-\frac{1}{2}\nabla\psi\cdot\partial_{i}f\label{eq:PoiDeriv1}\\
A\partial_{ij}\psi & = & \partial_{ij}\phi-\frac{1}{2}\nabla\partial_{j}\psi\cdot\partial_{i}f-\frac{1}{2}\nabla\psi\cdot\partial_{ij}f-\frac{1}{2}\nabla\psi\cdot\partial_{j}f\label{eq:PoiDeriv2}\\
A\partial_{ijk}\psi & = & \partial_{ijk}\phi-\frac{1}{2}\nabla\partial_{jk}\psi\cdot\partial_{i}f-\frac{1}{2}\nabla\partial_{j}\psi\cdot\partial_{ik}f-\frac{1}{2}\nabla\partial_{k}\psi\cdot\partial_{ij}f-\frac{1}{2}\nabla\psi\cdot\partial_{ijk}f\label{eq:PoiDeriv3}\\
 &  & -\frac{1}{2}\nabla\partial_{k}\psi\cdot\partial_{jk}f-\frac{1}{2}\nabla\psi\cdot\partial_{jk}f-\frac{1}{2}\nabla\psi\cdot\partial_{k}f\nonumber 
\end{eqnarray}
We will denote by $\beta_{\psi,i}$ numbers that satisfy 
\begin{equation}
\sup_{\left|\alpha\right|=i}\norm{\partial^{\alpha}\psi}\lesssim\left(1+\left|\theta\right|^{\beta_{\psi,i}}\right)\label{eq:betanotation}
\end{equation}
where we used multi-index notation for derivatives. We use a similar
notation for the derivatives of $f$, that is $\beta_{f,i}$ and assume
that these bounds are a priori given.

Using Theorem \ref{thm:pardouxPoisson} we can obtain $p_{\psi,i}$
to satisfy Equation \eqref{eq:PoissonDerivBounds} in terms of the
$\beta$'s which we formulate as the following lemma.

\begin{lemma} Suppose that $\phi$ and its derivatives are bounded
and Assumption \ref{assu:fromPaper1} and Equations Equations (\ref{eq:assPardouxDrift})
and (\ref{eq:assPardouxVolat}) hold. Then the choice 
\begin{align*}
p_{\psi,0} & = & 0\\
p_{\psi,1} & = & 0\\
p_{\psi,2} & = & \frac{\beta_{f,1}}{2}\\
p_{\psi,3} & = & \beta_{f,1}\vee\frac{\beta_{f,2}}{2}\\
p_{\psi,4} & = & \frac{1}{2}\left(3\beta_{f,1}\vee\left(\beta_{f,1}+\beta_{f,2}.\right)\vee\beta_{f,3}\right)
\end{align*}
satisfies Equation \eqref{eq:PoissonDerivBounds}.

\end{lemma} \begin{proof} Assumption \ref{assu:fromPaper1} yields
that $\beta_{f,0}=1$ is a valid choice. Applying Theorem \ref{thm:pardouxPoisson}
to Equation (\ref{eq:PoiDeriv0}) implies that $\beta_{\psi,0}:=\beta_{\psi,1}:=0$
satisfies Equation (\ref{eq:betanotation}). Applying Theorem \ref{thm:pardouxPoisson}
to Equation (\ref{eq:PoiDeriv1}) yields that 
\[
\beta_{\psi,2}\leq\beta_{f,1}.
\]

Applying Theorem \ref{thm:pardouxPoisson} to Equation (\ref{eq:PoiDeriv2})
yields that 
\[
\beta_{\psi,3}\leq2\beta_{f,1}\vee\beta_{f,2}.
\]
Applying Theorem \ref{thm:pardouxPoisson} to Equation (\ref{eq:PoiDeriv3})
yields that 
\begin{eqnarray*}
\beta_{\psi,4} & \leq & \left(\beta_{f,1}+\left(2\beta_{f,1}\vee\beta_{f,2}.\right)\right)\vee\left(\beta_{f,1}+\beta_{f,2}.\right)\vee\beta_{f,3}\\
 & \leq & 3\beta_{f,1}\vee\left(\beta_{f,1}+\beta_{f,2}.\right)\vee\beta_{f,3}.
\end{eqnarray*}
Thus we have established Equation (\ref{eq:PoissonDerivBounds}).
\end{proof}

\global\long\def\theequation{B.\arabic{equation}}
 \global\long\def\theremark{B.\arabic{remark}}
 \global\long\def\thelemma{B.\arabic{lemma}}
  \setcounter{equation}{0}

\global\long\def\theequation{C.\arabic{equation}}
 \global\long\def\theremark{C.\arabic{remark}}
 \global\long\def\thelemma{C.\arabic{lemma}}
  \setcounter{equation}{0}  

\section{Analytic expressions for the Gaussian Toy Model \label{app:Toy}}
We sketch the derivation of the analytic expression of the MSE are used for the plots in Section \ref{sec:anal_toy}.

\subsection{Expected MSE for fixed Data Sets\label{sec:appToyFixedModels}}

We outline how an analytic expression for the MSE of the
sample average can be derived. The following method generalises to
any polynomial test function but we concentrate on the sample average
for the second moment of the posterior given by 
$
S_{2}=\frac{1}{M}\sum_{j=0}^{M-1}\theta_{j}^{2}.
$
Its MSE can be expressed using Equation \eqref{eq:SGaussianPosterior}
\begin{equation}
MSE=\IE\left(\frac{1}{M}\sum_{k=0}^{M-1}\theta_{j}^{2}-(\mu_{p}^{2}+\sigma_{p}^{2})\right)^{2}=\IE S_{2}^{2}-2\IE S_{2}\left(\mu_{p}^{2}+\sigma_{p}^{2}\right)+\left(\mu_{p}^{2}+\sigma_{p}^{2}\right)^{2}.\label{eq:ToyMSE}
\end{equation}

In order to express $\IE S_{2}^{2}$ in Equation \eqref{eq:ToyMSE}
we derive the recurrence equations for $\IE\theta_{j}^{i}$ for $i=1,\dots,4$
by taking the expectations of 
\begin{align}
\theta_{j+1} & =  (1-A\, h)\theta_{j}+hB_{j}+\sqrt{h}\eta_{j}\nonumber \\
\theta_{j+1}^{2} & =  (1-A\, h)^{2}\theta_{j}^{2}+h^{2}B_{j}^{2}+h\eta_{j}^{2}+2(1-A\, h)\theta_{j}hB_{j}+2(1-A\, h)\theta_{j}\sqrt{h}\eta_{j}+2hB_{j}\sqrt{h}\eta_{j}\label{eq:RecursionTh2}\\
\theta_{j+1}^{3} & =  \left((1-A\, h)\theta_{j}+hB_{j}+\sqrt{h}\eta_{j}\right)^{3}\nonumber \\
 & =  (1-A\, h)^{3}\theta_{j}^{3}+3(1-A\, h)^{2}\theta_{j}^{2}hB_{j}+3(1-A\, h)\theta_{j}h^{2}B_{j}^{2}+h^{3}B_{j}^{3}\nonumber \\
 &   +3\sqrt{h}\eta\left(\dots\right)+3\eta^{2}h\left((1-A\, h)\theta_{j}+hB_{j}\right)+\eta^{3}h^{\frac{3}{2}}\nonumber \\
\theta_{j+1}^{4} & =  (1-A\, h)^{4}\theta_{j}^{4}+4(1-A\, h)^{3}\theta_{j}^{3}hB_{j}+6(1-A\, h)^{2}\theta_{j}^{2}h^{2}B_{j}^{2}\\
 &   +4(1-A\, h)\theta_{j}h^{3}B_{j}^{3}+h^{4}B_{j}^{4}+4\eta\left(\dots\right)+4\eta^{3}\left(\dots\right)\\
 &   +6h\eta^{2}\left((1-A\, h)^{2}\theta_{j}^{2}+2(1-A\, h)\theta_{j}hB_{j}+h^{2}B_{j}^{2}\right)+\eta^{4}h^{2}.
\end{align}
The recurrent equation for $\IE\theta_{j}$ is linear and first order
and can therefore be solved explicitly. Plugging the result into the
equation for $\IE\theta_{j}^{2}$ turns it into a first order linear
equation as well. Repeating this processes yields explicit expressions
for $\IE\theta_{j}^{i}$ $i=1,\dots,4$.The sums can be carried out
explicitly because the terms are of the form of a geometric sum or
a geometric term with a polynomial factor. This allows us to obtain
an analytic expression for $\IE S_{2}^{2}$ by reducing it to $\IE\theta_{j}^{i}$
as follows 
\begin{equation}
\IE S_{2}^{2}=\frac{1}{M^{2}}\IE\left(\sum_{i=0}^{M-1}\theta_{i}^{4}+2\sum_{i=0}^{M-1}\theta_{i}^{2}\sum_{j=i+1}^{M-1}\theta_{j}^{2}\right).\label{eq:toyMSEergAvgSq}
\end{equation}
The cross terms can be removed using Equation (\ref{eq:RecursionTh2})
so that 
\begin{eqnarray}
\theta_{j}^{2} & = & \left(1-Ah\right)^{2(j-i)}\theta_{i}^{2}+\sum_{k=0}^{j-1-i}\left(1-Ah\right)^{2k}\label{eq:toyCrossTerms}\\
 &  & \left[h^{2}B_{j-1-k}^{2}+h^{2}\eta_{j-1-k}^{2}+2\left(1-Ah\right)B_{j-1-k}h\theta_{j-1-k}+\right.\nonumber \\
 &  & \left.+2\left(1-Ah\right)\eta_{j-1-k}\theta_{j-1-k}\sqrt{h}+2h^{\frac{3}{2}}B_{j-1-k}\eta_{j-1-k}\right].\nonumber 
\end{eqnarray}
Plugging this into Equation \eqref{eq:toyMSEergAvgSq} yields

\begin{eqnarray*}
\IE S_{2}^{2} & = & \IE\frac{1}{M^{2}}\sum_{i=0}^{M-1}\theta_{i}^{4}\left(1+2\sum_{j=i+1}^{M-1}\left(1-Ah\right)^{2(j-i)}\right)\\
 &  & +\IE\frac{1}{M^{2}}\sum_{i=0}^{M-1}\theta_{i}^{2}\sum_{j=i+1}^{M-1}\sum_{k=0}^{j-1-i}\left(1-Ah\right)^{2k}\\
 &  & \quad\left[h^{2}B_{j-1-k}^{2}+h\eta_{j-1-k}^{2}+2\left(1-Ah\right)B_{j-1-k}h\theta_{j-1-k}\right.\\
 &  & \quad\left.+2\left(1-Ah\right)\eta_{j-1-k}\sqrt{h}+2h^{\frac{3}{2}}B_{j-1-k}\eta_{j-1-k}\right].
\end{eqnarray*}
Using the recurrence Equation to express $\theta_{j-1-k}$ in terms
of $\theta_{i}$ we conclude that $\IE S_{2}^{2}$ is equal to 
\begin{align*}
 & \frac{1}{M^{2}}\sum_{i=0}^{M-1}\IE\theta_{i}^{4}\left(1+2\sum_{j=i+1}^{M-1}\left(1-Ah\right)^{j-i}\right)\\
 & +\frac{1}{M^{2}}\sum_{i=0}^{M-1}\IE\theta_{i}^{2}\sum_{j=i+1}^{M-1}\sum_{k=0}^{j-1-i}\left(1-Ah\right)^{2k}\left[h^{2}\IE B^{2}+h\right].\\
 & +\IE\frac{1}{M^{2}}\sum_{i=0}^{M-1}\theta_{i}^{2}\sum_{j=i+1}^{M-1}\sum_{k=0}^{j-1-i}\left(1-Ah\right)^{2k}2\left(1-Ah\right)B_{j-1-k}h\\
 & \left((1-Ah)^{(j-1-k)-i}\theta_{i}+\sum_{l=0}^{(j-1-k)-i-1}(1-Ah)^{l}\left(hB_{j-1-k-l-1}+\sqrt{h}\eta_{j-1-k-l-1}\right)\right)\\
\end{align*}
Taking the expectations into the sum yields 
\begin{eqnarray*}
\IE S_{2}^{2} & = & \frac{1}{M^{2}}\sum_{i=0}^{M-1}\IE\theta_{i}^{4}\left(1+2\sum_{j=i+1}^{M-1}\left(1-Ah\right)^{j-i}\right)\\
 &  & +\frac{1}{M^{2}}\sum_{i=0}^{M-1}\IE\theta_{i}^{2}\sum_{j=i+1}^{M-1}\sum_{k=0}^{j-1-i}\left(1-Ah\right)^{2k}\left[h^{2}\IE B^{2}+h\right].\\
 &  & +\IE\frac{1}{M^{2}}\sum_{i=0}^{M-1}\IE\theta_{i}^{3}\sum_{j=i+1}^{M-1}\sum_{k=0}^{j-1-i}\left(1-Ah\right)^{2k}2\left(1-Ah\right)\IE Bh(1-Ah)^{(j-1-k)-i}\\
 &  & +\frac{1}{M^{2}}\sum_{i=0}^{M-1}\IE\theta_{i}^{2}\sum_{j=i+1}^{M-1}\sum_{k=0}^{j-1-i}\left(1-Ah\right)^{2k}2\left(1-Ah\right)\IE Bh\sum_{l=0}^{(j-1-k)-i-1}(1-Ah)^{l}h\IE B.
\end{eqnarray*}
We have an expressions for $\IE B$ but in the following we derive
the expressions for $\IE B^{2}$ required to express $\IE S_{2}^{2}$.
The terms $\IE B^{3}$ and $\IE B^{4}$ are needed for the derivation
of $\IE\theta_{j}^{4}$. In order to derive expressions for $\IE B^{p}$,
we introduce the power sums 
\[
p_{k}=\sum_{i=1}^{N}X_{i}^{k}
\]
and the elementary symmetric polynomials 
\begin{equation}
e_{0}  =1, e_{1}  =\sum_{i=1}^N X_{i}, 
e_{2}  {\textstyle =\sum_{1\leq i<j\leq N}X_{i}X_{j},},
\dots, e_{N}  =\prod_{i=1}^{N}X_{i}.
\end{equation}
Computing $e_{i}$ naively has complexity of order $\mathcal{O}\left(N^{i}\right)$
for $i\ll N$. Using Newton's identities
\begin{equation*}
e_{1}  =p_{1}, e_{2}  =\frac{1}{2}\left(e_{1}p_{1}-p_{2}\right), e_{3}  =\frac{1}{3}\left(-e_{1}p_{2}+e_{2}p_{1}+p_{3}\right), e_{4}  =\frac{1}{4}\left(e_{1}p_{3}-e_{2}p_{2}+e_{3}p_{1}-p_{4}\right)
\end{equation*}
$e_{i}$ can be expressed in terms of $p_{k}$, $k\leq i$ each of
which can be computed with complexity of order $\mathcal{O}\left(N\right)$. 

We consider the term $B=\frac{N}{n}\frac{\sum_{i=1}^{n}X_{\tau_{i}}}{2\sigma_{x}^{2}}$
where $\tau_{i}$ are sampled with replacement from a fixed data set
$\left\{ 1,\dots,N\right\} $. The second moment can be calculated
as follows 
\begin{eqnarray*}
\IE B^{2} & = & \left(\frac{N}{n2\sigma_{x}^{2}}\right)^{2}\sum_{i,j}\IE X_{\tau_{i}}X_{\tau_{j}}\\
 & = & \left(\frac{N}{n2\sigma_{x}^{2}}\right)^{2}\left(n(n-1)\underbrace{\IE X_{\tau_{1}}X_{\tau_{2}}}_{\text{Mom}_{2,1}}+n\underbrace{\IE X_{\tau_{1}}X_{\tau_{1}}}_{\text{Mom}_{2,2}}\right).
\end{eqnarray*}
We use Newton's identities to express $\text{Mom}_{2,1}$ and $\text{Mom}_{2,2}$
\begin{equation*}
\text{Mom}_{2,1}  =  \frac{2e_{2}}{N(N-1)}\quad \text{Mom}_{2,2}  =  \frac{p_{2}}{N}.
\end{equation*}
Similarly, we obtain
\begin{eqnarray*}
\IE B^{3} & = & \left(\frac{N}{n2\sigma_{x}^{2}}\right)^{3}\sum_{i,j,k}\IE X_{\tau_{i}}X_{\tau_{j}}X_{\tau_{k}}\\
 & = & \left(\frac{N}{n2\sigma_{x}^{2}}\right)^{3}\left(n(n-1)(n-2)\underbrace{\IE X_{\tau_{1}}X_{\tau_{2}}X_{\tau_{3}}}_{\text{Mom}_{3,1}}+3n(n-1)\underbrace{\IE X_{\tau_{1}}X_{\tau_{2}}^{2}}_{\text{Mom}_{3,2}}+n\underbrace{\IE X_{\tau_{1}}^{3}}_{\text{Mom}_{3,3}}\right)\\
\text{Mom}_{3,1} & = & \frac{\sum_{i\neq j\neq l}X_{i}X_{j}X_{l}}{N(N-1)(N-2)}=\frac{6e_{3}}{N(N-1)(N-2)}\\
\text{Mom}_{3,2} & = & \frac{p_{1}p_{2}-p_{3}}{N(N-1)},\quad \text{Mom}_{3,3}  =  \frac{p_{3}}{N}.
\end{eqnarray*}
A similar calculation yields a representation of $\IE B^{4}$ in terms of $p_1,\dots,p_4.$

\subsection{Expected MSE for Random Data\label{sub:appToyAnalyticNtoInfty}}

We sketch the derivation of the  MSE
\[
\mathbb{E}_{\theta,X}\left(\frac{1}{M}\sum_{k=0}^{M-1}\theta_{j}^{2}-(\mu_{p}^{2}+\sigma_{p}^{2})\right)^{2}
\]
where we take expectation with respect to $X_{i}\overset{\text{i.i.d.}}{\sim}\mathcal{N}\left(\theta^{\dagger},\sigma_{X}^{2}\right)$
for $i=1,\dots,N$ and the randomness in the recursion for $\theta_{j}$.
We obtain an analytic expression for the MSE by deriving expressions
for $\mathbb{E}_{X,\theta}\theta_{j}^{p}$. We illustrate the computation
for $p=2$, noting that we assume $\theta_{0}=0\text{ a.s.}$. We
know that 
\begin{eqnarray}
\theta_{j}^{2} & = & \sum_{k=0}^{j-1}\left(1-Ah\right)^{2k}\left[h^{2}B_{j-1-k}^{2}+h\eta_{j-1-k}^{2}+2\left(1-Ah\right)\eta_{j-1-k}\theta_{j-1-k}\sqrt{h}\right.\label{eq:toyThetaSq}\\
 &  & \left.+2h^{\frac{3}{2}}B_{j-1-k}\eta_{j-1-k}+2\left(1-Ah\right)B_{j-1-k}h\theta_{j-1-k}\right]\nonumber \\
 & = & \sum_{k=0}^{j-1}\left(1-Ah\right)^{2k}\left[h^{2}B_{j-1-k}^{2}+2\left(1-Ah\right)\eta_{j-1-k}\theta_{j-1-k}\sqrt{h}+2h^{\frac{3}{2}}B_{j-1-k}\eta_{j-1-k}\right.\\
 &  & \left.h\eta_{j-1-k}^{2}+2\left(1-Ah\right)B_{j-1-k}h\sum_{l=0}^{j-1-k-1}(1-Ah)^{l}\left(hB_{j-k-2-l}+\sqrt{h}\eta_{j-k-2-l}\right)\right].
\end{eqnarray}
The expectation $\mathbb{E}_{X,\theta}\,\theta_{j}$ therefore boils
down to calculating $\IE B^{2}$, $\IE B\, B^{\prime}$ and $\IE B$
where $B$ and $B^{\prime}$ are independent samples of Equation \eqref{eq:toyBdistribution}.
We start by calculating

\begin{eqnarray*}
\IE BB^{\prime} & = & \frac{N^{2}}{n^{2}4\sigma_{x}^{4}}\sum_{i=1}^{n}\sum_{j=1}^{n}\IE X_{\tau_{i}}X_{\tilde{\tau}_{j}}\\
 & = & \frac{N^{2}}{n^{2}4\sigma_{x}^{4}}\sum_{i=1}^{n}\sum_{j=1}^{n}\left(\frac{1}{N}\IE X^{2}+\frac{N-1}{N}\IE X\tilde{X}\right)\\
 & = & \frac{N^{2}}{n^{2}4\sigma_{x}^{4}}\mbox{\ensuremath{\left(\sum_{i=1}^{n}\sum_{j=1}^{n}\frac{1}{N}(\theta^{\dagger}{}^{2}+\sigma_{X}^{2})+\frac{N-1}{N}\theta^{\dagger}{}^{2}\right)}}.
\end{eqnarray*}
Similarly, we obtain 
\[
\IE B^{2}=\frac{N^{2}}{n^{2}4\sigma_{x}^{4}}\frac{n(\theta^{\dagger2}+\sigma_{\theta}^{2})+n(n-1)\theta^{\dagger2}}{1}.
\]
Deriving $\mathbb{E}\theta_{j}^{p}$ for $p=1,2,3,4$ requires the
calculation of $\mathbb{E}B_{1}^{\alpha_{1}}B_{2}^{\alpha_{2}}B_{3}^{\alpha_{3}}B_{4}^{\alpha_{4}}$
where $B_{i}$ are i.i.d. following the distribution of Equation \eqref{eq:toyBdistribution}
for $\alpha_{i}\ge0$ and $\sum_{i}\alpha_{i}\leq4$. The arguments
so far allow us to derive $T_{1}$ and $T_{\text{3}}$ in 
\begin{eqnarray*}
MSE & = & \IE\left(S_{2}^{2}-2S_{2}(\mu_{p}^{2}+\sigma_{p}^{2})+(\mu_{p}^{2}+\sigma_{p}^{2})^{2}\right)\\
 & = & \underbrace{\IE S_{2}^{2}}_{T_{1}}-2\underbrace{\IE S_{2}\mu_{p}^{2}}_{T_{2}}-2\underbrace{\IE S_{2}\sigma_{p}^{2}}_{T_{3}}+\underbrace{\IE\left(\mu_{p}^{4}+2\mu_{p}^{2}\sigma_{p}^{2}+\sigma_{p}^{4}\right)}_{T_{4}}.
\end{eqnarray*}
Recall that the posterior for this toy model is given by 
\[
\mathcal{N}(\mu_{p},\sigma_{p}^{2})=\mathcal{N}\left(\frac{\sum_{i=1}^{N}X_{i}}{\frac{\sigma_{x}^{2}}{\sigma_{\theta}^{2}}+N},\left(\frac{1}{\sigma_{\theta}^{2}}+\frac{N}{\sigma_{x}^{2}}\right)^{-1}\right)
\]
and hence $T_{4}$ can be computed explicitly. The summands of $T_{2}$
can be derived similarly to Equation \eqref{eq:toyThetaSq} in terms
of the quantities $\mathbb{E}B\mu_{p}^{2}$, $\mathbb{E}B^{2}\mu_{p}^{2}$
and $\mathbb{E}BB^{\prime}\mu_{p}^{2}.$ The explicit expression can
be obtained from the supplemented Mathematica\textsuperscript{\textregistered} file.

\end{document}